\documentclass[
reqno]{amsart} 
\usepackage[T1]{fontenc}
\usepackage[english]{babel}
\usepackage[linktocpage=true,colorlinks=true, linkcolor=blue, citecolor=red, urlcolor=green]{hyperref}

\usepackage{amssymb}
\usepackage{bm,bbm,enumerate}
\usepackage{color}
\usepackage{xcolor}

\usepackage{float}
\restylefloat{table}

\newcommand{\mc}[1]{\mathcal{#1}}
\newcommand{\mf}[1]{\mathfrak{#1}}
\newcommand{\mb}[1]{\mathbb{#1}}
\newcommand{\id}{\mathbbm{1}}
\newcommand{\tint}{{\textstyle\int}}
\newcommand{\ul}{\underline}

\DeclareMathOperator{\Mat}{Mat}
\DeclareMathOperator{\ad}{ad}
\DeclareMathOperator{\tr}{tr}
\DeclareMathOperator{\Res}{Res}
\DeclareMathOperator{\End}{End}
\DeclareMathOperator{\Hom}{Hom}
\DeclareMathOperator{\Span}{Span}
\DeclareMathOperator{\ord}{ord}

\theoremstyle{plain}
\newtheorem{theorem}{Theorem}[section]
\newtheorem{lemma}[theorem]{Lemma}
\newtheorem{proposition}[theorem]{Proposition}
\newtheorem{corollary}[theorem]{Corollary}

\theoremstyle{definition}

\theoremstyle{remark}
\newtheorem{remark}[theorem]{Remark}
\newtheorem{example}[theorem]{Example}

\numberwithin{equation}{section}

\definecolor{light}{gray}{.9}

\begin{document}

\title{$\ul p$-reduced multicomponent KP hierarchy and classical $\mc W$-algebras $\mc W(\mf{gl}_N,\ul p)$}

\author{Sylvain Carpentier}
\address{Mathematics Department, Columbia University, USA}
\email{sc4278@math.columbia.edu}
\urladdr{https://www.simonsfoundation.org/team/sylvain-carpentier/}
\author{Alberto De Sole}
\address{Dipartimento di Matematica, Sapienza Universit\`a di Roma,
P.le Aldo Moro 2, 00185 Rome, Italy}
\email{desole@mat.uniroma1.it}
\urladdr{http://www1.mat.uniroma1.it/~desole}
\author{Victor G. Kac}
\address{Department of Mathematics, MIT,
77 Massachusetts Ave., Cambridge, MA 02139, USA}
\email{kac@math.mit.edu}
\urladdr{http://www-math.mit.edu/~kac/}
\author{Daniele Valeri}
\address{School of Mathematics and Statistics, University of Glasgow,
G12 8QQ Glasgow, UK}
\email{daniele.valeri@glasgow.ac.uk}
\urladdr{https://www.maths.gla.ac.uk/~dvaleri/}
\author{Johan van de Leur}
\address{Mathematical Institute, Utrecht University, 3508 TA Utrecht, The Netherlands}
\email{j.w.vandeleur@uu.nl}
\urladdr{http://www.staff.science.uu.nl/~leur0102/}



\begin{abstract}
For each partition $\ul p$ of an integer $N\geq2$, consisting of $r$ parts,
an integrable hierarchy of Lax type Hamiltonian PDE has been constructed recently
by some of us.
In the present paper we show that any tau-function of the $\ul p$-reduced $r$-component KP hierarchy
produces a solution of this integrable hierarchy.
Along the way we provide an algorithm for the explicit construction of the generators
of the corresponding classical $\mc W$-algebra $\mc W(\mf{gl}_N,\ul p)$, and write down explicit formulas for evolution
of these generators along the commuting Hamiltonian flows.
\end{abstract}


\maketitle


\pagestyle{plain}

\section{Introduction}
\label{S1}

Let $N\geq2$ be an integer and let $\ul p$ be a partition of $N$ in $r$ parts.
Let $f$ be the nilpotent element of $\mf{gl}_N$
in Jordan form corresponding to the partition $\ul p$.
As a special case of a general construction for a reductive Lie algebra $\mf g$
and its nilpotent element $f$,
we have the corresponding Poisson vertex algebra $\mc W(\mf{gl}_N,\ul p)$,
called the classical affine $\mc W$-algebra, see e.g. \cite{DSKV13}.
In the paper \cite{DSKV16b}
for all these classical affine $\mc W$-algebras
an integrable hierarchy of Hamiltonian PDE was constructed.
This construction was extended to all classical Lie algebras $\mf g$
and all their nilpotent elements in \cite{DSKV18}.

In the case of an arbitrary reductive Lie algebra $\mf g$
and its principal nilpotent element $f$
the classical affine $\mc W$-algebra,
or rather the corresponding algebra of local Poisson brackets,
was constructed long ago in the seminal paper of Drinfeld and Sokolov \cite{DS85},
where the associated integrable hierarchy of Hamiltonian PDE was constructed as well.
It was also shown there that in the case $\mf g=\mf{gl}_N$
one gets along these lines the Gelfand-Dickey $N$-KdV hierarchy
of Lax equations, constructed in \cite{GD76},
using the method of fractional powers of differential operators.
The case $N=2$ is the classical KdV hierarchy.

The principal nilpotent element in $\mf{gl}_N$
corresponds to the partition of $N$ in $r=1$ parts.
It was shown in \cite{DJKM82}
that the $N$-KdV hierarchy of Gelfand-Dickey
is obtained by a simple reduction of the ($r=1$ component)
KP hierarchy introduced in \cite{Sato81}.
The key discovery of the Kyoto school was the notion of the tau-function,
which encodes a solution of the KP hierarchy
(and has a beautiful geometric meaning as a point in an infinite-dimensional Grassmann manifold).
A tau-function $\tau$ of the KP hierarchy is a function in infinitely many variables $t_1,t_2,\dots$,
and its reduction to the $N$-KdV hierarchy is given by the simple constraint
\begin{equation}\label{eq:intro1}
\frac{\partial\tau}{\partial t_N}=\text{const.}\,\tau
\,.
\end{equation}
In \cite{Sato81}, \cite{DJKM82} and subsequent papers,
polynomial, soliton type, and theta-function type tau-functions
were constructed,
which led to the construction of solutions of the $N$-KdV hierarchy.

The $r$-component KP hierarchy
was introduced in \cite{Sato81} and \cite{DJKM81},
and its theory was further developed in subsequent works,
including \cite{KvdL03} and \cite{Dic03}.
A tau-function of this hierarchy is a collection of functions $\vec\tau$ in the variables
$t^{(a)}_j$, where $j\in\mb Z_{\geq1}$ and $a=1,2,\dots,r$.
In the paper \cite{KvdL03}
the following reduction of the tau-functions of the $r$-component KP hierarchy
was attached to any partition $\ul p=(p_1\geq\dots\geq p_r>0)$,
consisting of $r$ parts:
\begin{equation}\label{eq:intro2}
\sum_{a=1}^r
\frac{\partial\vec\tau}{\partial t^{(a)}_{p_a}}=\text{const.}\,\vec\tau
\,.
\end{equation}
These are tau-functions of a hierarchy of evolution PDE,
called the $\ul p$-reduced $r$-component KP hierarchy (or  the $\ul p$-KdV hierarchy).
Recently, in \cite{KvdL19}
all polynomial tau-functions of this hierarchy were constructed,
see Theorem \ref{mmT} below.
The goal of the present paper is to link the $\ul p$-reduced $r$-component
KP hierarchy to the hierarchy of Hamiltonian PDE,
attached to the $\mc W$-algebra $\mc W(\mf{gl}_N,\ul p)$.
Namely,
starting with a tau-function of the $r$-component KP hierarchy
satisfying the constraint \eqref{eq:intro2},
we construct a solution of the hierarchy of Hamiltonian equations \eqref{eq:hameq}
attached to the $\mc W$-algebra $\mc W(\mf{gl}_N,\ul p)$,
see Theorem \ref{thm:main}, which is the first main result of this paper.

The first step in the proof of Theorem \ref{thm:main}
is the construction of the Lax operator corresponding to the tau-function satisfying \eqref{eq:intro2},
which is given by formula \eqref{eq:Lsol}. 
Proposition \ref{thm:lax-sol} states that this Lax operator satisfies 
the hierarchy of Lax equations \eqref{eq:527}.
On the other hand,
the Lax operator constructed in \cite{DSKV16b}
satisfies equation \eqref{eq:527} as well,
and it turns out that it is a generic pseudodifferential operator		
of the same shape as the one given by \eqref{eq:Lsol}.		
This implies that any solution of the $\ul p$-reduced $r$-component KP hierarchy 		
is a solution of the Lax equations \eqref{eq:527}.		

However, though the Lax equations \eqref{eq:527}
are implied by the Hamiltonian equations \eqref{eq:hameq},
a priori the former do not imply the latter.
The second step in the proof of Theorem \ref{thm:main}
is Theorem \ref{prop:sylvain},
which states that indeed the Lax equations \eqref{eq:527}
do imply, for an operator of the correct shape, 
the full hierarchy of Hamiltonian equations \eqref{eq:hameq}.

According to \cite{DSKV16b},
the generators of the $\mc W$-algebra $\mc W(\mf{gl}_N,\ul p)$
are naturally encoded into the $r\times r$ matrix differential operator
$$
W(\partial)
=
\begin{pmatrix}
W_{\bm1\bm1}(\partial) & W_{\bm1\bm2}(\partial) \\
W_{\bm2\bm1}(\partial) & W_{\bm2\bm2}(\partial) 
\end{pmatrix}
\,,
$$
with blocks of sizes $r_1\times r_1$, $r_1\times(r-r_1)$, $(r-r_1)\times r_1$ and $(r-r_1)\times(r-r_1)$,
where $r_1$ is the multiplicity of the largest part $p_1$ of the partition $\ul p$
(see \eqref{eq:matrW}-\eqref{eq:blockW}).
The coefficients of its entries are all the (free) generators of $\mc W(\mf{gl}_N,\ul p)$.
The key point in the proof of Theorem \ref{prop:sylvain}
is Corollary \ref{prop:dW4dtn},
which states that the differential operator $W_{\bm2\bm2}(\partial)$
does not evolve along the Hamiltonian flow \eqref{eq:hameq}. Moreover Theorem  \ref{prop:sylvain}
provides
explicit evolution equations (\ref{eq:evol2})-(\ref{eq:QR1}) for all generators of
$\mc W\mf{(gl}_N,\underline{p})$; in particular, we obtain that all these equations are
Hamiltonian and the corresponding Hamiltonian flows commute. This is the second
main result of the paper.

Along the way we obtain, in Section \ref{sec:8b},
a new, algorithmic way, to construct the generators of the $\mc W$-algebra
$\mc W(\mf{gl}_N,\ul p)$.
In Section $14$ we consider in more detail two simplest examples of the $\ul p$-reduced $r$-component KP hierarchies beyond the $r=1$ Gelfand-Dickey case, corresponding to the partiations in $r$ parts: $N=p+1+...+1$ with $ p>1$, $r>1$, and $N=p+2$ with $p>2$, $r=2$. We show that a reduction of the first hierarchy is the well-known $p$-constrained $r-1$ vector KP hierarchy, see e.g. \cite{Zhang99}, while the second hierarchy seems to be new. We construct polynomial tau-functions for these two examples, and soliton type tau-functions for general $\ul p $-reduced $r$-component KP hierarchies in Section 15.

A variation of a special case of the reduction \eqref{eq:intro2}		
of the $r$-component KP hierarchy was considered in \cite{Zhang99}		
and applied to the construction of solutions of the $p$-constrained KP hierarchies.

Throughout the paper the base field is the field of complex numbers $\mb C$.

\subsubsection*{Acknowledgments} 

The first, second, third and fourth author are extremely grateful to the IHES
for their kind hospitality during the summer of 2019, when the paper was completed.
The first author was supported by a Junior Fellow award from the Simons Foundation.
The second author was partially supported 
by the national PRIN fund n.\ 2015ZWST2C$\_$001
and the University funds n. RM116154CB35DFD3 and RM11715C7FB74D63.
The third author is supported by the Bert and Ann Kostant fund.
The fourth author received funding from the European Research Council (ERC) under the European Union's 
Horizon 2020 research and innovation program (QUASIFT grant agreement 677368).
%

\section{Review of the $\ul p$-KdV bilinear equations}
\label{S2}

Let $r$ be a positive integer.
We will use the following notation. For $\underline m=(m_1,m_2,\ldots, m_r)\in \mathbb{Z}^r$ we let 
\[
|\underline m|_0=0\,,\,\,
|\underline m|_a=\sum_{i=1}^a m_i \,,\,\,
1\le a\le r\,\,,\,\,\,\, |\underline m|=|\underline m|_r\, .
\]
Let $\underline e_a=(\delta_{ia})_{i=1}^r$, for $1\leq a\leq r$, be the standard basis of $\mathbb{Z}^r$.

Let $\ul p$ be a partition of a positive integer $N$, consisting of $r$ parts, i.e.
 $\ul p=(p_1,p_2,\dots,p_r)\in\mb Z_{>0}^r$, where $p_1\ge p_2\ge\cdots\ge p_r>0$,  
and $N=\sum_{i=1}^r p_i$.

Let $\ul{\bm t}= (\bm t^{(1)},\dots,\bm t^{(r)})$ be an $r$-tuple of infinite sequences
$\bm t^{(a)}=\big(t_j^{(a)}\big)_{j\in\mb Z_{>0}}$ of independent variables (times).
We also denote $\frac{\partial}{\partial\ul{\bm t}}
=\big(\frac{\partial}{\partial t_j^{(a)}}\big)_{1\leq a\leq r,\,j\in\mb Z_{>0}}$.
Consider the following operators
\[
e_+^{(a)}(\ul{\bm t},z)=\exp\left(\sum_{j=1}^\infty z^j t_j^{(a)}\right)\,, \quad 
e_-^{(a)}\big(\frac{\partial}{\partial\ul{\bm t}},z\big)
=\exp\left(-\sum_{j=1}^\infty \frac{z^{-j} }{j} \frac{\partial}{\partial t_j^{(a)}}\right)
\,.
\]

A \emph{tau-function} of \emph{charge} $k\in\mb Z$ is a collection of functions 
of the time variables $\ul{\bm t}$,
parameterized by the elements $\ul m\in\mb Z^r$ such that $|\ul m|=k$:
\begin{equation}\label{eq:tau}
  \vec{\tau}(\ul{\bm t})
  =\big\{ \tau^{\underline m}(\ul{\bm t}) \big\}_{|\underline m|=k}
\,\subset\mc F\,.
\end{equation} 
Throughout the paper we shall assume, unless otherwise specified, 
that all functions that we shall consider 
are smooth in all the time variables and lie in a certain differential field $\mc F$.
The $\ul p$-\emph{KdV bilinear equation} on the tau-function $\vec{\tau}(\ul{\bm t})$
of charge $k$ is defined as  the following system of bilinear equations
\cite{KvdL03}, \cite[Eq.(41)]{KvdL19}:
\begin{equation}
\label{mn-KdV3}
\Res_z \sum_{a=1}^r 
(-1)^{|\underline{m}^\prime+\underline{m}^{\prime\prime}|_{a-1}}
z^{m_a^\prime-m^{\prime\prime}_a+\ell p_a-2}
e_+^{(a)}(\ul{\bm t}^\prime-\ul{\bm t}^{\prime\prime},z)
\,e_-^{(a)}\Big(\frac{\partial}{\partial \ul{\bm t}^\prime}-\frac{\partial}{\partial \ul{\bm t}^{\prime\prime}},z\Big)
\tau^{\underline{m}^{\prime}-\underline{e}_a}(\ul{\bm t}^\prime)
\tau^{\underline{m}^{\prime\prime}+\underline{e}_a}(\ul{\bm t}^{\prime\prime})  =0 \, ,
\end{equation}
for $\ell\in\mb Z_{\geq0}$, and $\underline m^{\prime},\underline m^{\prime\prime}\in\mb Z^r$ such that 
$|\underline{m}^{\prime}|= k+1$ and $|\underline{m}^{\prime\prime}|=k-1$.
Here and further $\Res_z$ denotes, as usual, the coefficient of $z^{-1}$.
\begin{remark}\label{rem:senso}
Note that equation \eqref{mn-KdV3}, as stated, is not well defined
since the coefficient of each power of $z$ inside the residue is an infinite sum and
it may leads to divergences
(indeed, $e_+^{(a)}$ expands as an infinite series in $z$
while $e_-^{(a)}$ expands as an infinite series in $z^{-1}$).
Thus, equation \eqref{mn-KdV3}, in order to make sense, has to be correctly interpreted
as follows.
For each collection of integers $\{n_j^{(a)}\in\mb Z_{\geq0}\,|\,j\in\mb Z_{>0},\,a=1,\dots,r\}$,
all but finitely many equal to zero,
we get the corresponding (meaningful) equation ``coming from \eqref{mn-KdV3}''
by formally applying the derivatives
\begin{equation}\label{eq:deriv}
\prod_{j,a}\Big(\frac{\partial}{\partial{t^\prime}_j^{(a)}}\Big)^{n_j^{(a)}}
\end{equation}
inside the residue in the LHS and then setting $\ul{\bm t}^\prime=\ul{\bm t}^{\prime\prime}\,(=\ul{\bm t})$.
In doing so, only a finite number of terms survive from the expansion of $e_-^{(a)}$
and we thus get a meaningful collection of equations,
which are known as the Hirota bilinear equations \cite{DJKM81,KvdL03}.
\end{remark}
\begin{remark}\label{rem:translate}
If $\vec{\tau}(\ul{\bm t})$ solves the $\ul p$-\emph{KdV bilinear equation} \eqref{mn-KdV3} with charge $k$,
then, for arbitrary $\ul q\in\mb Z^r$,
we obtain a solution $T^{\ul q}\vec{\tau}(\ul{\bm t})$ of charge $k-|\ul q|$
by shifting all upper indices by $\ul q$:
\begin{equation}\label{eq:shift}
(T^{\ul q}\tau)^{\ul m}(\ul{\bm t})
:=
\tau^{\ul m+\ul q}(\ul{\bm t})
\,.
\end{equation}
\end{remark}
\begin{remark}\label{rem:ell=0}
Equation \eqref{mn-KdV3} for $\ell=0$ is the equation on the tau-function of the $r$-component
KP hierarchy.
Let 
$$
D_\ell=\sum_{a=1}^r \frac{\partial}{\partial t_{\ell p_a}^{(a)}}
\,,\,\,
\ell\in\mathbb{Z}_{>0}
\,\,,\,\,\,\,
D=D_1
\,.
$$
One can show \cite{KvdL03} that equations \eqref{mn-KdV3} for all $\ell\in\mb Z_{\geq0}$ 
for the tau-function $\vec{\tau}(\ul{\bm t})$
are equivalent to the equation for the tau-functions of the 
$r$-component KP-hierarchy ($\ell=0$) with the constraints:
\begin{equation}
\label{Kac2}
D_\ell\tau^{\underline m}(\ul{\bm t})
=
c_\ell \tau^{\underline m}(\ul{\bm t})
\,\,,\,\,\,\,
\ell\in\mathbb{Z}_{>0},\ c_\ell\in\mathbb{C}\,,
\end{equation}
and that these constraints with $\ell>0$ are equivalent to that with $\ell=1$.
%
Note that, if all the tau-functions $\tau^{\underline m}(\ul{\bm t})$ are polynomial,
equation \eqref{Kac2} can hold only when all constants $c_j$ vanish,
and hence all functions $\tau^{\ul m}(\ul{\bm t})$ 
are in the kernel of all operators $D_\ell$, $\ell\in\mb Z_{>0}$.
\end{remark}

One can describe all polynomial solutions of equation \eqref{mn-KdV3} as follows.
Fix the following data: an integer
$s\in\{1,\dots,N\}$
and, for every $1\leq i\leq s$ and $1\leq a\leq r$, 
\begin{equation}\label{eq:data}
n_i^{(a)} \in\mb Z_{>0}
\,\,,\qquad
b_i^{(a)} \in\mb C
\,\,,\qquad
\bm c_i^{(a)}
=
\big(c_{ij}^{(a)}\in\mb C\big)_{j\in\mb Z_{>0}}
\,.
\end{equation}
Define the integers:
\begin{equation}\label{eq:data6}
k_i:=
\max\big\{
\lceil\frac{n_i^{(a)}}{p_a}\rceil-1\,\big|\, 1\le a\le r \big\}
\,\,,\,\,\,\, 
1\leq i\leq s
\,,
\end{equation}
and
\begin{equation}\label{eq:data5}
k:=s+\sum_{i=1}^sk_i\,,
\end{equation}
where $\lceil x\rceil$ denotes the upper integer part of $x\in\mb Q$.
Consider the following functions associated to the data \eqref{eq:data}:
\begin{equation}\label{eq:hi}
h_i(\ul{\bm t}) 
=
\sum_{a=1}^r b_i^{(a)} S_{n_i^{(a)}} (\bm t^{(a)}+\bm c_i^{(a)})
\,,\,\, 1\leq i\leq s
\,,
\end{equation}
where $S_n(\bm t)$ are the elementary Schur polynomials.
Then, for each choice of the data \eqref{eq:data},
we have a tau-function $\vec{\tau}(\ul{\bm t})$ of charge $k$,
solution of the $\ul p$-KdV bilinear equation \eqref{mn-KdV3},
defined by letting
$\tau^{\ul m}(\ul{\bm t})=0$ unless $\ul m\in\mb Z_{\geq0}^r$ and $|\ul m|=k$,
in which case we let
\begin{equation}\label{eq:tau-sol}
\tau^{\ul m}(\ul{\bm t})
=
\det(M)
\,,
\end{equation}
where $M$ is the $k\times k$ matrix
written in block form as
\begin{equation}\label{eq:Mai1}
M=\big(M_{ai}\big)_{\substack{1\leq a\leq r \\ 1\leq i\leq s}}
\,,
\end{equation}
with the $m_a\times(k_i+1)$-block $M_{ai}$ given by
\begin{equation}\label{eq:Mai2}
M_{ai}
=
\bigg(
\frac{\partial^{m_a+1-\alpha}D^\beta h_i}{\partial \big(t_1^{(a)}\big)^{m_a+1-\alpha}}
\bigg)_{\substack{1\leq \alpha\leq m_a \\ 0\leq \beta\leq k_i}}
\,.
\end{equation}
(Since the $h_i$ are linear combinations of shifted elementary Schur polynomials,
one can replace $D^{\beta}$ in \eqref{eq:Mai2} by $D_{\beta}$.)
\begin{theorem}
\label{mmT}
(\cite[Thm.7]{KvdL19}).
For all choices of the data \eqref{eq:data},
the tau-function $\vec{\tau}(\ul{\bm t})$ defined by \eqref{eq:tau-sol}--\eqref{eq:Mai2}
is a polynomial solution of the $\ul p$-KdV bilinear equation \eqref{mn-KdV3},
of charge $k$ as in \eqref{eq:data5}.
All other polynomial solutions of \eqref{mn-KdV3}
(of arbitrary charge) are obtained by a shift as in \eqref{eq:shift}.
\end{theorem}

\section{The $\ul p$-KdV as a dynamical system}
\label{S3}

Now we rewrite equation \eqref{mn-KdV3} on the tau-functions
in the form of evolution equations \cite{DSKV13}, \cite{KvdL03}, \cite{Sato81}.
First, we turn equations \eqref{mn-KdV3} to  $r\times r$-matrix equations.
Given a tau-function $\vec{\tau}(\ul{\bm t})$, define the matrices
$P^\pm (\ul m,\ul{\bm t},z)=\left(P^\pm_{ab}(\ul m,\ul{\bm t},z)\right)_{a,b=1}^r$, where
\begin{equation}
\label{eq:P1}
P_{ab}^{\pm}(\ul m,\ul{\bm t},\pm z)=
(-1)^{|\ul e_a|_{b-1}}
\frac{z^{\delta_{ab}-1}}{\tau^{\ul m}(\ul{\bm t})}
e_-^{(b)}\big(\pm \frac{\partial}{\partial \ul{\bm t}},z\big)
\tau^{\ul m\pm\ul e_a\mp\ul e_b}(\ul{\bm t})
\,,\,\,\text{ if }\,\,\tau^{\ul m}(\ul{\bm t})\neq0
\,,
\end{equation}
and $P_{ab}^{\pm}(\ul m,\ul{\bm t},z)=0$ otherwise.\
Note that, if $\tau^{\ul m}(\ul{\bm t})\neq0$, we have  
\begin{equation}\label{eq:daniele1}
P^\pm (\ul m,\ul{\bm t},z)=\id_{r\times r}+O(z^{-1})
\,.
\end{equation}
We also let
\begin{equation}
\label{eq:R1}
R(\ul m,z)
=
\sum_{a=1}^r
(-1)^{|\ul m|_{a-1}}
z^{m_a}
E_{aa}
\,.
\end{equation}
Then, the $\ul p$-KdV bilinear equation \eqref{mn-KdV3} on $\vec{\tau}(\ul{\bm t})$
turns into the following system of equations on the collection of matrices 
$\big\{P^{\pm}(\ul m,\ul{\bm t},z)\big\}_{|\ul m|=k}$:
\begin{equation}
\label{mn-KdV7}
\Res_z P^{+}(\ul m^\prime,\ul{\bm t}^\prime,z)
R(\ul m^\prime-\ul m^{\prime\prime},z)
\Big( 
\sum_{a=1}^{r}z^{\ell p_a}
e_+^{(a)}(\ul{\bm t}^\prime-\ul{\bm t}^{\prime\prime},z)
E_{aa} 
\Big)
P^{-}(\ul m^{\prime\prime},\ul{\bm t}^{\prime\prime},-z)^T=0\
\,,
\end{equation}
for all 
$\ell\in\mathbb{Z}_{\ge 0}$ and $\ul m^\prime,\ul m^{\prime\prime}\in\mb Z^r$ such that 
$|\ul m^\prime|=|\ul m^{\prime\prime}|=k$.
Indeed, 
if $\tau^{\ul m^\prime}(\ul{\bm t}^\prime),\tau^{\ul m^{\prime\prime}}(\ul{\bm t}^{\prime\prime})\neq0$,
then the $(i,j)$-entry of the LHS of \eqref{mn-KdV7} coincides
with the LHS of \eqref{mn-KdV3} 
with $\ul m^\prime+\ul e_i$ in place of $\ul m^\prime$
and $\ul m^{\prime\prime}-\ul e_j$ in place of $\ul m^{\prime\prime}$.
Hereafter $A^T$ stands for the transpose of the matrix $A$
and $E_{ab}$ denotes the standard basis elements of the space of matrices.
Note that equation \eqref{mn-KdV7} has the same divergence issues as equation \eqref{mn-KdV3}
and it has to be correctly interpreted,
as explained in Remark \ref{rem:senso}.

Introduce a new variable $x$ and replace in \eqref{mn-KdV7} 
all $t_1^{(a)}$ by $t_1^{(a)}+x$ for all $1\le a\le r$.
Hence, the translated times are $\ul{\bm t}+x\ul{\bm e}_1$,
where $\bm e_{1\,j}^{(a)}=\delta_{j1}$.
We denote by $\vec{\tau}(x,\ul{\bm t}):=\tau(\ul{\bm t}+x\ul{\bm e}_1)$ the resulting ``translated'' tau-function,
and by 
\begin{equation}
\label{eq:P2}
P^{\pm}(\ul m,x,\ul{\bm t},z)
=
\big(P^{\pm}_{ab}(\ul m,x,\ul{\bm t},z)\big)_{a,b=1}^r
:=
P^{\pm}(\ul m,\ul{\bm t}+x\ul{\bm e}_1,z)
\,,
\end{equation}
the resulting translated matrices $P^{\pm}$.
Then, the system of equations \eqref{mn-KdV7}, with $\ul{\bm t}^\prime$ translated by $x^\prime$
and $\ul{\bm t}^{\prime\prime}$ translated by $x^{\prime\prime}$, 
can be equivalently rewritten, using this new notation,
as follows:
\begin{equation}
\label{eq:kdvP}
\Res_z 
P^+(\ul m^\prime,x^\prime,\ul{\bm t}^\prime,z)
R(\ul m^\prime-\ul m^{\prime\prime},z)
\Big( 
\sum_{a=1}^{r}z^{\ell p_a}
e_+^{(a)}(\ul{\bm t}^\prime-\ul{\bm t}^{\prime\prime},z)
E_{aa} 
\Big)
P^-(\ul m^{\prime\prime},x^{\prime\prime},\ul{\bm t}^{\prime\prime},-z)^T
e^{z(x^\prime-x^{\prime\prime})}
=0
\,,
\end{equation}
for all $\ell\in\mathbb{Z}_{\ge 0}$ and $\ul m^{\prime},\ul m^{\prime\prime}\in\mb Z^r$
such that $|\ul m^\prime|=|\ul m^{\prime\prime}|=k$.

In order to rewrite equation \eqref{eq:kdvP-b} in a nicer form,
we need some simple results on pseudodifferential operators.
For a scalar pseudodifferential operator $a(\partial)$, we denote by $a^*(\partial)$
its formal adjoint, and by $a(z)$ its symbol.
Also, for a matrix pseudodifferential operator $A(\partial)=\sum_j A_j\partial_j$,
we denote by $A^*(\partial)=\sum_j (-\partial)^j\cdot A_j^T$ be its formal adjoint 
and by $A(z)=\sum_j A_j z^j$ its symbol. Here $\circ$ is the product of matrix pseudodifferential operators.
We shall also denote, as usual, by $A(\partial)_+$ the differential part
of the matrix pseudodifferential operator $A(\partial)$,
and let $A(\partial)_-=A(\partial)-A(\partial)_+$. We shall drop the sign $\circ$ if no confusion may arise.
\begin{lemma}\label{lem:fund1}
For every matrix pseudodifferential operators $P(\partial),Q(\partial)$, we have
$$
\Res_z P(z)Q^*(-z)
=\Res_z(P\circ Q^T)(z)
\,.
$$
\end{lemma}
\begin{proof}
In the scalar case it is stated and proved in \cite[Lem.2.1(a)]{DSKV16a}, putting $\lambda=0$ there
and using \cite[Eq.(2.1)]{DSKV16a}.
The matrix case is obtained from the scalar case looking at each matrix entry.
\end{proof}
\begin{lemma}\label{lem:fund2}
For every matrix pseudodifferential operators 
$A(x,\partial),B(x,\partial)$,
where $\partial=\partial_x$,
we have
\begin{equation}\label{eq:fund1}
\Res_z A(x,z)\big(e^{xz}\partial^{-1}\circ e^{-xz}\big)\circ B(x,-z)^T
=
\big(A(x,\partial)\circ B(x,\partial)^*\big)_-
\,.
\end{equation}
\end{lemma}
\begin{proof}
Note that $e^{xz}\partial^{-1}\circ e^{-xz}=\iota_{\partial,z}(\partial-z)^{-1}$,
where $\iota_{\partial,z}$ denotes the geometric expansion in non-negative powers of $z$.
We then apply
Lemma \ref{lem:fund1} to the matrix operators with symbols
$P(z)=A(x,z)$ and $Q(z)=(B^*)^T(x,z)\iota_{\partial,z}(\partial-z)^{-1}$,
so that $Q^*(-z)=\iota_{\partial,z}(\partial-z)^{-1}B^T(x,-z)$.
We thus get
$$
\Res_z A(x,z)e^{xz}\partial^{-1}\circ e^{-xz}B(x,-z)^T
=
\Res_z (A\circ B^*)(z)\iota_{\partial,z}(\partial-z)^{-1}
=
(A(x,\partial)\circ B(x,\partial)^*)_-
\,.
$$
\end{proof}
\begin{lemma}\label{lem:fund3}
Let $f_i(x),g_i(x)$, $i=1,\dots,n$, be smooth functions in the variable $x$,
and assume that they lie in a domain.
Then,
\begin{equation}\label{eq:a}
\sum_{i=1}^n f_i(x^\prime)g_i(x^{\prime\prime})=0
\,,
\end{equation}
as a function in two variables $x^\prime,x^{\prime\prime}$,
if and only if
\begin{equation}\label{eq:b}
\sum_{i=1}^n f_i(x)\partial^{-1}\circ g_i(x)=0
\,,
\end{equation}
as a pseudodifferential operator.
\end{lemma}
\begin{proof}
By replacing the domain containing all our functions 
with its field of fractions,
we see that it is enough to prove the claim over a function field $\mc F$.
We can identify the function in two variables $\sum_{i=1}^n f_i(x^\prime)g_i(x^{\prime\prime})$
with the corresponding element $\sum_{i=1}^n f_i\otimes g_i\in\mc F\otimes_{\mb C}\mc F$
(the tensor product being over the subfield of constants $\mb C$).
Hence, the lemma reduces to proving that the linear map
$$
\mc F\otimes\mc F
\to
\mc F\partial^{-1}\circ\mc F
\subset
\mc F((\partial^{-1}))
\,\,,\,\,\,
f\otimes g\mapsto f\partial^{-1}\circ g
\,,
$$
is injective.
Suppose that \eqref{eq:b} holds;
we want to prove that
\begin{equation}\label{eq:c}
\sum_{i=1}^n f_i\otimes g_i=0\,\,\text{ in }\,\,\mc F\otimes\mc F
\,.
\end{equation}
We prove the claim by induction on $n\geq1$.
Suppose, by contradiction, that \eqref{eq:c} fails,
and assume, without loss of generality,
that the functions $g_1,\dots,g_n$ are linearly independent over $\mb C$,
and all the functions $f_1,\dots,f_n$ are non-zero.
For $n=1$, equation \eqref{eq:b} implies,
looking at the order $-1$ term in $\mc F((\partial^{-1}))$,
$f_1g_1=0$, 
so that either $f_1=0$ or $g_1=0$,
a contradiction.
For $n\geq2$, we have, by \eqref{eq:b}
$$
f_n\partial^{-1}\circ g_n
=-\sum_{i=1}^{n-1}f_i\partial^{-1}g_i
\,.
$$
Dividing both sides by $f_n\neq0$,
and multiplying by $\partial$ on the left of both sides,
we get
$$
g_n
=-\frac1{f_n}\sum_{i=1}^{n-1}f_ig_i
\,\,\text{ and }\,\,
\sum_{i=1}^{n-1}\big(\frac{f_i}{f_n}\big)^\prime\partial^{-1}g_i=0
\,.
$$
By assumption, the functions $g_i$ are linearly independent over the constants.
Hence, by the inductive assumption,
the functions $\big(\frac{f_i}{f_n}\big)^\prime$ are all zero,
i.e. $f_i=\alpha_i f_n$, $i=1,\dots,n-1$, for some non-zero constants $\alpha_1,\dots,\alpha_{n-1}\in\mb C$.
Hence, \eqref{eq:b} can be rewritten as
$$
f_n\partial^{-1}\circ\big(\sum_{i=1}^n\alpha_ig_i\big)=0
\,,
$$
where we set $\alpha_n=1$.
Dividing by $f_n$ and multiplying on the left by $\partial$, we get
$\sum_{i=1}^n\alpha_ig_i=0$, contradicting
the linear independence assumption.
\end{proof}

By Lemma \ref{lem:fund3}, equation \eqref{eq:kdvP} is equivalent to
$$
\Res_z 
P^+(\ul m^\prime,x,\ul{\bm t}^\prime,z)
R(\ul m^\prime\!-\ul m^{\prime\prime},z)
\Big(\!
\sum_{a=1}^{r}z^{\ell p_a}
e_+^{(a)}(\ul{\bm t}^\prime\!-\ul{\bm t}^{\prime\prime},z)
E_{aa} 
\!\Big)
e^{zx}\partial^{-1}\circ e^{-zx}
P^-(\ul m^{\prime\prime}\!,x,\ul{\bm t}^{\prime\prime}\!,\!-z)^T
=0
\,,
$$
where $\partial=\partial_x$.
Then, applying Lemma \ref{lem:fund2}, we rewrite the above equation as
\begin{equation}\label{eq:kdvP-b}
\Big(
P^+(\ul m^\prime,x,\ul{\bm t}^\prime,\partial)
\circ
R(\ul m^\prime\!-\ul m^{\prime\prime},\partial)
\circ 
\Big(\!
\sum_{a=1}^{r}\partial^{\ell p_a}\circ
e_+^{(a)}(\ul{\bm t}^\prime\!-\ul{\bm t}^{\prime\prime},\partial)\circ
E_{aa} 
\!\Big)
\circ
P^-(\ul m^{\prime\prime}\!,x,\ul{\bm t}^{\prime\prime}\!,\partial)^*
\Big)_-
=0
\,.
\end{equation}
This equation holds for every $\ell\in\mb Z_{\geq0}$
and $\ul m^\prime,\ul m^{\prime\prime}\in\mb Z^r$ such that $|\ul m^\prime|=|\ul m^{\prime\prime}|=k$.
Note that also equation \eqref{eq:kdvP-b}, as \eqref{mn-KdV3}, has to be correctly interpreted
as it may involve diverging sums.
As explained in Remark \ref{rem:senso},
the way to give a meaningful sense to it,
is to apply arbitrary derivatives \eqref{eq:deriv} w.r.t. $\ul{\bm t}^\prime$ to the LHS 
and then set $\ul{\bm t}^\prime=\ul{\bm t}^{\prime\prime}$.
In this way, all divergences disappear.

Next, we set $\ell=0$, $\ul m^\prime=\ul m^{\prime\prime}\,(=\ul m)$
and $\ul{\bm t}^\prime=\ul{\bm t}^{\prime\prime}\,(=\ul{\bm t})$
in equation \eqref{eq:kdvP-b}, to get
$$
\Big(
P^+(\ul m,x,\ul{\bm t},\partial)
\circ
P^-(\ul m,x,\ul{\bm t},\partial)^*
\Big)_-
=0
\,.
$$
This, combined with \eqref{eq:daniele1}, implies
\begin{equation}
\label{ee11}
P^- (\ul m,x,\ul{\bm t},\partial)^*
=
\big(P^+ (\ul m,x,\ul{\bm t},\partial)\big)^{-1}
\,,
\end{equation}
if $\tau^{\ul m}(\ul{\bm t})\neq0$ (so that $P^+(\ul m,x,\ul{\bm t},\partial)$ is invertible).

If instead we first apply $\frac{\partial}{\partial t_j^{(a)}}$ to both sides of \eqref{eq:kdvP-b}
and then set $\ell=0$, $\ul m^\prime=\ul m^{\prime\prime}\,(=\ul m)$
and $\ul{\bm t}^\prime=\ul{\bm t}^{\prime\prime}\,(=\ul{\bm t})$, we get
$$
\Big(
\frac{\partial P^+}{\partial t_j^{(a)}} (\ul m,x,\ul{\bm t},\partial)
\circ
P^-(\ul m,x,\ul{\bm t},\partial)^*
\Big)_-
+
\Big(
P^+(\ul m,x,\ul{\bm t},\partial)
\partial^j
\circ
P^-(\ul m,x,\ul{\bm t},\partial)^*
\Big)_-
=0
\,.
$$
This, combined with \eqref{ee11} and \eqref{eq:daniele1}, gives the Sato-Wilson equation for 
$P^+(\ul m,x,\ul{\bm t},\partial)$, \cite[Lem.4.2]{KvdL03}:
\begin{equation}
\label{Sato1}
\frac{\partial P^+}{\partial t_j^{(a)}}(\ul m,x,\ul{\bm t},\partial)
=
-\left(
P^+ (\ul m,x,\ul{\bm t},\partial)\circ  
E_{aa} \partial^j \circ 
P^+(\ul m,x,\ul{\bm t},\partial)^{-1}
\right)_-\circ 
P^+(\ul m,x,\ul{\bm t},\partial)
\, ,
\end{equation}
for all $1\leq a\leq r$ and $j\in\mb Z_{>0}$,
and $\ul m$ such that $\tau^{\ul m}(\ul{\bm t})\neq0$.

For $\ul m\in\mb Z^r$  such that $|\underline m|=k$,  let 
\begin{equation}\label{eq:La}
L_a=L_a(\ul m,x,\ul{\bm t},\partial)
=
P^+(\ul m,x,\ul{\bm t},\partial)\circ E_{aa}\partial\circ 
P^{+}(\ul m,x,\ul{\bm t},\partial)^{-1}
\,,
\end{equation}
if $\tau^{\ul m}(x,\ul{\bm t})\neq0$,
and $L_a=0$ otherwise.
Clearly, we have
\begin{equation}\label{eq:Lcommute}
L_aL_b=0
\,\text{ if }\, a\neq b
\,.
\end{equation}
Moreover, setting 
$\ell=1$, $\ul m^\prime=\ul m^{\prime\prime}\,(=\ul m)$ 
and $\ul{\bm t}^\prime=\ul{\bm t}^{\prime\prime}\,(=\ul{\bm t})$
in equation \eqref{eq:kdvP-b}, we obtain the following constraint for the operators $L_a$:
\begin{equation}\label{eq:Lconstraint}
\Big(\sum_{a=1}^r L_a^{p_a}\Big)_-=0
\,.
\end{equation}
Finally, the Sato-Wilson equation \eqref{Sato1}
implies that the operators $L_a$ 
evolve according to the Lax equations (see e.g. \cite[Lem.4.3]{KvdL03}):
\begin{equation}\label{eq:Lax}
\frac{\partial L_a}{\partial t_j^{(b)}}=[(L_b^j)_+, L_a]
\,\,,\quad 1\le a,b\le r,\,j\in\mb Z_{>0}
\,.
\end{equation}
Note that, even though \eqref{Sato1} holds only (in fact makes sense)
for $\ul m$ such that $\tau^{\ul m}(x,\ul{\bm t})\neq0$,
the Lax equation \eqref{eq:Lax} holds for every $\ul m$.

\section{The $p_1$-reduction}
\label{S4}

Recall the partition $\ul p=(p_1,p_2,\cdots, p_r)$ from Section \ref{S2}. 
We will assume from now on that $p_1>1$. 
Let $r_1$ be the multiplicity of the largest part $p_1$.

The $p_1$-reduction consists of putting in \eqref{mn-KdV7}
all  times ${t^\prime}_j^{(a)}$ and ${t^{\prime\prime}}_j^{(a)}$ equal zero 
for all $j\in\mb Z_{>0}$ and $r_1<a\leq r$. 
Namely, we let
\begin{equation}\label{eq:Q}
Q^{\pm}(\ul m,\ul{\bm t},z)
=
P^{\pm}(\ul m,\ul{\bm t},z)\,\big|_{\bm t^{(a)}=0 \text{ for } a>r_1}
\,.
\end{equation}
By equation \eqref{eq:daniele1} we immediately have
\begin{equation}\label{eq:daniele2}
Q^\pm (\ul m,\ul{\bm t},z)=\id_{r\times r}+O(z^{-1})
\,\,\text{ if }\,\,
\tau^{\ul m}(\ul{\bm t})\neq0
\,.
\end{equation}
Setting ${\bm t^\prime}^{(a)}={\bm t^{\prime\prime}}^{(a)}=0$ for $a>r_1$ 
in equation \eqref{mn-KdV7}, we get
\begin{equation}\label{mn-KdV7-b}
\Res_z Q^{+}(\ul m^\prime,\ul{\bm t}^\prime,z)
R(\ul m^\prime-\ul m^{\prime\prime},z)
\Big( 
\sum_{a=1}^{r_1}
e_+^{(a)}(\ul{\bm t}^\prime-\ul{\bm t}^{\prime\prime},z)
E_{aa} 
z^{\ell p_1}
+\!\!\!
\sum_{a=r_1+1}^{r}z^{\ell p_a}
E_{aa} 
\Big)
Q^{-}(\ul m^{\prime\prime},\ul{\bm t}^{\prime\prime},-z)^T=0\
\,,
\end{equation}
for all 
$\ell\in\mathbb{Z}_{\ge 0}$ and $\ul m^\prime,\ul m^{\prime\prime}\in\mb Z^r$ such that 
$|\ul m^\prime|=|\ul m^{\prime\prime}|=k$.

Next, in analogy with what we did in Section \ref{S3},
we shift $t_1^{(a)}$ by $x$ for all $1\leq a\leq r_1$.
We denote by 
\begin{equation}
\label{eq:Q2}
Q^{\pm}(\ul m,x,\ul{\bm t},z)
=
Q^{\pm}(\ul m,\ul{\bm t}+x{\ul{\bm e}}_1,z)
\,,
\end{equation}
the resulting translated matrices $Q^{\pm}$,
where now $\bm e_{1\,j}^{(a)}=\delta_{j,1}\delta_{a\leq r_1}$.
Note that $Q^{\pm}(\ul m,x,\ul{\bm t},z)$ is not the restriction of 
$P^{\pm}(\ul m,x,\ul{\bm t},z)$ at $\bm t^{(a)}=0$ for $a>r_1$,
since, first restricting and then shifting by $x$ is not the same as first shifting and then restricting.

The system of equations \eqref{mn-KdV7-b}, with $\ul{\bm t}^\prime$ translated by $x^\prime$
and $\ul{\bm t}^{\prime\prime}$ translated by $x^{\prime\prime}$, 
can be equivalently rewritten, using this new notation,
as follows:
\begin{equation}\label{eq:kdvQ}
\begin{split}
& \Res_z 
Q^+(\ul m^\prime,x^\prime,\ul{\bm t}^\prime,z)
R(\ul m^\prime-\ul m^{\prime\prime},z)
\Big( 
\sum_{a=1}^{r_1}
e_+^{(a)}(\ul{\bm t}^\prime-\ul{\bm t}^{\prime\prime},z)
E_{aa} 
\Big)
z^{\ell p_1}
Q^-(\ul m^{\prime\prime},x^{\prime\prime},\ul{\bm t}^{\prime\prime},-z)^T
e^{z(x^\prime-x^{\prime\prime})} \\
& +
\Res_z 
Q^+(\ul m^\prime,x^\prime,\ul{\bm t}^\prime,z)
R(\ul m^\prime-\ul m^{\prime\prime},z)
\Big( 
\sum_{a=r_1+1}^{r}z^{\ell p_a}
E_{aa} 
\Big)
Q^-(\ul m^{\prime\prime},x^{\prime\prime},\ul{\bm t}^{\prime\prime},-z)^T
=0
\,,
\end{split}
\end{equation}
for all $\ell\in\mathbb{Z}_{\ge 0}$ and $\ul m^{\prime},\ul m^{\prime\prime}\in\mb Z^r$
such that $|\ul m^\prime|=|\ul m^{\prime\prime}|=k$.

Next, in the same way as we derived equation \eqref{eq:kdvP-b} starting from \eqref{eq:kdvP},
we use Lemmas \ref{lem:fund2} and \ref{lem:fund3} to get from equation \eqref{eq:kdvQ}
\begin{equation}\label{eq:kdvQ-b}
\begin{split}
&\Big(
Q^+(\ul m^\prime,x,\ul{\bm t}^\prime,\partial)
\circ
R(\ul m^\prime\!-\ul m^{\prime\prime},\partial)
\circ 
\Big(\!
\sum_{a=1}^{r_1}
e_+^{(a)}(\ul{\bm t}^\prime\!-\ul{\bm t}^{\prime\prime},\partial)
E_{aa} 
\!\Big)
\partial^{\ell p_1}
\circ
Q^-(\ul m^{\prime\prime}\!,x,\ul{\bm t}^{\prime\prime}\!,\partial)^*
\Big)_- \\
& +
\Res_z 
Q^+(\ul m^\prime,x,\ul{\bm t}^\prime,z)
R(\ul m^\prime-\ul m^{\prime\prime},z)
\Big( 
\sum_{a=r_1+1}^{r}z^{\ell p_a}
E_{aa} 
\Big)
\partial^{-1}\circ
Q^-(\ul m^{\prime\prime},x,\ul{\bm t}^{\prime\prime},-z)^T
=0
\,.
\end{split}
\end{equation}
The above equation holds for every $\ell\in\mb Z_{\geq0}$
and $\ul m^\prime,\ul m^{\prime\prime}\in\mb Z^r$ such that $|\ul m^\prime|=|\ul m^{\prime\prime}|=k$.
We observe once more that equation \eqref{eq:kdvQ-b},
as written, makes no sense as it may have diverging series.
As explained in Remark \ref{rem:senso}, it has to be correctly interpreted
as the collection of equations obtained by applying the derivatives \eqref{eq:deriv} to the LHS
and then setting $\ul{\bm t}^\prime=\ul{\bm t}^{\prime\prime}$.
In doing so, all diverging series disappear.

It is convenient to write the matrices $Q^{\pm}$ in block form as
\begin{equation}\label{eq:block}
Q^{\pm}(\ul m,x,\ul{\bm t},z)
=
\begin{pmatrix}
Q^{\pm}_{\bm1\bm1}(\ul m,x,\ul{\bm t},z) & Q^{\pm}_{\bm1\bm2}(\ul m,x,\ul{\bm t},z) \\
Q^{\pm}_{\bm2\bm1}(\ul m,x,\ul{\bm t},z) & Q^{\pm}_{\bm2\bm2}(\ul m,x,\ul{\bm t},z) 
\end{pmatrix}\,,
\end{equation}
where
$$
Q^{\pm}_{\bm1\bm1}=\big(Q^{\pm}_{ab}\big)_{1\leq a,b\leq r_1}
\,\,,\,\,\,\,
Q^{\pm}_{\bm1\bm2}=\big(Q^{\pm}_{ab}\big)_{\substack{1\leq a\leq r_1 \\ r_1< b\leq r}}
\,\,,\,\,\,\,
Q^{\pm}_{\bm2\bm1}=\big(Q^{\pm}_{ab}\big)_{\substack{r_1<a\leq r \\ 1\leq b\leq r_1}}
\,\,,\,\,\,\,
Q^{\pm}_{\bm2\bm2}=\big(Q^{\pm}_{ab}\big)_{r_1<a,b\leq r}
\,.
$$
Note that, by \eqref{eq:daniele2}, we have
\begin{equation}\label{eq:daniele3}
Q^\pm_{\bm a\bm b} (\ul m,x,\ul{\bm t},z)
=
\delta_{\bm a,\bm b}\id+
\sum_{j=1}^\infty 
Q^\pm_{\bm a\bm b;j} (\ul m,x,\ul{\bm t})z^{-j}
\,\,,\,\text{ if }\,\,
\tau^{\ul m}(\ul{\bm t})\neq0
\,.
\end{equation}

Equation \eqref{eq:kdvQ-b} can be rewritten as the set of four equations,
depending on $\bm a,\bm b\in\{\bm 1,\bm 2\}$:
\begin{equation}\label{eq:kdvQ-c}
\begin{split}
&\Big(
Q^+_{\bm a\bm 1}(\ul m^\prime,x,\ul{\bm t}^\prime,\partial)
\circ 
\Big(\!
\sum_{a=1}^{r_1}
(-1)^{|\ul m^\prime-\ul m^{\prime\prime}|_{a-1}}
\partial^{m_a^\prime-m_a^{\prime\prime}}
e_+^{(a)}(\ul{\bm t}^\prime\!-\ul{\bm t}^{\prime\prime},\partial)
E_{aa} 
\!\Big)
\partial^{\ell p_1}
\circ
Q^-_{\bm b\bm 1}(\ul m^{\prime\prime}\!,x,\ul{\bm t}^{\prime\prime}\!,\partial)^*
\Big)_- \\
& +
\Res_z 
Q^+_{\bm a\bm 2}(\ul m^\prime,x,\ul{\bm t}^\prime,z)
\Big( 
\sum_{a=r_1+1}^{r}
(-1)^{|\ul m^\prime-\ul m^{\prime\prime}|_{a-1}}
z^{m_a^\prime-m_a^{\prime\prime}+\ell p_a}
E_{aa} 
\Big)
\partial^{-1}\circ
Q^-_{\bm b\bm 2}(\ul m^{\prime\prime},x,\ul{\bm t}^{\prime\prime},-z)^T
=0
\,.
\end{split}
\end{equation}

Following the same path as in Section \ref{S3} starting from \eqref{eq:kdvP-b},
taking various special cases of equations \eqref{eq:kdvQ-c}
we derive all the ``reduced analogues'' 
of the inversion formula \eqref{ee11},
the Sato-Wilson equation \eqref{Sato1},
the constraint condition \eqref{eq:Lconstraint} for the Lax operators $L_a$,
and the Lax equations \eqref{eq:Lax}.

First, we set $\ell=0$, $\ul m^\prime=\ul m^{\prime\prime}\,(=\ul m)$,
and $\ul{\bm t}^\prime=\ul{\bm t}^{\prime\prime}\,(=\ul{\bm t})$ in \eqref{eq:kdvQ-c}.
As a result, we get,
$$
\Big(
Q^+_{\bm a\bm 1}(\ul m,x,\ul{\bm t},\partial)
\circ 
Q^-_{\bm b\bm 1}(\ul m,x,\ul{\bm t},\partial)^*
\Big)_- \\
+
\Res_z 
Q^+_{\bm a\bm 2}(\ul m,x,\ul{\bm t},z)
\partial^{-1}\circ
Q^-_{\bm b\bm 2}(\ul m,x,\ul{\bm t},-z)^T
=0
\,.
$$
Using \eqref{eq:daniele3}, this leads to
\begin{equation}\label{eq:inverse1}
Q^+_{\bm a\bm 1}(\ul m,x,\ul{\bm t},\partial)
\circ 
Q^-_{\bm b\bm 1}(\ul m,x,\ul{\bm t},\partial)^*
-\delta_{\bm a,\bm1}\delta_{\bm b,\bm1}\id
+
\delta_{\bm b,\bm2}
Q^+_{\bm a\bm 2;1}(\ul m,x,\ul{\bm t})
\partial^{-1}
-
\delta_{\bm a,\bm2}
\partial^{-1}\circ
Q^-_{\bm b\bm 2;1}(\ul m,x,\ul{\bm t})^T
=0
\,.
\end{equation}
This is the ``reduced analogue'' of the inversion formula \eqref{ee11}.
It specializes, for the various choices of the indices $\bm a,\bm b$ 
to four equations,
which hold whenever $\tau^{\ul m}(\ul{\bm t})\neq0$.
For $\bm a=\bm b=\bm1$, we get
\begin{equation}\label{eq:inverse11}
Q^-_{\bm1\bm 1}(\ul m,x,\ul{\bm t},\partial)^*
=
Q^+_{\bm1\bm 1}(\ul m,x,\ul{\bm t},\partial)^{-1}
\,.
\end{equation}
For $\bm a=\bm1$, $\bm b=\bm2$, we get
$$
Q^+_{\bm1\bm 1}(\ul m,x,\ul{\bm t},\partial)
\circ 
Q^-_{\bm2\bm 1}(\ul m,x,\ul{\bm t},\partial)^*
+
Q^+_{\bm1\bm 2;1}(\ul m,x,\ul{\bm t})
\partial^{-1}
=0
\,,
$$
which, after applying $*$ to both sides, leads to
\begin{equation}\label{eq:inverse12}
Q^-_{\bm2\bm 1}(\ul m,x,\ul{\bm t},\partial)
=
\partial^{-1}
\circ
Q^+_{\bm1\bm 2;1}(\ul m,x,\ul{\bm t})^T
Q^+_{\bm1\bm 1}(\ul m,x,\ul{\bm t},\partial)^{*\,-1}
\,.
\end{equation}
For $\bm a=\bm2$, $\bm b=\bm1$, we get
$$
Q^+_{\bm2\bm 1}(\ul m,x,\ul{\bm t},\partial)
\circ 
Q^-_{\bm1\bm 1}(\ul m,x,\ul{\bm t},\partial)^*
-
\partial^{-1}\circ
Q^-_{\bm1\bm 2;1}(\ul m,x,\ul{\bm t})^T
=0
\,,
$$
which, by \eqref{eq:inverse11}, leads to
\begin{equation}\label{eq:inverse21}
Q^+_{\bm2\bm 1}(\ul m,x,\ul{\bm t},\partial)
=
\partial^{-1}\circ
Q^-_{\bm1\bm 2;1}(\ul m,x,\ul{\bm t})^T
Q^+_{\bm1\bm 1}(\ul m,x,\ul{\bm t},\partial)
\,.
\end{equation}
Finally, for $\bm a=\bm b=\bm2$, we get
$$
Q^+_{\bm2\bm 1}(\ul m,x,\ul{\bm t},\partial)
\circ 
Q^-_{\bm2\bm 1}(\ul m,x,\ul{\bm t},\partial)^*
+
Q^+_{\bm2\bm 2;1}(\ul m,x,\ul{\bm t})
\partial^{-1}
-
\partial^{-1}\circ
Q^-_{\bm2\bm 2;1}(\ul m,x,\ul{\bm t})^T
=0
\,,
$$
which, by \eqref{eq:inverse11}, \eqref{eq:inverse12} and \eqref{eq:inverse21}, leads to
$$
\partial^{-1}\circ
Q^-_{\bm1\bm 2;1}(\ul m,x,\ul{\bm t})^T
Q^+_{\bm1\bm 2;1}(\ul m,x,\ul{\bm t})
\partial^{-1}
=
Q^+_{\bm2\bm 2;1}(\ul m,x,\ul{\bm t})
\partial^{-1}
-
\partial^{-1}\circ
Q^-_{\bm2\bm 2;1}(\ul m,x,\ul{\bm t})^T
\,,
$$
or, multiplying both sides on the left and on the right by $\partial$
and comparing coefficients, we get
\begin{equation}\label{eq:inverse22}
\begin{split}
& Q^+_{\bm2\bm 2;1}(\ul m,x,\ul{\bm t})
=
Q^-_{\bm2\bm 2;1}(\ul m,x,\ul{\bm t})^T
\,,\\
& Q^-_{\bm1\bm 2;1}(\ul m,x,\ul{\bm t})^T
Q^+_{\bm1\bm 2;1}(\ul m,x,\ul{\bm t})
=
\frac{\partial Q^+_{\bm2\bm 2;1}}{\partial x}(\ul m,x,\ul{\bm t})
\,.
\end{split}
\end{equation}

Next, we apply $\frac{\partial}{\partial t_j^{(a)}}$, for $1\leq a\leq r_1$ and $j\in\mb Z_{>0}$,
 to both sides of \eqref{eq:kdvQ-c}
and then set $\ell=0$, $\ul m^\prime=\ul m^{\prime\prime}\,(=\ul m)$
and $\ul{\bm t}^\prime=\ul{\bm t}^{\prime\prime}\,(=\ul{\bm t})$.
As a result we get
\begin{align*}
\Big(
\frac{\partial Q^+_{\bm a\bm 1}}{\partial t_j^{(a)}}
(\ul m,x,\ul{\bm t},\partial)
&\circ
Q^-_{\bm b\bm 1}(\ul m,x,\ul{\bm t},\partial)^*
\Big)_- 
+
\Big(
Q^+_{\bm a\bm 1}(\ul m,x,\ul{\bm t},\partial)
\circ
E_{aa} \partial^{j}
\circ
Q^-_{\bm b\bm 1}(\ul m,x,\ul{\bm t},\partial)^*
\Big)_- \\
& +
\Res_z 
\frac{\partial Q^+_{\bm a\bm 2}}{\partial t_j^{(a)}}
(\ul m,x,\ul{\bm t},z)
\partial^{-1}\circ
Q^-_{\bm b\bm 2}(\ul m,x,\ul{\bm t},-z)^T
=0
\,.
\end{align*}
Using \eqref{eq:daniele3}, this equation can be rewritten as
\begin{equation}\label{eq:Qsato}
\begin{split}
\frac{\partial Q^+_{\bm a\bm 1}}{\partial t_j^{(a)}}
(\ul m,x,\ul{\bm t},\partial)
&\circ
Q^-_{\bm b\bm 1}(\ul m,x,\ul{\bm t},\partial)^*
+
\Big(
Q^+_{\bm a\bm 1}(\ul m,x,\ul{\bm t},\partial)
\circ
E_{aa}\partial^{j}
 \circ
Q^-_{\bm b\bm 1}(\ul m,x,\ul{\bm t},\partial)^*
\Big)_- \\
& +
\delta_{\bm b,\bm 2}
\frac{\partial Q^+_{\bm a\bm 2;1}}{\partial t_j^{(a)}}
(\ul m,x,\ul{\bm t})
\partial^{-1}
=0
\,.
\end{split}
\end{equation}
This is the ``reduced analogue'' of the Sato-Wilson equation \eqref{Sato1}.
Let us write down explicitly the four equations that we get for the various choices 
of $\bm a,\bm b\in\{\bm1,\bm2\}$.
Setting $\bm a=\bm b=\bm1$ in \eqref{eq:Qsato} and using \eqref{eq:inverse11},
we get
\begin{equation}\label{eq:Qsato11}
\begin{split}
\frac{\partial Q^+_{\bm1\bm 1}}{\partial t_j^{(a)}}
(\ul m,x,\ul{\bm t},\partial)
=
-
\Big(
Q^+_{\bm1\bm 1}(\ul m,x,\ul{\bm t},\partial)
\circ
E_{aa}
\partial^{j}
 \circ
Q^+_{\bm1\bm 1}(\ul m,x,\ul{\bm t},\partial)^{-1}
\Big)_-
\circ Q^+_{\bm1\bm 1}(\ul m,x,\ul{\bm t},\partial)
\,.
\end{split}
\end{equation}
Setting $\bm a=\bm1,\,\bm b=\bm2$ in \eqref{eq:Qsato} we get
\begin{align*}
\frac{\partial Q^+_{\bm1\bm 1}}{\partial t_j^{(a)}}
(\ul m,x,\ul{\bm t},\partial)
&\circ
Q^-_{\bm2\bm 1}(\ul m,x,\ul{\bm t},\partial)^*
+
\Big(
Q^+_{\bm1\bm 1}(\ul m,x,\ul{\bm t},\partial)
\circ
E_{aa}
\partial^{j}
 \circ
Q^-_{\bm2\bm 1}(\ul m,x,\ul{\bm t},\partial)^*
\Big)_- \\
& +
\frac{\partial Q^+_{\bm1\bm 2;1}}{\partial t_j^{(a)}}
(\ul m,x,\ul{\bm t})
\partial^{-1}
=0
\,.
\end{align*}
Using \eqref{eq:inverse12} and \eqref{eq:Qsato11},
this equation becomes
\begin{align*}
\frac{\partial Q^+_{\bm1\bm 2;1}}{\partial t_j^{(a)}}
(\ul m,&x,\ul{\bm t})
=
\Big(
Q^+_{\bm1\bm 1}(\ul m,x,\ul{\bm t},\partial)
\circ
E_{aa}
\partial^{j}
 \circ
Q^+_{\bm1\bm 1}(\ul m,x,\ul{\bm t},\partial)^{-1}
Q^+_{\bm1\bm 2;1}(\ul m,x,\ul{\bm t})
\partial^{-1}
\Big)_- \partial \\
& -
\Big(
Q^+_{\bm1\bm 1}(\ul m,x,\ul{\bm t},\partial)
\circ
E_{aa}
\partial^{j}
 \circ
Q^+_{\bm1\bm 1}(\ul m,x,\ul{\bm t},\partial)^{-1}
\Big)_-
\circ 
Q^+_{\bm1\bm 2;1}(\ul m,x,\ul{\bm t}) \\
& =
\Big(
Q^+_{\bm1\bm 1}(\ul m,x,\ul{\bm t},\partial)
\circ
E_{aa}
\partial^{j}
 \circ
Q^+_{\bm1\bm 1}(\ul m,x,\ul{\bm t},\partial)^{-1}
\Big)_+
\circ 
Q^+_{\bm1\bm 2;1}(\ul m,x,\ul{\bm t}) \\
& -
\Big(
Q^+_{\bm1\bm 1}(\ul m,x,\ul{\bm t},\partial)
\circ
E_{aa}
\partial^{j}
 \circ
Q^+_{\bm1\bm 1}(\ul m,x,\ul{\bm t},\partial)^{-1}
\circ
Q^+_{\bm1\bm 2;1}(\ul m,x,\ul{\bm t})
\partial^{-1}
\Big)_+ \partial
\,,
\end{align*}
which is equivalent to
\begin{equation}\label{eq:Qsato12}
\frac{\partial Q^+_{\bm1\bm 2;1}}{\partial t_j^{(a)}}
(\ul m,x,\ul{\bm t})
=
\Big(
Q^+_{\bm1\bm 1}(\ul m,x,\ul{\bm t},\partial)
\circ
E_{aa}
\partial^{j}
 \circ
Q^+_{\bm1\bm 1}(\ul m,x,\ul{\bm t},\partial)^{-1}
\Big)_+
\big(
Q^+_{\bm1\bm 2;1}(\ul m,x,\ul{\bm t}) 
\big) 
\,.
\end{equation}
Note that in the RHS of \eqref{eq:Qsato12}
the differential operator is applied to $Q^+_{\bm1\bm2;1}$.
Next, setting $\bm a=\bm2$, $\bm b=\bm1$ in \eqref{eq:Qsato}, we get
\begin{align*}
\frac{\partial Q^+_{\bm2\bm 1}}{\partial t_j^{(a)}}
(\ul m,x,\ul{\bm t},\partial)
&\circ
Q^-_{\bm1\bm 1}(\ul m,x,\ul{\bm t},\partial)^*
+
\Big(
Q^+_{\bm2\bm 1}(\ul m,x,\ul{\bm t},\partial)
\circ
E_{aa}
\partial^{j}
 \circ
Q^-_{\bm1\bm 1}(\ul m,x,\ul{\bm t},\partial)^*
\Big)_-
=0
\,,
\end{align*}
which, by \eqref{eq:inverse11}, becomes
\begin{equation}\label{eq:Qsato21a}
\begin{split}
\frac{\partial Q^+_{\bm2\bm 1}}{\partial t_j^{(a)}}
(\ul m,x,\ul{\bm t},\partial)
=
- \Big(
Q^+_{\bm2\bm 1}(\ul m,x,\ul{\bm t},\partial)
\circ
E_{aa}
\partial^{j}
 \circ
Q^+_{\bm1\bm 1}(\ul m,x,\ul{\bm t},\partial)^{-1}
\Big)_-
\circ
Q^+_{\bm1\bm 1}(\ul m,x,\ul{\bm t},\partial)
\,.
\end{split}
\end{equation}
By equations \eqref{eq:inverse21} and \eqref{eq:Qsato11},
this equation can be rewritten as
\begin{align*}
\Big(
\frac{\partial Q^-_{\bm1\bm 2;1}}{\partial t_j^{(a)}}
&(\ul m,x,\ul{\bm t})
\Big)^T 
=
Q^-_{\bm1\bm 2;1}(\ul m,x,\ul{\bm t})^T
\Big(
Q^+_{\bm1\bm 1}(\ul m,x,\ul{\bm t},\partial)
\circ
E_{aa}
\partial^{j}
 \circ
Q^+_{\bm1\bm 1}(\ul m,x,\ul{\bm t},\partial)^{-1}
\Big)_- \\
& - 
\partial\circ
\Big(
\partial^{-1}\circ
Q^-_{\bm1\bm 2;1}(\ul m,x,\ul{\bm t})^T
Q^+_{\bm1\bm 1}(\ul m,x,\ul{\bm t},\partial)
\circ
E_{aa}
\partial^{j}
 \circ
Q^+_{\bm1\bm 1}(\ul m,x,\ul{\bm t},\partial)^{-1}
\Big)_- \\
& =
-
Q^-_{\bm1\bm 2;1}(\ul m,x,\ul{\bm t})^T
\Big(
Q^+_{\bm1\bm 1}(\ul m,x,\ul{\bm t},\partial)
\circ
E_{aa}
\partial^{j}
\circ
Q^+_{\bm1\bm 1}(\ul m,x,\ul{\bm t},\partial)^{-1}
\Big)_+ \\
& +
\partial\circ
\Big(
\partial^{-1}\circ
Q^-_{\bm1\bm 2;1}(\ul m,x,\ul{\bm t})^T
Q^+_{\bm1\bm 1}(\ul m,x,\ul{\bm t},\partial)
\circ
E_{aa}
\partial^{j}
 \circ
Q^+_{\bm1\bm 1}(\ul m,x,\ul{\bm t},\partial)^{-1}
\Big)_+
\,,
\end{align*}
or, taking adjoint of both sides,
\begin{align*}
\frac{\partial Q^-_{\bm1\bm 2;1}}{\partial t_j^{(a)}}
&(\ul m,x,\ul{\bm t})
=
-
\Big(
Q^+_{\bm1\bm 1}(\ul m,x,\ul{\bm t},\partial)
\circ
E_{aa}
\partial^{j}
 \circ
Q^+_{\bm1\bm 1}(\ul m,x,\ul{\bm t},\partial)^{-1}
\Big)_+^*
\circ
Q^-_{\bm1\bm 2;1}(\ul m,x,\ul{\bm t})
\\
& +
\Big(
\big(
Q^+_{\bm1\bm 1}(\ul m,x,\ul{\bm t},\partial)
\circ
E_{aa}
\partial^{j}
 \circ
Q^+_{\bm1\bm 1}(\ul m,x,\ul{\bm t},\partial)^{-1}
\big)^*
\circ
Q^-_{\bm1\bm 2;1}(\ul m,x,\ul{\bm t})
\partial^{-1}
\Big)_+
\partial
\,,
\end{align*}
which is equivalent to
\begin{equation}\label{eq:Qsato21}
\frac{\partial Q^-_{\bm1\bm 2;1}}{\partial t_j^{(a)}}
(\ul m,x,\ul{\bm t})
=
-
\Big(
Q^+_{\bm1\bm 1}(\ul m,x,\ul{\bm t},\partial)
\circ
E_{aa}
\partial^{j}
 \circ
Q^+_{\bm1\bm 1}(\ul m,x,\ul{\bm t},\partial)^{-1}
\Big)_+^*
\big(Q^-_{\bm1\bm 2;1}(\ul m,x,\ul{\bm t})\big)
\,.
\end{equation}
Finally, we set $\bm a=\bm b=\bm2$ in \eqref{eq:Qsato} to get
\begin{align*}
\frac{\partial Q^+_{\bm2\bm 1}}{\partial t_j^{(a)}}
(\ul m,x,\ul{\bm t},\partial)
&\circ
Q^-_{\bm2\bm 1}(\ul m,x,\ul{\bm t},\partial)^*
+
\Big(
Q^+_{\bm2\bm 1}(\ul m,x,\ul{\bm t},\partial)
\circ
E_{aa}
\partial^{j}
 \circ
Q^-_{\bm2\bm 1}(\ul m,x,\ul{\bm t},\partial)^*
\Big)_- \\
& +
\frac{\partial Q^+_{\bm2\bm 2;1}}{\partial t_j^{(a)}}
(\ul m,x,\ul{\bm t})
\partial^{-1}
=0
\,.
\end{align*}
Using equations \eqref{eq:inverse21}, \eqref{eq:inverse12} and \eqref{eq:Qsato21a},
we get
\begin{align*}
& \frac{\partial Q^+_{\bm2\bm 2;1}}{\partial t_j^{(a)}}
(\ul m,x,\ul{\bm t}) \\
& = 
- \Big(
\partial^{-1}\circ
Q^-_{\bm1\bm 2;1}(\ul m,x,\ul{\bm t})^T
Q^+_{\bm1\bm 1}(\ul m,x,\ul{\bm t},\partial)
\circ
E_{aa}
\partial^{j} 
\circ
Q^+_{\bm1\bm 1}(\ul m,x,\ul{\bm t},\partial)^{-1}
\Big)_-
\circ
Q^+_{\bm1\bm 2;1}(\ul m,x,\ul{\bm t}) \\
& +
\Big(
\partial^{-1}\circ
Q^-_{\bm1\bm 2;1}(\ul m,x,\ul{\bm t})^T
Q^+_{\bm1\bm 1}(\ul m,x,\ul{\bm t},\partial)
\circ
E_{aa}
\partial^{j}
 \circ
Q^+_{\bm1\bm 1}(\ul m,x,\ul{\bm t},\partial)^{-1}
\circ
Q^+_{\bm1\bm 2;1}(\ul m,x,\ul{\bm t})
\partial^{-1}
\Big)_- 
\partial \\
& = 
\Big(
\partial^{-1}\circ
Q^-_{\bm1\bm 2;1}(\ul m,x,\ul{\bm t})^T
Q^+_{\bm1\bm 1}(\ul m,x,\ul{\bm t},\partial)
\circ
E_{aa}
\partial^{j}
 \circ
Q^+_{\bm1\bm 1}(\ul m,x,\ul{\bm t},\partial)^{-1}
\Big)_+
\circ
Q^+_{\bm1\bm 2;1}(\ul m,x,\ul{\bm t}) \\
& -
\Big(
\partial^{-1}\circ
Q^-_{\bm1\bm 2;1}(\ul m,x,\ul{\bm t})^T
Q^+_{\bm1\bm 1}(\ul m,x,\ul{\bm t},\partial)
\circ
E_{aa}
\partial^{j}
 \circ
Q^+_{\bm1\bm 1}(\ul m,x,\ul{\bm t},\partial)^{-1}
\circ
Q^+_{\bm1\bm 2;1}(\ul m,x,\ul{\bm t})
\partial^{-1}
\Big)_+
\partial
\,,
\end{align*}
which is equivalent to
\begin{equation}\label{eq:Qsato22}
\begin{split}
& \frac{\partial Q^+_{\bm2\bm 2;1}}{\partial t_j^{(a)}}
(\ul m,x,\ul{\bm t})  \\
& = 
\Big(
\partial^{-1}\circ
Q^-_{\bm1\bm 2;1}(\ul m,x,\ul{\bm t})^T
Q^+_{\bm1\bm 1}(\ul m,x,\ul{\bm t},\partial)
\circ
E_{aa}
\partial^{j}
 \circ
Q^+_{\bm1\bm 1}(\ul m,x,\ul{\bm t},\partial)^{-1}
\Big)_+
\big(
Q^+_{\bm1\bm 2;1}(\ul m,x,\ul{\bm t})
\big)
\,.
\end{split}
\end{equation}

Next, we introduce the reduced Lax operators (cf. \eqref{eq:La}).
For $\ul m\in\mb Z^r$  such that $|\underline m|=k$,  and $1\leq a\leq r_1$, let 
\begin{equation}\label{eq:mcLa}
\mc L_a(\ul m,x,\ul{\bm t},\partial)
=
Q^+_{\bm1\bm1}(\ul m,x,\ul{\bm t},\partial)\circ E_{aa}\partial\circ 
Q^{+}_{\bm1\bm1}(\ul m,x,\ul{\bm t},\partial)^{-1}
\,\in\Mat_{r_1\times r_1}\mc F((\partial^{-1}))
\,,
\end{equation}
if $\tau^{\ul m}(x,\ul{\bm t})\neq0$,
and $\mc L_a=0$ otherwise.
Clearly, 
\begin{equation}\label{eq:QLcommute}
\mc L_a(\ul m,x,\ul{\bm t},\partial)\circ \mc L_b(\ul m,x,\ul{\bm t},\partial)=0
\,\text{ if }\,
1\leq a\neq b\leq r_1\,,
\end{equation}
and they satisfy the ``reduced analogue'' 
of the constraint \eqref{eq:Lconstraint},
obtained by setting $\ell=1$, $\ul m^\prime=\ul m^{\prime\prime}\,(=\ul m)$,
$\ul{\bm t}^\prime=\ul{\bm t}^{\prime\prime}\,(=\ul{\bm t})$
and $\bm a=\bm b=\bm1$ in \eqref{eq:kdvQ-c}:
\begin{equation}\label{eq:Qconstraint}
\begin{split}
\Big(
\sum_{a=1}^{r_1}
\mc L_a^{p_1}(\ul m,x,\ul{\bm t},\partial) 
\Big)_- 
=
\sum_{a=r_1+1}^{r}
\sum_{i=1}^{p_a}
(-1)^{p_a+i}
Q^+_{\bm1\bm 2;i}(\ul m,x,\ul{\bm t})
E_{aa} 
\partial^{-1}\circ
Q^-_{\bm1\bm 2;p_a-i+1}(\ul m,x,\ul{\bm t})^T
\,,
\end{split}
\end{equation}
where we used the expansions \eqref{eq:daniele3}.
We can also rewrite all equations \eqref{eq:Qsato11}, \eqref{eq:Qsato12}, \eqref{eq:Qsato21}
and \eqref{eq:Qsato22} in terms of the operators $\mc L_a$,
to get the ``reduced analogue'' of the Lax equation \eqref{eq:Lax}.
Equation \eqref{eq:Qsato11} gives the Lax equation
\begin{equation}\label{eq:QLax11}
\frac{\partial \mc L_a}{\partial t_j^{(b)}}(\ul m,x,\ul{\bm t},\partial)
=
[(\mc L_b(\ul m,x,\ul{\bm t},\partial)^j)_+, \mc L_a(\ul m,x,\ul{\bm t},\partial)]
\,\,,\quad 1\le a,b\le r_1,\,j\in\mb Z_{>0}
\,.
\end{equation}
Equations \eqref{eq:Qsato12} and \eqref{eq:Qsato21} become, respectively,
\begin{equation}\label{eq:QLax12}
\frac{\partial Q^+_{\bm1\bm 2;1}}{\partial t_j^{(a)}}
(\ul m,x,\ul{\bm t})
=
\big(
\mc L_a^j(\ul m,x,\ul{\bm t},\partial)
\big)_+
\big(
Q^+_{\bm1\bm 2;1}(\ul m,x,\ul{\bm t}) 
\big) 
\,,
\end{equation}
and
\begin{equation}\label{eq:QLax21}
\frac{\partial Q^-_{\bm1\bm 2;1}}{\partial t_j^{(a)}}
(\ul m,x,\ul{\bm t})
=
-
\big(
\mc L_a^j(\ul m,x,\ul{\bm t},\partial)
\big)_+^*
\big(Q^-_{\bm1\bm 2;1}(\ul m,x,\ul{\bm t})\big)
\,.
\end{equation}
These equations mean that $Q^+_{\bm1\bm2;1}$ is a matrix eigenfunction for ${\mc L}_a$,
while $Q^-_{\bm11\bm2;1}$ is an adjoint eigenfunction for ${\mc L}_a$.
Equation \eqref{eq:Qsato22} becomes
\begin{equation}\label{eq:QLax22}
\frac{\partial Q^+_{\bm2\bm 2;1}}{\partial t_j^{(a)}}
(\ul m,x,\ul{\bm t}) 
= 
\Big(
\partial^{-1}\circ
Q^-_{\bm1\bm 2;1}(\ul m,x,\ul{\bm t})^T
\mc L_a^j(\ul m,x,\ul{\bm t},\partial)
\Big)_+
\big(
Q^+_{\bm1\bm 2;1}(\ul m,x,\ul{\bm t})
\big)
\,.
\end{equation}

\section{
The constrained Lax operator and solution to the Lax equation
}
\label{sec:5}

Let
\begin{equation}\label{eq:WLax}
\widetilde{\mc L}(\ul m,x,\ul{\bm t},\partial)
:=
\sum_{a=1}^{r_1}
\mc L_a(\ul m,x,\ul{\bm t},\partial)^{p_1}
\,\in\Mat_{r_1\times r_1}\mc F((\partial^{-1}))
\,,
\end{equation}
where, as usual, $\mc F$ denotes the differential field containing all functions
(in the space and time variables) that we consider.
In the present section, we use the constraint equation \eqref{eq:Qconstraint}
to show that the operator \eqref{eq:WLax} has the same form
as the $\mc W$-algebra Lax operator for $\mc W(\mf{gl}_N,\ul p)$
defined in \cite{DSKV16b} (which we shall review in the next Section \ref{sec:6}).
Equation \eqref{eq:Qconstraint} can be rewritten, in terms of the operator \eqref{eq:WLax}, 
as
\begin{equation}\label{eq:Qconstraint2}
\widetilde{\mc L}(\ul m,x,\ul{\bm t},\partial)_-
=
-\Res_z
Q^+_{\bm1\bm 2}(\ul m,x,\ul{\bm t},z)
\Big(\!\!\!
\sum_{a=r_1+1}^{r}
z^{p_a} E_{aa}
\Big)
\partial^{-1}\circ
Q^-_{\bm1\bm 2}(\ul m,x,\ul{\bm t},-z)^T
\,.
\end{equation}
In order to rewrite the RHS of \eqref{eq:Qconstraint2},
we shall use the following alternative version of Lemma \ref{lem:fund2}:
\begin{lemma}\label{lem:fund2b}
For every 
$A(x,\partial)\in\Mat_{h\times k}\mc F[[\partial]]\partial^{-m},
B(x,\partial)\in\Mat_{h'\times k}\mc F[[\partial]]\partial^{-n}$,
with symbols $A(x,z),B(x,z)\in\Mat\mc F((\partial))$, 
we have
\begin{equation}\label{eq:fund1b}
\Res_z A(x,z)\partial^{-1}\circ B(x,-z)^T
=
\sum_{i,j=0}^{m+n-1}
\Big(
\frac{(-x)^i}{i!}
A(x,\partial)\partial^{i+j}
\circ
B(x,\partial)^*
\circ
\frac{x^j}{j!}
\Big)_-
\,.
\end{equation}
\end{lemma}
\begin{proof}
Equation \eqref{eq:fund1} holds for every 
matrix pseudodifferential operators 
$A(x,\partial),B(x,\partial)$.
In particular, for ``Laurent'' differential operators
$A(x,\partial),B(x,\partial)\in\Mat\mc F[\partial,\partial^{-1}]$.
Let then set
\begin{equation}\label{eq:Atilde}
\widetilde{A}(x,z)
=
A(x,z)e^{xz}
\,\,,\,\,\,\,
\widetilde{B}(x,z)
=
B(x,z)e^{xz}
\,,
\end{equation}
so that
\begin{equation}\label{eq:Btilde}
\widetilde{B}(x,-z)^T
=
e^{-xz}B(x,-z)^T
\,.
\end{equation}
Note that, for pseudodifferential operators $A$ and $B$, equations \eqref{eq:Atilde} and \eqref{eq:Btilde}
do not make sense, as they involve diverging series.
For ``Laurent'' differential operators there are no divergence problems,
hence $\widetilde{A}(x,z),\widetilde{B}(x,z)$ are well defined elements in $\Mat\mc F((z))$.
By inverting formulas \eqref{eq:Atilde} and
taking the corresponding (infinite order) pseudodifferential operators, we get
\begin{equation}\label{eq:A1}
A(x,\partial)
=
\widetilde{A}(x,z)e^{-xz}\big|_{z=\partial}
=
\sum_{i=0}^\infty
\frac{(-x)^i}{i!}
\widetilde{A}(x,\partial)\partial^i
\,\in\Mat_{h\times k}\mc F((\partial))
\,,
\end{equation}
and
\begin{equation}\label{eq:A2}
B(x,\partial)^*
=
\big|_{z=\partial}
\circ e^{xz}\widetilde{B}(x,-z)^T
=
\sum_{j=0}^\infty
\partial^j\circ \widetilde{B}(x,\partial)^*
\frac{x^j}{j!}
\,\in\Mat_{k\times h'}\mc F((\partial))
\,.
\end{equation}
By \eqref{eq:Atilde}, \eqref{eq:Btilde}, \eqref{eq:A1} and \eqref{eq:A2}, equation \eqref{eq:fund1} gives
\begin{equation}\label{eq:fund1c}
\Res_z \widetilde{A}(x,z)\partial^{-1}\circ \widetilde B(x,-z)^T
=
\sum_{i,j=0}^\infty
\Big(
\frac{(-x)^i}{i!}
\widetilde{A}(x,\partial)
\partial^{i+j}
\circ \widetilde{B}(x,\partial)^*
\circ
\frac{x^j}{j!}
\Big)_-
\,.
\end{equation}
So far, we only proved equation \eqref{eq:fund1c} 
for $\widetilde{A}(x,z)$, $\widetilde{B}(x,z)$
of the form \eqref{eq:Atilde}, with $A(x,z),B(x,z)$ Laurent polynomials in $z$.
On the other hand, 
if the powers of $z$ in $\widetilde{A}(x,z)$ and $\widetilde{B}(x,z)$ are bounded from below 
by $-m$ and $-n$ respectively,
only the coefficients of $z^{\leq n-1}$ in $\widetilde{A}(x,z)$,
and of $z^{\leq m-1}$ in $\widetilde{B}(x,z)$,
give possibly non zero contribution to either side of equation \eqref{eq:fund1c}.
It follows that equation \eqref{eq:fund1c} actually holds
for every Laurent polynomials $\widetilde{A}(x,z)$, $\widetilde{B}(x,z)$,
and therefore for every Laurent series 
$\widetilde{A}(x,z)$, $\widetilde{B}(x,z)\,\in\Mat\mc F((z))$.
Moreover, 
only the terms with $i+j\leq m+n-1$ give a non-zero contribution to the RHS of \eqref{eq:fund1c}.
The claim follows.
\end{proof}

As a consequence of equation \eqref{eq:Qconstraint2} and Lemma \ref{lem:fund2b},
we can write the operator $\widetilde{\mc L}(\ul m,x,\ul{\bm t},\partial)$, defined in \eqref{eq:WLax},
in a form similar to \cite[Eq.(5.16)]{DSKV16b}.
For this, introduce the matrices
\begin{equation}\label{eq:W12}
\widetilde{W}_{\bm1\bm2}(\ul m,x,\ul{\bm t},\partial)
=\big(\widetilde{W}_{ab}(\ul m,x,\ul{\bm t},\partial)\big)_{\substack{1\leq a\leq r_1 \\ r_1< b\leq r}}
\,,\,\,
\widetilde{W}_{\bm2\bm1}(\ul m,x,\ul{\bm t},\partial)
=\big(\widetilde{W}_{ab}(\ul m,x,\ul{\bm t},\partial)\big)_{\substack{r_1<a\leq r \\ 1\leq b\leq r_1}}
\,,
\end{equation}
where
\begin{equation}\label{eq:W+ab}
\widetilde{W}_{ab}(\ul m,x,\ul{\bm t},\partial)
=
\sum_{j=0}^{p_b-1}
\sum_{i=0}^{j}
\frac{(-x)^i}{i!}
Q^+_{ab;j-i+1} (\ul m,x,\ul{\bm t})
\partial^{j}
\,\,\text{ if }\,\,
1\leq a\leq r_1 ,\, r_1< b\leq r
\,,
\end{equation}
and
\begin{equation}\label{eq:W-ab}
\widetilde{W}_{ab}(\ul m,x,\ul{\bm t},\partial)
=
-\sum_{j=0}^{p_a-1}
\sum_{i=0}^{j}
(-\partial)^{j}
\circ
Q^-_{ba;j-i+1}(\ul m,x,\ul{\bm t})
\frac{(-x)^i}{i!}
\,\,\text{ if }\,\,
r_1< a\leq r ,\, 1\leq b\leq r_1
\,.
\end{equation}
Notice that these are differential operators of order bounded from above by $\min\{p_a,p_b\}-1$,
as in \cite[Eq.(5.5)]{DSKV16b}.
By the definition \eqref{eq:WLax} of the operator $\widetilde{\mc L}(\ul m,x,\ul{\bm t},\partial)$,
equation \eqref{eq:mcLa} and the expansion \eqref{eq:daniele3},
we have
\begin{align*}
& \widetilde{\mc L}(\ul m,x,\ul{\bm t},\partial)
=
Q^+_{\bm1\bm1}(\ul m,x,\ul{\bm t},\partial)\partial^{p_1}\circ 
Q^{+}_{\bm1\bm1}(\ul m,x,\ul{\bm t},\partial)^{-1} \\
& =
\id_{r_1}\partial^{p_1}
+\text{ terms of oder } < p_1
\,.
\end{align*}
Moreover, by \eqref{eq:W+ab} and \eqref{eq:W-ab},
the matrix 
$$
\widetilde{W}_{\bm1\bm2}(\ul m,x,\ul{\bm t},\partial)
\big(\!\!\!
\sum_{a=r_1+1}^r E_{aa} \partial^{-p_a}
\big)
\circ
\widetilde{W}_{\bm2\bm1}(\ul m,x,\ul{\bm t},\partial)
$$
is a pseudodifferential operator of order strictly less than $p_1$.
Set then
\begin{equation}\label{eq:W11}
\widetilde{W}_{\bm1\bm1}(\ul m,x,\ul{\bm t},\partial)
=
\widetilde{\mc L}(\ul m,x,\ul{\bm t},\partial)_+
-
\id_{r_1}\partial^{p_1}
+
\Big(
\widetilde{W}_{\bm1\bm2}(\ul m,x,\ul{\bm t},\partial)
\big(\!\!\!
\sum_{a=r_1+1}^r E_{aa} \partial^{-p_a}
\big)
\circ
\widetilde{W}_{\bm2\bm1}(\ul m,x,\ul{\bm t},\partial)
\Big)_+
\,,
\end{equation}
which is an $r_1\times r_1$-matrix differential operator of order 
bounded from above by $p_1-1$,
as in \cite[Eq.(5.5)]{DSKV16b}.
\begin{theorem}\label{thm:constraintL}
The operator $\widetilde{\mc L}(\ul m,x,\ul{\bm t},\partial)$ defined in \eqref{eq:WLax}
has the following form
\begin{equation}\label{eq:WL2}
\widetilde{\mc L}(\ul m,x,\ul{\bm t},\partial)
=
\id_{r_1}\partial^{p_1}
+
\widetilde{W}_{\bm1\bm1}(\ul m,x,\ul{\bm t},\partial)
-
\widetilde{W}_{\bm1\bm2}(\ul m,x,\ul{\bm t},\partial)
\circ
\big(\!\!\!
\sum_{a=r_1+1}^r E_{aa} \partial^{-p_a}
\big)
\circ
\widetilde{W}_{\bm2\bm1}(\ul m,x,\ul{\bm t},\partial)
\,,
\end{equation}
where the matrices 
$\widetilde{W}_{\bm1\bm2}(\ul m,x,\ul{\bm t},\partial)\in\Mat_{r_1\times(r-r_1)}\mc F[\partial]$,
$\widetilde{W}_{\bm2\bm1}(\ul m,x,\ul{\bm t},\partial)\in\Mat_{(r-r_1)\times r_1}\mc F[\partial]$,
and
$\widetilde{W}_{\bm1\bm1}(\ul m,x,\ul{\bm t},\partial)\in\Mat_{r_1\times r_1}\mc F[\partial]$,
are as in \eqref{eq:W12}--\eqref{eq:W11}.
\end{theorem}
\begin{proof}
In order to prove equation \eqref{eq:WL2},
we would like to apply Lemma \ref{lem:fund2b} to rewrite equation \eqref{eq:Qconstraint2}.
Note, though, that equation \eqref{eq:fund1b} cannot be applied directly 
in the RHS of \eqref{eq:Qconstraint2},
since $Q^{\pm}_{\bm1\bm2}(\ul m,x,\ul{\bm t},z)$ are formal power series in $z^{-1}$,
and not Laurent series in $z$, as required by Lemma \ref{lem:fund2b}.
In order to apply Lemma \ref{lem:fund2b} we therefore replace $z$ by $z^{-1}$
by using the obvious identity
$$
\Res_z f(z)=\Res_z z^{-2}f(z^{-1})
\,.
$$
Hence, equation \eqref{eq:Qconstraint2} becomes
$$
\widetilde{\mc L}(\ul m,x,\ul{\bm t},\partial)_-
=
-\Res_z
Q^+_{\bm1\bm 2}(\ul m,x,\ul{\bm t},z^{-1})
\Big(\!\!\!
\sum_{a=r_1+1}^{r}
z^{-p_a-2} E_{aa}
\Big)
\partial^{-1}\circ
Q^-_{\bm1\bm 2}(\ul m,x,\ul{\bm t},-z^{-1})^T
\,.
$$
We then apply Lemma \ref{lem:fund2b} with ($r_1+1\leq a\leq r$)
$$
A(x,z)
=
Q^+_{\bm1\bm 2}(\ul m,x,\ul{\bm t},z^{-1})
z^{-p_a-2} E_{aa}
\,\in\Mat_{r_1\times(r-r_1)}\mc V[[z]]z^{-p_a-1}
$$
and
$$
B(x,z)
=
Q^-_{\bm1\bm 2}(\ul m,x,\ul{\bm t},z^{-1})
\,\in\Mat_{r_1\times(r-r_1)}\mc V[[z]]z
\,.
$$
(Here we are using expansions \eqref{eq:daniele3}.)
As a result, we get
\begin{equation}\label{eq:Qconstraint3}
\widetilde{\mc L}(\ul m,x,\ul{\bm t},\partial)_-
=
-\sum_{a=r_1+1}^{r}
\sum_{i,j=0}^{p_a-1}
\Big(
\frac{(-x)^i}{i!}
Q^+_{\bm1\bm 2}(\ul m,x,\ul{\bm t},\partial^{-1})
\circ
E_{aa}
\partial^{i+j-p_a-2} 
\circ
Q^-_{\bm1\bm 2}(\ul m,x,\ul{\bm t},\partial^{-1})^*
\frac{x^j}{j!}
\Big)_-
\,.
\end{equation}
By expansions \eqref{eq:daniele3}, equation \eqref{eq:Qconstraint3} becomes
\begin{equation}\label{eq:Qconstraint4}
\begin{split}
& \widetilde{\mc L}(\ul m,x,\ul{\bm t},\partial)_- \\
& =
-
\sum_{a=r_1+1}^{r}
\sum_{i,j=0}^{p_a-1}
\sum_{h,k=1}^{p_a}
\Big(
\frac{(-x)^i}{i!}
Q^+_{\bm1\bm2;h} (\ul m,x,\ul{\bm t})
E_{aa}
\partial^{i+j+h+k-p_a-2} 
\circ
(-1)^{k}
Q^-_{\bm1\bm2;k}(\ul m,x,\ul{\bm t})^T
\frac{x^j}{j!}
\Big)_- 
\,.
\end{split}
\end{equation}
Recalling definition \eqref{eq:W12} of the matrices 
$\widetilde{W}_{\bm1\bm2}(\ul m,x,\ul{\bm t},\partial)$ and $\widetilde{W}_{\bm2\bm1}(\ul m,x,\ul{\bm t},\partial)$, 
we can rewrite equation \eqref{eq:Qconstraint4} as
\begin{equation}\label{eq:Qconstraint5}
\widetilde{\mc L}(\ul m,x,\ul{\bm t},\partial)_-
=
-\Big(
\widetilde{W}_{\bm1\bm2}(\ul m,x,\ul{\bm t},\partial)
\big(\!\!\!
\sum_{a=r_1+1}^r E_{aa} \partial^{-p_a}
\big)
\circ
\widetilde{W}_{\bm2\bm1}(\ul m,x,\ul{\bm t},\partial)
\Big)_- 
\,.
\end{equation}
Combining \eqref{eq:Qconstraint5} and \eqref{eq:W11}, we finally get equation \eqref{eq:WL2},
completing the proof.
\end{proof}

Next, we use the evolution equation \eqref{eq:QLax11} 
for the operators $\mc L_a(\ul m,x,\ul{\bm t},\partial)$, $a=1,\dots,r_1$,
to derive an evolution equation for the Lax operator \eqref{eq:WLax}.
For this, we need to set
\begin{equation}\label{eq:times}
t_j^{(a)}=t_j 
\,\text{ for all }\,
a=1,\dots,r_1
\,.
\end{equation}
In other words, we let $\bm t=(t_j)_{j\in\mb Z_{>0}}$, and
\begin{equation}\label{eq:Lt}
\widetilde{\mc L}(\ul m,x,{\bm t},\partial)
:=
\widetilde{\mc L}(\ul m,x,\ul{\bm t},\partial)
\big|_{\bm t^{(a)}=\bm t\,\forall a=1,\dots,r_1}
\,.
\end{equation}
\begin{theorem}\label{thm:lax}
The operator $\widetilde{\mc L}(\ul m,x,{\bm t},\partial)$ defined by \eqref{eq:Lt} and \eqref{eq:WLax}
evolves according to the Lax equations
\begin{equation}\label{eq:Laxeq-tilde}
\frac{\partial\widetilde{\mc L}}{\partial t_j}(\ul m,x,{\bm t},\partial)
=
\big[\big(\widetilde{\mc L}(\ul m,x,{\bm t},\partial)^{\frac{j}{p_1}}\big)_+,\widetilde{\mc L}(\ul m,x,{\bm t},\partial)]
\,\,,\,\,\,\,
j\in\mb Z_{\geq1}
\,.
\end{equation}
\end{theorem}
\begin{proof}
By equations \eqref{eq:WLax} and \eqref{eq:QLcommute}, we immediately have
\begin{equation}\label{eq:rome1}
\widetilde{\mc L}(\ul m,x,\ul{\bm t},\partial)^{\frac{j}{p_1}}
=
\sum_{a=1}^{r_1}
\mc L_a(\ul m,x,\ul{\bm t},\partial)^{j}
\,\,,\,\,\,\,j\in\mb Z_{\geq1}
\,.
\end{equation}
Moreover, since both $\frac{\partial}{\partial t_j^{(b)}}$ and the adjoint action of 
$\big(\mc L_b(\ul m,x,\ul{\bm t},\partial)^j\big)_+$ are derivations of the product
of pseudodifferential operators, we immediately get from \eqref{eq:QLax11} that
\begin{equation}\label{eq:rome3}
\frac{\partial (\mc L_a)^{n}}{\partial t_j^{(b)}}(\ul m,x,\ul{\bm t},\partial)
=
[(\mc L_b(\ul m,x,\ul{\bm t},\partial)^j)_+, (\mc L_a)^n(\ul m,x,\ul{\bm t},\partial)]
\,\,,\,\,\,\,n\in\mb Z_{\geq0}
,\,j\in\mb Z_{\geq1}\,.
\end{equation}
Hence, 
setting \eqref{eq:times}, we can use equations \eqref{eq:rome1} and \eqref{eq:rome3}
to get
\begin{align*}
& \frac{\partial \widetilde{\mc L}}{\partial t_j}(\ul m,x,{\bm t},\partial)
=
\sum_{b=1}^{r_1}
\frac{\partial \widetilde{\mc L}}{\partial t_j^{(b)}}(\ul m,x,\ul{\bm t},\partial)
=
\sum_{a,b=1}^{r_1}
\frac{\partial (\mc L_a)^{p_1}}{\partial t_j^{(b)}}(\ul m,x,\ul{\bm t},\partial) \\
& =
\sum_{a,b=1}^{r_1}
[(\mc L_b(\ul m,x,\ul{\bm t},\partial)^j)_+, (\mc L_a)^{p_1}(\ul m,x,\ul{\bm t},\partial)] \\
& =
[(\widetilde{\mc L}(\ul m,x,\ul{\bm t},\partial)^{\frac{j}{p_1}})_+, \widetilde{\mc L}(\ul m,x,\ul{\bm t},\partial)]
\,.
\end{align*}
\end{proof}

\section{
The $\mc W$-algebra Lax operator for $\mc W(\mf{gl}_N,\ul p)$
and the associated integrable Hamiltonian hierarchy
}
\label{sec:6}

In the present section we briefly review the theory of classical $\mc W$-algebras,
the construction of the Lax operator $\mc L(\partial)$ for the $\mc W$-algebra 
$\mc W(\mf{gl}_N,\ul p)$,
and the associated integrable hierarchy of Hamiltonian equations in Lax form.
The interested reader is referred to \cite{BDSK09,DSKV13,DSKV16a,DSKV16b,DSKV16c,DSKV18}.

\subsection{Poisson vertex algebras and integrable Hamiltonian equations}
\label{sec:6.1}

Recall from \cite{BDSK09} that a \emph{Poisson vertex algebra} (PVA)
is a differential algebra,  i.e. a unital commutative associative algebra with a derivation $\partial$,
endowed with a $\lambda$-\emph{bracket},
i.e. a bilinear (over $\mb C$) map $\{\cdot\,_\lambda\,\cdot\}:\,\mc V\times\mc V\to\mc V[\lambda]$, 
satisfying the following axioms ($a,b,c\in\mc V$):
\begin{enumerate}[(i)]
\item
sesquilinearity:
$\{\partial a_\lambda b\}=-\lambda\{a_\lambda b\}$,
$\{a_\lambda\partial b\}=(\lambda+\partial)\{a_\lambda b\}$;
\item
skewsymmetry:
$\{b_\lambda a\}=-\{a_{-\lambda-\partial} b\}$,
where $\partial$ in the RHS is moved to the left and acts on the coefficients;
\item
Jacobi identity:
$\{a_\lambda \{b_\mu c\}\}-\{b_\mu\{a_\lambda c\}\}
=\{\{a_\lambda b\}_{\lambda+\mu}c\}$.
\item
left Leibniz rule:
$\{a_\lambda bc\}=\{a_\lambda b\}c+\{a_\lambda c\}b$.
\end{enumerate}
Applying skewsymmetry to the left Leibniz rule
we get
\begin{enumerate}[(i)]
\setcounter{enumi}{4}
\item
right Leibniz rule:
$\{ab_\lambda c\}=\{a_{\lambda+\partial} c\}_\to b+\{b_{\lambda+\partial} c\}_\to a$,
where $\to$ means that $\partial$ is moved to the right.
\end{enumerate}

For example, given a Lie algebra $\mf g$ with a symmetric invariant bilinear form $(\cdot\,|\,\cdot)$,
we have the corresponding \emph{classical affine} PVA.
It is defined as the algebra 
$\mc V(\mf g)=S(\mb C[\partial]\mf g)$
of differential polynomials over $\mf g$,
with the PVA $\lambda$-bracket given by
\begin{equation}\label{lambda}
\{a_\lambda b\}=[a,b]+(a| b)\lambda
\quad \text{ for }\quad a,b\in\mf g\,,
\end{equation}
and extended to $\mc V(\mf g)$ by the sesquilinearity axiom and the Leibniz rules.

As usual, we denote by $\tint:\,\mc V\to\mc V/\partial\mc V$ the canonical quotient map
of vector spaces.
Recall that, if $\mc V$ is a Poisson vertex algebra,
then $\mc V/\partial\mc V$ carries a well defined Lie algebra structure given by
$\{\tint f,\tint g\}=\tint\{f_\lambda g\}|_{\lambda=0}$,
and we have a representation of the Lie algebra $\mc V/\partial\mc V$ on $\mc V$
given by $\{\tint f,g\}=\{f_\lambda g\}|_{\lambda=0}$.
A \emph{Hamiltonian equation} on $\mc V$ associated to a \emph{Hamiltonian functional} 
$\tint h\in\mc V/\partial\mc V$ is the evolution equation 
\begin{equation}\label{ham-eq}
\frac{du}{dt}=\{\tint h,u\}\,\,, \,\,\,\, u\in\mc V\,.
\end{equation}
The minimal requirement for \emph{integrability} is to have an infinite collection
of linearly independent integrals of motion in involution:
$$
\tint h_0=\tint h,\,\tint h_1,\,\tint h_2,\,\dots\,
\,\text{ s.t. }\,\, \{\tint h_i,\tint h_j\}=0\,\,\text{ for all }\,\, i,j\geq0
\,.
$$
In this case, we have the \emph{integrable hierarchy} of Hamiltonian equations
\begin{equation}\label{eq:int-hier}
\frac{du}{dt_j}=\{\tint h_j,u\}\,\,, \,\,\,\, j\in\mb Z_{\geq0}
\,.
\end{equation}

\subsection{Classical $\mc W$-algebras}
\label{sec:6.2}

Let $\mf g$ be a reductive Lie algebra,
with a non-degenerate invariant symmetric bilinear form $(\cdot\,|\,\cdot)$,
and consider the classical affine PVA $\mc V(\mf g)$ with $\lambda$-bracket defined by \eqref{lambda}.
Given an $\mf{sl}_2$-tirple $(e,f,\chi)$ in $\mf g$,
such that $[e,f]=\chi$, $[\chi,e]=e$, $[\chi,f]=-f$,
we have the corresponding Dynkin grading of $\mf g$,
namely the $\ad\chi$-eigenspace decomposition
\begin{equation}\label{eq:grading}
\mf g=\bigoplus_{k\in\frac{1}{2}\mb Z}\mf g_{k}
\,\,,\,\,\,\,
\mf g_k=\big\{a\in\mf g\,\big|\,[\chi,a]=ka\big\}
\,.
\end{equation}
It is well known that this grading depends, up to conjugation,
only on the adjoint orbit of $f$.

For a subspace $\mf p\subset\mf g$,
we will denote by $\mc V(\mf p)$ the differential subalgebra $S(\mb C[\partial]\mf p)$
of $\mc V(\mf g)$.
Consider the differential subalgebra
$\mc V(\mf g_{\leq\frac12})$ of $\mc V(\mf g)$,
and denote by $\rho:\,\mc V(\mf g)\twoheadrightarrow\mc V(\mf g_{\leq\frac12})$,
the differential algebra homomorphism defined on generators by
\begin{equation}\label{rho}
\rho(a)=\pi_{\leq\frac12}(a)+(f| a),
\qquad a\in\mf g\,,
\end{equation}
where $\pi_{\leq\frac12}:\,\mf g\to\mf g_{\leq\frac12}\,\big(:=\bigoplus_{k\leq\frac12}\mf g_k\big)$ 
denotes the projection with kernel $\mf g_{\geq1}$.
The \emph{classical} $\mc W$-\emph{algebra} $\mc W(\mf g,f)$ is, by definition,
the differential algebra
\begin{equation}\label{20120511:eq2}
\mc W(\mf g,f)
=
\big\{w\in\mc V(\mf g_{\leq\frac12})\,\big|\,\rho\{a_\lambda w\}=0\,
\text{ for all }a\in\mf g_{\geq\frac12}\}\,,
\end{equation}
endowed with the following PVA $\lambda$-bracket \cite[Lemma 3.2]{DSKV13}
\begin{equation}\label{20120511:eq3}
\{v_\lambda w\}^{\mc W}=\rho\{v_\lambda w\},
\qquad v,w\in\mc W\,.
\end{equation}

We can describe explicitly the classical $\mc W$-algebra as an algebra of differential polynomials.
For this,
fix a subspace $U\subset\mf g_{\geq-\frac12}$ complementary to $[f,\mf g_{\geq\frac12}]$
and compatible with the grading \eqref{eq:grading}.
Obviously, 
the orthogonal complement of $[f,\mf g_{\geq\frac12}]$ in $\mf g_{\geq-\frac12}$ 
w.r.t. the bilinear form $(\cdot\,|\,\cdot)$ is $\mf g^f\subset\mf g_{\leq\frac12}$, the centralizer of $f$ in $\mf g$.
Hence, 
we have the ``dual'' direct sum decompositions \cite[Eq.(3.6)-(3.7)]{DSKV18}
\begin{equation}\label{eq:U}
\mf g_{\geq-\frac12}=[f,\mf g_{\geq\frac12}]\oplus U
\,\,,\,\,\,\,
\mf g_{\leq\frac12}=U^\perp\oplus\mf g^f
\,.
\end{equation}
As a consequence,
we have the decomposition in a direct sum of subspaces
\begin{equation}\label{eq:decomp}
\mc V(\mf g_{\leq\frac12})=\mc V(\mf g^f)\oplus\langle U^\perp\rangle\,,
\end{equation}
where $\langle U^\perp\rangle$
is the differential algebra ideal of $\mc V(\mf g_{\leq\frac12})$ generated by $U^\perp$.
Let $\pi_{\mf g^f}:\,\mc V(\mf g_{\leq\frac12})\twoheadrightarrow\mc V(\mf g^f)$
be the canonical quotient map, with kernel $\langle U^\perp\rangle$.

\begin{theorem}[{\cite[Cor.4.1]{DSKV16c}}]
\label{thm:structure-W}
The map $\pi_{\mf g^f}$ restricts to a differential algebra isomorphism
$$
\pi:=\pi_{\mf g^f}|_{\mc W(\mf g,f)}:\,\mc W(\mf g,f)\,\stackrel{\sim}{\longrightarrow}\,\mc V(\mf g^f)
\,,
$$
hence we have the inverse differential algebra isomorphism
$$
w:\,\mc V(\mf g^f)\,\stackrel{\sim}{\longrightarrow}\,\mc W(\mf g,f)
\,,
$$
which associates to every element $q\in\mf g^f$ the (unique) element $w(q)\in\mc W$
of the form $w(q)=q+r$, with $r\in\langle U^\perp\rangle$.
\end{theorem}

\subsection{Generators for the classical $\mc W$-algebra $\mc W(\mf{gl}_N,\ul p)$}
\label{sec:6.3}

Consider the Lie algebra $\mf g=\mf{gl}_N$,
with the trace form $(a,b)=\tr(ab)$.
Associated to the partition $\ul p=(p_1,\dots,p_r)$ of $N$, 
is the index set of cardinality $N$
\begin{equation}\label{eq:I}
I\,=\,\big\{(a,i)\,\,\text{ with }\,\,
1\leq a\leq r,\,1\leq i\leq p_a
\big\}
\,,
\end{equation}
which we order lexicographically.
We then let $V$ be the vector space with basis $\{e_\alpha\}_{\alpha\in I}$,
and we identify $\mf{gl}_N=\mf{gl}(V)$.
Let $f\in\mf{gl}_N$ be the nilpotent element in Jordan form
associated to the partition $\ul p$,
i.e. 
\begin{equation}\label{eq:f}
f(e_{(a,i)})=e_{(a,i+1)} \,\text{ for }\, i<p_a
\,,\,\,\text{ and }\,
f(e_{(a,p_a)})=0
\,.
\end{equation}
Let also $\chi\in\mf{gl}_N$ be the diagonal matrix
with eigenvalues 
\begin{equation}\label{eq:x}
\chi_{(a,i)}=\frac{p_a+1}{2}-i
\,\in\frac12\mb Z
\,.
\end{equation}
Note that the adjoint action of $\chi$ defines a Dynkin grading \eqref{eq:grading} of $\mf g$.
We then consider the corresponding classical $\mc W$-algebra \eqref{20120511:eq2},
which we denote $\mc W(\mf{gl}_N,\ul p)$.

By \cite[Prop.5.2]{DSKV16b}, the following elements form a basis for $\mf g^f$, the centralizer of $f$ in $\mf g$,
\begin{equation}\label{eq:basis-gf}
f_{ab;i}
=
\sum_{j=0}^i E_{(a,p_a+j-i),(b,j+1)}
\,\,,\,\,\,\,
1\leq a,b\leq r,\,0\leq i\leq\min\{p_a,p_b\}-1
\,,
\end{equation}
where $E_{\alpha,\beta}$, $\alpha,\beta\in I$, denote the standard matrices:
$E_{\alpha,\beta}(e_\gamma)=\delta_{\beta,\gamma}e_\alpha$.
Moreover, by \cite[Prop.5.1]{DSKV16b} the following is a subspace of $\mf g$ complementary to $[f,\mf g]$:
\begin{equation}\label{eq:basis-U}
U=\Span\big\{
E_{(b,1),(a,p_a-i)}\,\big|\,
1\leq a,b\leq r,\,0\leq i\leq\min\{p_a,p_b\}-1
\big\}
\,.
\end{equation}
Equivalently, it is a subspace of $\mf g_{\geq-\frac12}$ complementary to $[f,\mf g_{\geq\frac12}]$,
and it is obviously compatible with the grading \eqref{eq:grading}.
In fact, the basis \eqref{eq:basis-gf} of $\mf g^f$ and \eqref{eq:basis-U} of $U$ are dual to each other:
$$
\Big(\sum_{j=0}^i E_{(a,p_a+j-i),(b,j+1)}\,\Big|\,E_{(b',1),(a',p_{a'}-i')}\Big)
=
\delta_{a,a'}\delta_{b,b'}\delta_{i,i'}
\,.
$$
With this choice, Theorem \ref{thm:structure-W},
provides a differential algebra isomorphism
$w:\,\mc V(\mf g^f)\stackrel{\sim}{\longrightarrow}\mc W(\mf{gl}_N,\ul p)$,
and we get a set of generators for the $\mc W$-algebra,
viewed as an algebra of differential polynomials,
corresponding to the basis \eqref{eq:basis-gf} of $\mf g^f$:
\begin{equation}\label{eq:wabk}
w_{ab;k}=w(f_{ab;k})\,\in W(\mf{gl}_N,\ul p)
\,.
\end{equation}

\subsection{The Lax operator for $\mc W(\mf{gl}_N,\ul p)$
and associated integrable Hamiltonian hierarchy}
\label{sec:6.4}

Following \cite{DSKV16b},
we encode all the $\mc W$-algebra generators \eqref{eq:wabk}
into the $r\times r$-matrix differential operator
\begin{equation}\label{eq:matrW}
W(\partial)
=
\Big(\!\!\!
\sum_{i=0}^{\min\{p_a,p_b\}-1}
w_{ba;i}(-\partial)^i
\,\,\,\Big)_{a,b=1}^r
\,\in\Mat_{r\times r}\mc W(\mf{gl}_N,\ul p)[\partial]
\,,
\end{equation}
which we write in block form, cf. \eqref{eq:block},
\begin{equation}\label{eq:blockW}
W(\partial)
=
\begin{pmatrix}
W_{\bm1\bm1}(\partial) & W_{\bm1\bm2}(\partial) \\
W_{\bm2\bm1}(\partial) & W_{\bm2\bm2}(\partial) 
\end{pmatrix}
\,,
\end{equation}
with blocks of sizes $r_1\times r_1$, $r_1\times(r-r_1)$, $(r-r_1)\times r_1$ and $(r-r_1)\times(r-r_1)$.
By \cite[Eq.(5.16)]{DSKV16b},
the Lax operator $\mc L(\partial)\in\Mat_{r_1\times r_1}\mc W(\mf{gl}_N,\ul p)((\partial^{-1}))$
is obtained as the quasideterminant of the matrix $-(-\partial)^{\ul p}+W(\partial)$
w.r.t. the first $r_1$ rows and columns,
where 
\begin{equation}\label{eq:Dp}
(-\partial)^{\ul p}:=\sum_{a=1}^r E_{aa}(-\partial)^{p_a}
\,.
\end{equation}
Explicitly,
\begin{equation}\label{eq:WL}
\mc L(\partial)
=
-\id_{r_1}(-\partial)^{p_1}
+W_{\bm1\bm1}(\partial)
-W_{\bm1\bm2}(\partial)
\circ
\big(-(-\partial)^{\ul q}+W_{\bm2\bm2}(\partial)\big)^{-1}
\circ
W_{\bm2\bm1}(\partial)
\,,
\end{equation}
where $\ul q$ is obtained from the partition $\ul p$ by removing the $r_1$ parts of maximal size $p_1$,
so that
\begin{equation}\label{eq:Dq}
(-\partial)^{\ul q}=\sum_{a=r_1+1}^r E_{aa}(-\partial)^{p_a}\,.
\end{equation}
\begin{theorem}[{\cite[Sec.6.4]{DSKV16b}}]\label{thm:laxeq}
Given a partition $\ul p$ of $N$, 
consider the $r_1\times r_1$-matrix pseudodifferential operator
$\mc L(\partial)\in\Mat_{r_1\times r_1}\mc W(\mf{gl}_N,\ul p)((\partial^{-1}))$
defined by \eqref{eq:WL},
and let $\mc L(\partial)^{\frac{1}{p_1}}$ be an arbitrary $p_1$-root of $\mc L(\partial)$.
The local functionals
\begin{equation}\label{eq:densities}
\tint h_j
=
\frac{p_1}{j}\int \Res_\partial\tr\big(\mc L(\partial)^{\frac{j}{p_1}}\big)
\,\in\mc W(\mf{gl}_N,\ul p)
\,,\,\,
j\in\mb Z_{\geq1}
\,,
\end{equation}
are in involution w.r.t. the $\mc W$-algebra $\lambda$-bracket:
$$
\big\{\tint h_i,\tint h_j\big\}^{\mc W}=0
\,\text{ for all }\, i,j\in\mb Z_{\geq1}
\,.
$$
Hence, we have the corresponding integrable hierarchy of Hamiltonian equations
\begin{equation}\label{eq:hameq}
\frac{\partial w}{\partial t_j}
=
\big\{\tint h_j,w\}^{\mc W}
\,\text{ for all }\,
j\in\mb Z_{\geq1}
\,.
\end{equation}
Furthermore, the hierarchy \eqref{eq:hameq} 
implies the hierarchy of Lax equations for $\mc L(\partial)$:
\begin{equation}\label{eq:527}
\frac{\partial\mc L(\partial)}{\partial t_j}
=
\big[\big(\mc L(\partial)^{\frac{j}{p_1}}\big)_+,\mc L(\partial)]
\,\,,\,\,\,\,
j\in\mb Z_{\geq1}
\,.
\end{equation}
\end{theorem}
\begin{proof}
By \cite[Thm.4.6]{DSKV16b}, the matrix pseudodifferential operator $\mc L(\partial)$
is of Adler type w.r.t. the $\mc W$-algebra $\lambda$-bracket.
Then, the claim is a special case of \cite[Thm.5.1]{DSKV16a}.
\end{proof}
The main goal of the present paper is to construct
tau-functions of the Hamiltonian hierarchy \eqref{eq:hameq},
which are exhibited in Theorem \ref{thm:main} below.

\section{
Tau-functions for the Hamiltonian hierarchy 
associated to $\mc W(\mf{gl}_N,\ul p)$
}
\label{sec:6b}

Theorem \ref{thm:lax}
provides tau-function solutions $\widetilde{\mc L}(\ul m,x,\bm t,\partial)$
to the hierarchy of Lax equations \eqref{eq:Laxeq-tilde}.
On the other hand, according to Theorem \ref{thm:laxeq},
a solution to the Hamiltonian hierarchy \eqref{eq:hameq}
automatically provides a solution to the hierarchy of Lax equations \eqref{eq:527}
(which is the same as \eqref{eq:Laxeq-tilde}).
It is therefore natural to ask whether the matrix $\widetilde{W}(\ul m,x,\bm t,\partial)$
constructed in Section \ref{sec:5}
also solves the ``full'' Hamiltonian hierarchy \eqref{eq:hameq}.
This is essentially true, and it is the content of Theorem \ref{thm:main} below,
which is the main result of the paper.

Unfortunately, the forms of the operator $\widetilde{\mc L}(\ul m,x,\bm t,\partial)$
in \eqref{eq:WL2} and of the operator $\mc L(\partial)$ in \eqref{eq:WL}
do not quite match, due to a different choice of signs
(this is the reason for the tilde-notation).
To pass from \eqref{eq:WL2} to \eqref{eq:WL} 
we need to change the sign of $\partial$ and of $\mc L$.
On the other hand, 
if, after these changes of signs, we want that the Lax equations \eqref{eq:Laxeq-tilde}
and \eqref{eq:527} correspond to each other,
we need to change sign of the space variable $x$ and multiply the time variables $t_j$
by a factor $(-1)^{\frac{j}{p_1}}$ ($(-1)^{\frac1{p_1}}$ being an arbitrary $p_1$-root of $-1$).
So we let
\begin{equation}\label{eq:Lsol}
\mc L(\ul m,x,\bm t,\partial)
=
-\widetilde{\mc L}(\ul m,-x,\tilde{\bm t},-\partial)
\,\,\text{ were }\,\,
\tilde{\bm t}
=((-1)^{\frac{j}{p_1}}t_j)_{j\in\mb Z_{>0}}
\,.
\end{equation} 
Equation \eqref{eq:WL2} then can be rewritten as
\begin{equation}\label{eq:WLsol}
\mc L(\ul m,x,\bm t,\partial)
=
-\id_{r_1}(-\partial)^{p_1}
+W_{\bm1\bm1}(\ul m,x,\bm t,\partial)
+W_{\bm1\bm2}(\ul m,x,\bm t,\partial)
\circ
(-\partial)^{-\ul q}
\circ
W_{\bm2\bm1}(\ul m,x,\bm t,\partial)
\,,
\end{equation}
where
\begin{equation}\label{eq:W11sol}
W_{\bm a\bm b}(\ul m,x,\bm t,\partial)
=-\widetilde{W}_{\bm a\bm b}(\ul m,-x,\tilde{\bm t},-\partial)
\,\text{ for }\, (\bm a,\bm b)\neq(\bm2,\bm2)
\,.
\end{equation}
Note that \eqref{eq:WLsol} has the same form as \eqref{eq:WL},
if we set 
\begin{equation}\label{eq:W22-sol}
W_{\bm2\bm2}(\ul m,x,\bm t,\partial)=0
\,.
\end{equation}

Recalling \eqref{eq:W+ab}, \eqref{eq:W-ab} and \eqref{eq:W11},
we can find explicit formulas for the matrix entries
of $W(\ul m,x,\bm t,\partial)$ in terms of the wave operators $Q^{\pm}(\ul m,x,\bm t,\partial)$
constructed in Section \ref{S4}.
Let, as in \eqref{eq:matrW}-\eqref{eq:blockW},
\begin{equation}\label{eq:matrW-sol1}
W(\ul m,x,\bm t,\partial)
=
\big(
W_{\bm a\bm b}(\ul m,x,\bm t,\partial)
\big)_{\bm a,\bm b\in\{\bm1,\bm2\}}
=
\Big(\!\!\!
\sum_{i=0}^{\min\{p_a,p_b\}-1}
w_{ba;i}(\ul m,x,\bm t)\, (-\partial)^i
\,\Big)_{1\leq a,b\leq r}
\,,
\end{equation}
with coefficients $w_{ba;i}(\ul m,x,\bm t)\in\mc F$.
By \eqref{eq:W11sol} with $\bm a=\bm1$, $\bm b=\bm2$,
and by \eqref{eq:W+ab}, we get
\begin{equation}\label{eq:W12-sola}
\sum_{j=0}^{p_b-1}
w_{ba;j}(\ul m,x,\bm t)\, (-\partial)^j
=
-
\sum_{j=0}^{p_b-1}
\sum_{i=0}^{j}
\frac{x^i}{i!}
Q^+_{ab;j-i+1} (\ul m,-x,\tilde{\bm t})
(-\partial)^{j}
\,\,,\,\text{ for }\,
1\leq a\leq r_1,\,r_1\leq b\leq r
\,,
\end{equation}
or, equivalently,
\begin{equation}\label{eq:W12-solb}
w_{ba;j}(\ul m,x,\bm t)
=
-
\sum_{i=0}^{j}
\frac{x^i}{i!}
Q^+_{ab;j-i+1} (\ul m,-x,\tilde{\bm t})
\,\,,\,\text{ for }\,
1\leq a\leq r_1,\,r_1\leq b\leq r,\,0\leq j\leq p_b-1
\,,
\end{equation}
where the RHS is evaluated at the times $\bm t$ as in \eqref{eq:times} and \eqref{eq:Lsol}.
Similarly, by \eqref{eq:W11sol} with $\bm a=\bm2$, $\bm b=\bm1$,
and by \eqref{eq:W-ab}, we get
\begin{equation}\label{eq:W21-sol}
\sum_{j=0}^{p_a-1}
w_{ba;j}(\ul m,x,\bm t)\, (-\partial)^j
=
\sum_{j=0}^{p_a-1}
\sum_{i=0}^{j}
\partial^{j}
\circ
Q^-_{ba;j-i+1}(\ul m,-x,\tilde{\bm t})
\frac{x^i}{i!}
\,\,,\,\text{ for }\,
r_1< a\leq r,\,1\leq b\leq r_1
\,,
\end{equation}
from which the coefficients $w_{ba;j}(\ul m,x,\bm t)\in\mc F$,
with $r_1< a\leq r,\,1\leq b\leq r_1,\,0\leq j\leq p_a-1$,
can be easily computed by expanding the RHS.
Next, if we combine \eqref{eq:W11sol} with $\bm a=\bm b=\bm1$
with \eqref{eq:W11} and \eqref{eq:Lsol},
we just end up with equation \eqref{eq:WLsol}.
In order to get formulas for the remaining functions 
$w_{ba;j}(\ul m,x,\bm t)\in\mc F$,
with $1\leq a,b\leq r_1$ and $0\leq j\leq p_1-1$,
we need to use equations \eqref{eq:mcLa} and \eqref{eq:WLax}.
As a result, we get
\begin{equation}\label{eq:W11-sol}
\begin{split}
& \Big(
\sum_{j=0}^{p_1-1}
w_{ba;j}(\ul m,x,\bm t)\, (-\partial)^j
\Big)_{a,b=1}^{r_1}
=
W_{\bm1\bm1}(\ul m,x,\bm t,\partial)
=
-\widetilde{W}_{\bm1\bm1}(\ul m,-x,\tilde{\bm t},-\partial) \\
& =
-
\widetilde{\mc L}(\ul m,-x,\tilde{\bm t},-\partial)_+
+
\id_{r_1}(-\partial)^{p_1}
-
\Big(
\widetilde{W}_{\bm1\bm2}(\ul m,-x,\tilde{\bm t},-\partial)
\circ
(-\partial)^{-\ul q}
\circ
\widetilde{W}_{\bm2\bm1}(\ul m,-x,\tilde{\bm t},-\partial)
\Big)_+ \\
& =
-
\big(
Q^+_{\bm1\bm1}(\ul m,-x,\tilde{\bm t},-\partial)
\circ \id_{r_1}(-\partial)^{p_1}\circ 
Q^{+}_{\bm1\bm1}(\ul m,-x,\tilde{\bm t},-\partial)^{-1}
\big)_+ \\
&\quad
+
\id_{r_1}(-\partial)^{p_1}
-
\Big(
W_{\bm1\bm2}(\ul m,x,\bm t,\partial)
\circ
(-\partial)^{-\ul q}
\circ
W_{\bm2\bm1}(\ul m,x,\bm t,\partial)
\Big)_+
\,,
\end{split}
\end{equation}
from which the coefficients $w_{ba;j}(\ul m,x,\bm t)\in\mc F$,
with $1\leq a,b\leq r_1,\,0\leq j\leq p_1-1$,
can be explicitly derived, 
by computing the matrix entries in the RHS and expanding them as differential operators.
Finally, recalling \eqref{eq:W22-sol}, we set
\begin{equation}\label{eq:W22-sol-coeff}
w_{ba;j}(\ul m,x,\bm t)
=
0
\,\,,\,\text{ for }\,
r_1\leq a,b\leq r,\,0\leq j\leq \min\{p_a,p_b\}-1
\,.
\end{equation}

We can now state  the  first main result of the paper.
\begin{theorem}\label{thm:main}
Let $k\in\mb Z$ and $\ul m\in\mb Z^r$ be such that $|\ul m|=k$. Let
$\vec{\tau}(\ul{\bm t })=\{
 \tau^{\ul m}(\ul{\bm t})\}_{|\ul m|=k}$ be a tau-function of the $\underline p$-KdV hierarchy such that 
$ \tau^{\ul m}(\ul{\bm t})\ne 0$.
Then the functions 
$w_{ba;j}(\ul m,x,\bm t)\in\mc F$,
$1\leq a,b\leq r$, $0\leq j\leq \min\{p_a,p_b\}-1$,
defined by \eqref{eq:W12-solb}, \eqref{eq:W21-sol}, \eqref{eq:W11-sol} and \eqref{eq:W22-sol-coeff},
form a solution of the integrable hierarchy \eqref{eq:hameq}
associated to the classical $\mc W$-algebra $\mc W(\mf{gl}_N,\ul p)$.
\end{theorem}

The remainder of the paper is devoted to the proof of Theorem \ref{thm:main}.
First, we observe, in Section \ref{sec:7}, that the operator 
$\mc L(\ul m,x,\bm t,\partial)\in\Mat_{r_1\times r_1}\mc F((\partial^{-1}))$ 
associated by equation \eqref{eq:WLsol} to the functions
$w_{ba;j}(\ul m,x,\bm t)\in\mc F$,
$1\leq a,b\leq r$, $0\leq j\leq \min\{p_a,p_b\}-1$,
indeed solves the hierarchy of Lax equations \eqref{eq:527}.
This unfortunately does not suffice to prove Theorem \ref{thm:main},
as the Lax equations \eqref{eq:527} are implied by the hierarchy
of Hamiltonian equations \eqref{eq:hameq},
but a priori \eqref{eq:527} does not imply \eqref{eq:hameq}.
Only in Section \ref{sec:9} we will prove that,
in fact, the Lax equations \eqref{eq:527} do indeed 
imply the full hierarchy of Hamiltonian equations \eqref{eq:hameq},
see Theorem \ref{prop:sylvain} below.
In order to prove this fact,
a key point is the observation that,
along the Hamiltonian flow \eqref{eq:hameq},
the submatrix $W_{\bm2\bm2}(\partial)$ does not evolve.
This fact will be proved in Corollary \ref{prop:dW4dtn} in Section \ref{sec:8c}.
Before stating and proving the crucial point, that $W_{\bm2\bm2}(\partial)$ does not evolve,
we shall review in Section \ref{sec:8a} some notation and preliminary results on
the classical $\mc W$-algebra $\mc W(\mf{gl}_N,\ul p)$,
and we will provide in Section \ref{sec:8b} 
a new, algorithmic way,
to construct the generator matrix $W(\partial)\in\Mat_{r\times r}\mc W(\mf{gl}_N,\ul p)[\partial]$ 
defined in \eqref{eq:matrW}, see Corollary \ref{cor:W} below.
The proof of Corollary \ref{prop:dW4dtn}, i.e. that $W_{\bm2\bm2}(\partial)$ does not evolve,
will be then based on this algorithmic construction of the matrix $W(\partial)$.

\section{
Tau-function solutions for the Lax equations \eqref{eq:527}
}
\label{sec:7}

\begin{proposition}\label{thm:lax-sol}
Let $k\in\mb Z$, $\ul m\in\mb Z^r$  and $\vec{\tau}(\ul{\bm t})$ be as in Theorem \ref{thm:main}.
Consider the functions 
$w_{ba;j}(\ul m,x,\bm t)\in\mc F$,
$1\leq a,b\leq r$, $0\leq j\leq \min\{p_a,p_b\}-1$,
defined by \eqref{eq:W12-solb}, \eqref{eq:W21-sol}, \eqref{eq:W11-sol} and \eqref{eq:W22-sol-coeff}.
Then the matrix pseudodifferential operator 
$\mc L(\ul m,x,\bm t,\partial)\in\Mat_{r_1\times r_1}\mc F((\partial^{-1}))$
given by \eqref{eq:WLsol} (with the notation \eqref{eq:matrW-sol1})
solves the hierarchy of Lax equations \eqref{eq:527}.
\end{proposition}
\begin{proof}
By the constructions of Section \ref{sec:6b},
equation \eqref{eq:Lsol} holds.
The claim is then an immediate consequence of Theorem \ref{thm:lax}.
\end{proof}

\section{
Some preliminaries on PVA's and $\mc W$-algebras
}
\label{sec:8a}

\subsection{
Notational conventions
}
\label{sec:8a.0}

Given a finite-dimensional vector space $V$, consider the associative algebra $\End V$,
the Lie algebra $\mf g=\mf{gl}(V)$,
and the classical affine PVA $\mc V(\mf g)$ defined by \eqref{lambda}.
Even though the two spaces $\mf{gl}(V)$ and $\End V$ are canonically identified, 
we will keep them distinct.
For this, we shall usually denote the elements of the Lie algebra $\mf g$
by lowercase letters $a,b,\dots$,
and the same elements, viewed as elements
of the associative algebra $\End V$, by the corresponding uppercase letters $A,B,\dots$.
Also, we shall usually drop the tensor product sign for the elements of $\mc V(\mf g)\otimes\End V$;
hence, for example, $aB$ will denote the monomial of $\mc V(\mf{g})\otimes\End V$, with $a\in\mf{gl}(V)$
and $B\in\End V$. 

A $\lambda$-bracket between an element in $\mc V(\mf g)$
and element in $\mc V(\mf g)\otimes\End V$ has to be interpreted as
$$
\{a_\lambda bC\}=\{a_\lambda b\}C\,\in\mc V(\mf g)[\lambda]\otimes\End V
\,,
$$
while a $\lambda$-bracket between two elements of $\mc V(\mf g)\otimes\End V$ 
has to be interpreted as
$$
\{aB_\lambda cD\}=\{a_\lambda c\}B\otimes D\,\in\mc V(\mf g)[\lambda]\otimes\End V\otimes\End V
\,.
$$
On $\mc V(\mf g)((\partial^{-1}))\otimes\End V$ we also have a natural associative product, 
defined componentwise:
$$
(a(\partial)B)(c(\partial)D)=(a(\partial)\circ c(\partial))\,(BD)\,\in \mc V(\mf g)((\partial^{-1}))\otimes\End V
\,.
$$
Similarly, we have a natural associative product, defined componentwise, on 
$\mc V(\mf g)((\partial^{-1}))\otimes\End V\otimes\End V$.

In Section \ref{sec:8a.1}, given an element $A(\partial)=a(\partial)B\in\mc V(\mf g)((\partial^{-1}))\otimes\End V$, we denote by
$A^{*,1}(\partial)$ its adjoint with respect to the first factor of the tensor product, i.e. $A^{*,1}(\partial)=a^*(\partial)B$. 

\subsection{
Some PVA $\lambda$-bracket computations
}
\label{sec:8a.1}

Let $\{u_i\}_{i\in J}$ be a basis of $\mf g$ 
compatible with the grading \eqref{eq:grading},
and let $\{u^i\}_{i\in J}$ be the dual basis w.r.t. the trace form.
According to the notational convention described in Section \ref{sec:8a.0},
we let $\{U_i\}_{i\in J}$ and $\{U^i\}_{i\in J}$ be the same dual bases,
viewed as bases of the associative algebra $\End V$.
Consider the matrix differential operator
\begin{equation}\label{eq:Az}
A(\partial)
=
\id_V\partial+\sum_{i\in J}u_i U^i
\,\in\mc V(\mf g)[\partial]\otimes\End V
\,.
\end{equation}
Let also $\Omega_V\in\End V\otimes\End V$ be the operator of permutation of the two factors:
$\Omega_V(u\otimes v)=v\otimes u$, $u,v\in V$.
In terms of basis elements: $\Omega_V=\sum_{i\in J}U_i\otimes U^i$.
The operator $A(\partial)$ satisfies the following \emph{Adler identity}, 
see \cite[Eq.(5.25)]{DSKV18},
\begin{equation}\label{eq:adler-glN}
\begin{array}{c}
\displaystyle{
\vphantom{\Big(}
\{A(z)_\lambda A(w)\}
=
(\id\otimes A(w+\lambda+\partial))(z-w-\lambda-\partial)^{-1}
(A^{*,1}(\lambda-z)\otimes\id)\Omega_V
} \\
\displaystyle{
\vphantom{\Big(}
-\Omega_V\,(A(z)\otimes(z-w-\lambda-\partial)^{-1}A(w))
\,\in\mc V(\mf g)[\lambda]\otimes\End V\otimes\End V
\,.}
\end{array}
\end{equation}
Here the expression $(z-w-\lambda-\partial)^{-1}$ is assumed to be expanded in negative powers of $z$
(or of $w$, since the RHS will be the same).
The verification of \eqref{eq:adler-glN} is a straightforward computation.
The LHS is obviously independent of $z$ and $w$,
hence the RHS is independent of $z$ and $w$ as well.

We shall need in Section \ref{sec:8b} a formula for $\{A(z)_\lambda A^{-1}(w)\}$,
where $A^{-1}(w)$ is the symbol of 
the inverse of the operator \eqref{eq:Az} in $\mc V(\mf g)((\partial^{-1}))\otimes\End V$.
Recall, from \cite[Lem.2.3(g)]{DSKV18}
that, if $A(\partial)$ and $B(\partial)$
are matrix pseudodifferential operators with coefficients in a PVA $\mc V$
and if $B(\partial)$ is invertible, then
\begin{equation}\label{eq:AB-1}
\{A(z)_\lambda B^{-1}(w)\}
=-\big(1\otimes B^{-1}(w+\lambda+\partial)\big)
\{A(z)_{\lambda}B(w+x)\}\big(1\otimes \big|_{x=\partial}B^{-1}(w)\big)
\,.
\end{equation}
Here and further we use the following notation:
given a pseudodifferential operator 
$a(\partial)=\sum_{n=-\infty}^Na_n\partial^n\in\mc V((\partial^{-1}))$
and elements $b,c\in \mc V$, we let:
\begin{equation}\label{eq:notation}
a(z+x)\big(\big|_{x=\partial}b)c
=\sum_{n=-\infty}^Na_n((z+\partial)^nb)c\,\in\mc V
\,,
\end{equation}
where in the RHS we expand, for negative $n$, in the domain of large $z$.
As a consequence of the Adler identity \eqref{eq:adler-glN}
and equation \eqref{eq:AB-1}, we have (cf. \cite[Eq.(A.1)]{DSKV16b})
\begin{equation}\label{20190410:eq2}
\begin{split}
\{A(z)_{\lambda}A^{-1}(w)\}
& =
\Omega_V\Big(A^{-1}(w+\lambda+\partial)A(z)\otimes\id_V\Big)(z-w-\lambda)^{-1} \\
& -
(z-w-\lambda-\partial)^{-1}\Big(A^{*,1}(\lambda-z)A^{-1}(w)\otimes\id_V\Big)\Omega_V
\,.
\end{split}
\end{equation}

\subsection{
Some formulas for the classical $\mc W$-algebras
}
\label{sec:8a.1.5}

Consider the classical affine $\mc W$-algebra $\mc W(\mf g,f)$ defined 
by \eqref{20120511:eq2}-\eqref{20120511:eq3}.
Recall from \cite[Lem.3.1(b),Cor.3.3(d)]{DSKV13} that,
for $a\in\mf g_{\geq\frac12}$ and $g\in\mc V(\mf g)$, we have 
\begin{equation}\label{eq:Walg1}
\rho\{a_\lambda \rho(g)\}=\rho\{a_\lambda g\}
\,,
\end{equation}
while for $g,h\in\mc V(\mf g)$ such that $\rho(g),\rho(h)\in\mc W$,
we have 
\begin{equation}\label{eq:Walg2}
\{\rho(g)_\lambda\rho(h)\}^{\mc W}
=
\rho\{g_\lambda h\}
\,.
\end{equation}

\subsection{
Notation for $\mf{gl}_N$
}
\label{sec:8a.2}

Fix, as in the previous sections,
a partition $\ul p=(p_1,\dots,p_r)$ of $N$, 
with $p_1\geq\dots\geq p_r>0$, $p_1+\dots+p_r=N$.
Let $I$ be the corresponding index set \eqref{eq:I} of cardinality $N$,
and let $V$ be the vector space with basis $\{e_\alpha\}_{\alpha\in I}$.
We depict the basis elements $e_\alpha$, $\alpha\in I$,
as the boxes of a symmetric, with respect to the $y$-axis, pyramid,
with $r$ rows of length, from bottom to top, $p_1,\dots,p_r$.
For example, for the partition $\ul p=(4,4,4,2,1,1)$ of $N=16$, 
the corresponding pyramid is

\bigskip

\bigskip

\begin{figure}[H]
\setlength{\unitlength}{0.14in}
\setlength\fboxsep{0pt}
\centering
\begin{picture}(30,12)

\put(10,3){\colorbox{lightgray}{\framebox(2,2){}}}
\put(12,3){\framebox(2,2){$\dots$}}
\put(14,3){\framebox(2,2){(1,2)}}
\put(16,3){\colorbox{cyan}{\framebox(2,2){(1,1)}}}

\put(10,5){\colorbox{lightgray}{\framebox(2,2){}}}
\put(12,5){\framebox(2,2){}}
\put(14,5){\framebox(2,2){}}
\put(16,5){\colorbox{cyan}{\framebox(2,2){}}}

\put(10,7){\colorbox{lightgray}{\framebox(2,2){}}}
\put(12,7){\framebox(2,2){}}
\put(14,7){\framebox(2,2){}}
\put(16,7){\colorbox{cyan}{\framebox(2,2){}}}

\put(12,9){\colorbox{lime}{\framebox(2,2){}}}
\put(14,9){\colorbox{pink}{\framebox(2,2){}}}

\put(13,11){\colorbox{brown}{\framebox(2,2){}}}
\put(13,13){\colorbox{brown}{\framebox(2,2){(6,1)}}}

\put(8,2){\vector(1,0){12}}
\put(20,1){$x$}

\put(11,1.8){\line(0,1){0.4}}
\put(12,1.6){\line(0,1){0.8}}
\put(13,1.8){\line(0,1){0.4}}
\put(14,1.6){\line(0,1){0.8}}
\put(15,1.8){\line(0,1){0.4}}
\put(16,1.6){\line(0,1){0.8}}
\put(17,1.8){\line(0,1){0.4}}

\put(13.8,0.6){0}
\put(15.8,0.6){1}
\put(11.6,0.6){-1}

\put(14.8,-0.3){\tiny{$\frac12$}}
\put(16.8,-0.3){\tiny{$\frac32$}}
\put(12.6,-0.3){\tiny{$-\frac12$}}
\put(10.6,-0.3){\tiny{$-\frac32$}}

\put(7,6){\color{darkgray}{$V_{-,1}$}}
\put(9,10){\color{lime}{$V_{-,2}$}}
\put(10,13){\color{brown}{$V_{-,3}$}}

\put(19,6){\color{blue}{$V_{+,1}$}}
\put(17,10){\color{pink}{$V_{+,2}$}}
\put(16,13){\color{brown}{$V_{+,3}$}}

\put(25,11.5){\vector(-1,0){2}}
\put(24,12){$F$}

\put(23,8.5){\vector(1,0){2}}
\put(24,9){$F^T$}







\thinlines

\end{picture}
\caption{} 
\label{fig:pyramid}
\end{figure}
\noindent
According to this pictorial description,
the basis elements $e_{(a,h)}$, $(a,h)\in I$, are labeled 
by the row index $a$, counting from bottom to top,
and the column index $h$, counting from right to left.
In particular, the $x$ coordinate of the center of the box $e_{(a,h)}$ is
$\chi_{(a,h)}$ as in \eqref{eq:x}.
Let $r_1$ be the number of rows of the pyramid of maximal length $p_1$,
$r_2$ the number of rows of second maximal length,
and so on, up to $r_s$, the number or rows of minimal length.
We also let 
\begin{equation}\label{eq:Ri}
R_0=0\,\,\text{ and }\,\,R_i=r_1+\dots+r_i
\,\,\text{ for }\,\,
1\leq i\leq s
\,.
\end{equation}
In particular, $R_s=r$.
Note that the pyramid attached to $\ul p$ consists of $s$ rectangles,
of sizes $p_{R_i}\times r_i$, $i=1,\dots,s$.
In the example of Figure \ref{fig:pyramid}, we have
$s=3$, $r_1=3$, $r_2=1$, $r_3=2$,
$R_1=3$, $R_2=4$ and $R_3=6=r$.

According to the notational convention described in Section \ref{sec:8a.0},
we denote by $f\in\mf g=\mf{gl}(V)$ the nilpotent element \eqref{eq:f} of shift to the left,
and by $F\in\End V$ the same endomorphism, 
when viewed as an element of the associative algebra $\End V$.
We also let $F^T\in\End V$ be the shift to the right,
and $X\in\End V$ be the diagonalizable operator with eigenvalues \eqref{eq:x};
in formulas
\begin{equation}\label{eq:FX}
F(e_{(a,h)})=
\left\{\begin{array}{l} 
e_{(a,h+1)}\,\text{ if }\, h<p_a \\
0\,\,\text{ if }\,\, h=p_a
\end{array}\right.
\,,\,\,
F^T(e_{(a,h)})=
\left\{\begin{array}{l} 
0\,\,\text{ if }\,\, h=1 \\
e_{(a,h-1)}\,\text{ if }\, h>1
\end{array}\right.
\,,\,\,
X(e_{(a,h)})=\chi_{(a,h)} e_{(a,h)}
\,.
\end{equation}
Recall the $\ad\chi$-eigenspace decomposition \eqref{eq:grading} of $\mf g$.
Analogously,
we have the $X$-eigenspace decomposition of the vector space $V$:
\begin{equation}\label{eq:grV}
V=\bigoplus_{k\in\frac{1}{2}\mb Z}V[k]
\,\,,\,\,\,\,
V[k]=\big\{v\in V\,\big|\,X(v)=kv\big\}
\,,
\end{equation}
and the corresponding $\ad X$-eigenspace decomposition of $\End V$:
\begin{equation}\label{eq:grEndV}
\End V=\bigoplus_{k\in\frac{1}{2}\mb Z}(\End V)[k]
\,\,,\,\,\,\,
(\End V)[k]=\big\{A\in\mf g\,\big|\,[X,A]=kA\big\}
\,.
\end{equation}
With a slight abuse of terminology, we shall say that an element $A\in(\End V)[\geq k]$
has $\ad X$-eigenvalue greater than or equal to $k$,
similarly for the elements of $(\End V)[\leq k]$.
In the pictorial description of Figure \ref{fig:pyramid},
the endomorphisms in $\End V$ of positive $\ad X$ eigenvalue
move the blocks of the diagram to the right,
while the endomorphisms of negative $\ad X$ eigenvalue move them to the left.

Let $V_+=\ker(F^T)$ and $V_-=\ker(F)$,
which are spanned, respectively, by the rightmost and leftmost boxes of the pyramid.
In particular $\dim(V_-)=\dim(V_+)=r$.
They decompose as $V_{\pm}=\bigoplus_{i=1}^s V_{\pm,i}$,
where $V_{\pm,i}$ is the $r_i$-dimensional vector space spanned 
by the right/leftmost boxes in the $i$-th rectangle,
counting from bottom to top
(see Figure \ref{fig:pyramid}):
\begin{equation}\label{eq:Vpmi}
V_{+,i}
=
\Span\big\{e_{(a,1)}\big\}_{R_{i-1}<a\leq R_i}
\,\,,\,\,\,\,
V_{-,i}
=
\Span\big\{e_{(a,p_a)}\big\}_{R_{i-1}<a\leq R_i}
\,.
\end{equation}

\subsection{
The matrix differential operators $W(\partial)$ and $Z(\partial)$
}
\label{sec:8a.2.6}

Recall the matrix differential operator $W(\partial)$ defined in \eqref{eq:matrW}.
Once we fix the bases \eqref{eq:Vpmi} of $V_\pm$,
we can identify $V_+\simeq V_-\simeq\mb C^r$, and hence
\begin{equation}\label{eq:identif-pm}
\Hom(V_{-},V_{+})\simeq\Mat_{r\times r}\mb C
\,.
\end{equation}
Under this identification, \eqref{eq:Dp} becomes
\begin{equation}\label{eq:Dp2}
(-\partial)^{\ul p}:=\sum_{a=1}^r E_{(a,1)(a,p_a)}(-\partial)^{p_a}
\,,
\end{equation}
while \eqref{eq:matrW} becomes the following differential operator
\begin{equation}\label{eq:EndW}
W(\partial)
=
\sum_{a,b=1}^r
\sum_{i=0}^{\min\{p_a,p_b\}-1}
w_{ba;i}(-\partial)^i
E_{(a,1),(b,p_b)}
\,\in\mc W(\mf{gl}_N,\ul p)[\partial]\otimes\Hom(V_-,V_+)
\,.
\end{equation}
It is a matrix differential operator encoding all $\mc W$-algebra generators \eqref{eq:wabk}.
We have 
$W(\partial)=w(Z(\partial))$
and
$Z(\partial)=\pi(W(\partial))$,
where $w$ and $\pi$ are the differential algebra isomorphisms between $\mc W(\mf{gl}_N,\ul p)$
and $\mc V(\mf g^f)$ given by Theorem \ref{thm:structure-W}
(associated to the complementary subspace \eqref{eq:basis-U} of $[f,\mf g]$),
and $Z(\partial)$ is the following differential operator, 
encoding the $\mf g^f$-basis \eqref{eq:basis-gf}:
\begin{equation}\label{eq:EndZ}
Z(\partial)
=
\sum_{a,b=1}^r
\sum_{i=0}^{\min\{p_a,p_b\}-1}
f_{ba;i}(-\partial)^i
E_{(a,1),(b,p_b)}
\,\in\mc V(\mf g^f)[\partial]\otimes\Hom(V_-,V_+)
\,.
\end{equation}

\subsection{
The ``identity'' notation
}
\label{sec:8a.2.5}

Let $U\subset V$ be a subspace of $V$,
and assume that there is ``natural'' splitting $V=U\oplus U^\prime$.
(Usually, $U$ is spanned by some basis elements $\{e_\alpha\}_{\alpha\in I_0}$, 
for some subset $I_0\subset I$;
in this case $U^\prime$ is the span of the remaining basis elements $\{e_\alpha\}_{\alpha\in I\backslash I_0}$.)
We shall denote, with an abuse of notation,
by $\id_U$ both the identity map $U\stackrel{\sim}{\longrightarrow}U$,
the inclusion map $U\hookrightarrow V$,
and the projection map (with kernel $U^\prime$) $V\twoheadrightarrow U$;
the correct meaning of $\id_U$ should be clear from the context.
Likewise,
if we further have a subspace $U_1\subset U$ with a ``natural'' splitting $U=U_1\oplus U_1^\prime$,
the same symbol $\id_{U_1}$
can mean not only the three maps 
identity $U_1\stackrel{\sim}{\longrightarrow}U_1$,
inclusion $U_1\hookrightarrow V$,
and projection (with kernel $U_1^\prime\oplus U^\prime$) $V\twoheadrightarrow U_1$,
but also 
the inclusion $U_1\hookrightarrow U$,
and the projection (with kernel $U_1^\prime$) $U\twoheadrightarrow U_1$;
again, the correct meaning of $\id_{U_1}$ should be clear from the context.
For example, $V_{\pm}\subset V$ come with the natural splittings
\begin{equation}\label{eq:splitting}
V=V_+\oplus FV=V_-\oplus F^TV
\,,
\end{equation}
and, with the notation described above, we have the obvious identities
\begin{equation}\label{20180219:eq3}
FF^T=\id_V-\id_{V_+}=\id_{FV}
\,\,\text{ and }\,\,
F^TF=\id_V-\id_{V_-}=\id_{F^TV}
\,.
\end{equation}

\subsection{
Generalized quasi-determinants and the Lax operator $\mc L(\partial)$
}
\label{sec:8a.3}

Let $R$ be a unital associative algebra and let $V$ be a finite-dimensional vector space,
with direct sum decompositions $V=U\oplus U^\prime=W\oplus W^\prime$.
Assume that $A\in R\otimes \End(V)$ is invertible.
Then, according to \cite[Prop.4.2]{DSKV16a},
$\id_WA^{-1}\id_U\in R\otimes\Hom(U,W)$
is invertible
(with inverse in $R\otimes\Hom(W,U)$)
if and only if 
$\id_{U^\prime}A\id_{W^\prime}\in R\otimes\Hom(W^\prime,U^\prime)$ 
is invertible (with inverse in $R\otimes\Hom(U^\prime,W^\prime)$),
and, in this case, we have
\begin{equation}\label{eq:quasidet}
|A|_{U,W}
:=
(\id_{W} A^{-1}\id_{U})^{-1}
=
\id_U A\id_W
-
\id_U A\id_{W^\prime}
(\id_{U^\prime} A\id_{W^\prime})^{-1}
\id_{U^\prime} A\id_W
\,\in R\otimes\Hom(W,U)
\,,
\end{equation}
which is called the (generalized) \emph{quasideterminant} of $A$ w.r.t. $U$ and $W$,
cf. \cite{GGRW05,DSKV16a}.
Recall also that, given direct sum decompositions
$U=U_1\oplus U_1^\prime$ and $W=W_1\oplus W_1^\prime$,
we have the following hereditary property of quasideterminants:
\begin{equation}\label{eq:hered}
\big||A|_{U,W}\big|_{U_1,W_1}
=
|A|_{U_1,W_1}
\,,
\end{equation}
provided that all quasideterminants exist.

\subsection{
The Lax operator $\mc L(\partial)$ as a quasideterminant
}
\label{sec:8a.4}

If we apply the map $\rho(=\rho\otimes\id)$, defined in \eqref{rho}, to the matrix differential operator 
$A(\partial)$, defined in \eqref{eq:Az}, we get
\begin{equation}\label{eq:rhoA}
\rho A(\partial)
=
\id_{V}\partial
+F
+\sum_{i\in J_{\leq\frac12}}u_iU^i
\,\in\mc V(\mf g_{\leq\frac12})[\partial]\otimes\End(V)
\,,
\end{equation}
where $u_i$, $i\in J_{\leq\frac12}$, are the basis elements of $\ad\chi$-eigenvalue 
less than or equal to $\frac12$.

Consider the spaces $V_{\pm,1}\subset V$ defined in Section \ref{sec:8a.2}.
They are both $r_1$-dimensional and, once we fix their bases \eqref{eq:Vpmi},
we can identify
$V_{+,1}\simeq V_{-,1}\simeq\mb C^{r_1}$,
and (cf. \eqref{eq:identif-pm})
\begin{equation}\label{eq:identif}
\Hom(V_{-,1},V_{+,1})\simeq\Mat_{r_1\times r_1}\mb C
\,.
\end{equation}
According to \cite[Thm.5.8]{DSKV16b},
the matrix pseudodifferential operator $\mc L(\partial)$, defined in \eqref{eq:WL},
can be obtained, under the identification \eqref{eq:identif},
as the quasideterminant of $\rho A(\partial)$
with respect to $V_{+,1}$ and $V_{-,1}$:
\begin{equation}\label{eq:L}
\mc L(\partial)
=
|\rho A(\partial)|_{V_{+,1},V_{-,1}}
\,\in\mc W(\mf{gl}_N,\ul p)((\partial^{-1}))\otimes\Hom(V_{-,1},V_{+,1})
\,.
\end{equation}
This result also follows from Proposition \ref{thm:recursion3}(b).

\section{
Algorithmic construction of the generator matrix $W(\partial)$
}
\label{sec:8b}

\subsection{
The matrix $T(\partial)$
}
\label{sec:8b.1}

\begin{proposition}\label{prop:T1}
The following quasideterminant exists and it is a differential operator:
\begin{equation}\label{eq:EndT}
T(\partial)
=
|\rho A(\partial)|_{V_+,V_-}
\,\in\mc V(\mf g_{\leq\frac12})[\partial]\otimes\Hom(V_-,V_+)
\,.
\end{equation}
Moreover, if we expand it in the standard basis $\{E_{(a,1),(b,p_b)}\}_{a,b=1}^r$ of $\Hom(V_-,V_+)$ 
as $T(\partial)=\sum_{a,b=1}^r t_{a,b}(\partial) E_{(a,1),(b,p_b)}$, 
then
\begin{equation}\label{eq:ordT}
t_{ab}(\partial)
=
-\delta_{a,b}(-\partial)^{p_a}
+\,\big(\text{order }<\frac{p_a+p_b}{2}\big)
\,.
\end{equation}
\end{proposition}
\begin{proof}
By \eqref{eq:rhoA}, $\rho A(\partial)$ is a monic differential operator of order $1$,
hence its inverse can be computed by geometric series expansion
in $\mc V(\mf g_{\leq\frac12})((\partial^{-1}))\otimes\End V$.
In order to compute the quasideterminant \eqref{eq:EndT},
we use the RHS of equation \eqref{eq:quasidet}:
\begin{equation}\label{eq:T1-1}
T(\partial)
=
\id_{V_+} \rho A(\partial) \id_{V_-}
-
\id_{V_+} \rho A(\partial) \id_{F^TV}
(\id_{FV} \rho A(\partial) \id_{F^TV})^{-1}
\id_{FV} \rho A(\partial) \id_{V_-}
\,,
\end{equation}
and, for its existence, we need to prove that
$\id_{FV}\rho A(\partial)\id_{F^TV}$ is invertible.
Let $\{E_{\alpha\beta}\}_{\alpha,\beta\in I}$, where $I$ is as in \eqref{eq:I}, 
be the standard basis of $\End V$
w.r.t. the basis of $V$ described in Section \ref{sec:8a.2},
and, according to the notational convention described in Section \ref{sec:8a.0},
let $\{e_{\alpha\beta}\}_{\alpha,\beta\in I}$ be the same collection of elements, 
viewed in $\mc V(\mf g)$.
In terms of these bases, \eqref{eq:rhoA} becomes, recalling \eqref{eq:x},
$$
\rho A(\partial)
=
\sum_{a=1}^r\sum_{i=1}^{p_a}E_{(a,i),(a,i)}\partial
+\sum_{a=1}^r\sum_{i=1}^{p_a-1}E_{(a,i+1),(a,i)}
+\sum_{a,b=1}^r
\sum_{\substack{1\leq i\leq p_a,1\leq j\leq p_b \\ \big(j-i\geq\frac{p_b-p_a-1}2\big)}}
e_{(b,j),(a,i)}E_{(a,i),(b,j)}
\,.
$$
Hence,
\begin{equation}\label{eq:T1-2}
\begin{split}
& \id_{V_+} \rho A(\partial) \id_{V_-}
=
\sum_{a\,|\,p_a=1}E_{(a,1),(a,1)}\partial
+\sum_{a,b=1}^r 
e_{(b,p_b),(a,1)}E_{(a,1),(b,p_b)}
\,, \\
& \id_{V_+} \rho A(\partial) \id_{F^TV}
=
\sum_{a\,|\,p_a>1}E_{(a,1),(a,1)}\partial
+
\sum_{a,b=1}^r
\!\!\!\!\!\!
\sum_{\substack{j=1 \\ \big(j\geq\frac{p_b-p_a+1}2\big)}}^{p_b-1}
\!\!\!\!\!\!
e_{(b,j),(a,1)}E_{(a,1),(b,j)}
\,, \\
& \id_{FV} \rho A(\partial) \id_{V_-}
=
\sum_{a\,|\,p_a>1}E_{(a,p_a),(a,p_a)}\partial
+\sum_{a,b=1}^r
\!\!\!\!\!\!
\sum_{\substack{i=2 \\ \big(i\leq\frac{p_b+p_a+1}2\big)}}^{p_a}
\!\!\!\!\!\!
e_{(b,p_b),(a,i)}E_{(a,i),(b,p_b)}
\,, \\
& \id_{FV} \rho A(\partial) \id_{F^TV}
=
\id_{FV} (F+\id_V\partial+\sum_{i\in J_{\leq\frac12}}u_iU^i) \id_{F^TV} 
=
\sum_{a=1}^r\sum_{i=1}^{p_a-1}E_{(a,i+1),(a,i)} \\
&\qquad 
+\sum_{a=1}^r\sum_{i=2}^{p_a-1}E_{(a,i),(a,i)}\partial
+\sum_{a,b=1}^r
\sum_{\substack{2\leq i\leq p_a,1\leq j\leq p_b-1 \\ \big(j-i\geq\frac{p_b-p_a-1}2\big)}}
e_{(b,j),(a,i)}E_{(a,i),(b,j)}
\,.
\end{split}
\end{equation}
Clearly, $\id_{FV}F\id_{F^TV}\in\Hom(F^TV,FV)$ is invertible,
with inverse
\begin{equation}\label{eq:F-1}
(\id_{FV}F\id_{F^TV})^{-1}
=
\id_{F^TV}F^T\id_{FV}
=
\sum_{a=1}^r\sum_{i=1}^{p_a-1}E_{(a,i),(a,i+1)}
\,\in\Hom(FV,F^TV)
\,.
\end{equation}
Also, note that the differential operator
\begin{align*}
& N(\partial)
:=
(\id_{FV}F\id_{F^TV})^{-1}
\id_{FV}(\id_V\partial+\sum_{i\in J}u_iU^i) \id_{F^TV} \\
& =
\sum_{a,b=1}^r
\sum_{\substack{1\leq i\leq p_a\!-\!1,1\leq j\leq p_b\!-\!1 \\ \big(j-i\geq\frac{p_b-p_a+1}2\big)}}
\!\!\!\!\!\!\!\!\!
\big(
\delta_{a,b}
\delta_{j,i+1}
\partial
+
e_{(b,j),(a,i+1)}
\big)
E_{(a,i),(b,j)}
\,\in
\mc V(\mf g_{\leq\frac12})[\partial]\otimes\End(F^TV)
\end{align*}
has strictly positive $\ad X$-eigenvalue,
hence it is nilpotent.
As a result,
$\id_{FV}\rho A(\partial)\id_{F^TV}$ is invertible,
and its inverse can be computed via a (finite) geometric series expansion,
\begin{equation}\label{eq:T1proof1}
\begin{split}
& (\id_{FV}\rho A(\partial)\id_{F^TV})^{-1}
=
\sum_{\ell=0}^\infty
(-1)^\ell
N(\partial)^\ell
(\id_{FV}F\id_{F^TV})^{-1} \\
& =
\sum_{\ell=0}^\infty (-1)^\ell 
\!\!\!\!\!\!
\sum_{a_0,\dots,a_\ell=1}^r
\!\!\!\!\!\!\!\!\!
\sum_{\substack{
i_0<p_{a_0},\dots,i_\ell<p_{a_\ell} \\
\big(
i_j-i_{j-1}\geq\frac{p_{a_j}-p_{a_{j-1}}+1}2
\,\forall j
\big)
}}
\!\!\!\!\!\!\!\!\!
\big(
\delta_{a_1,a_0}
\delta_{i_1,i_0+1}
\partial
+
e_{(a_1,i_1),(a_0,i_0+1)}
\big)
\dots \\
& \qquad \dots
\big(
\delta_{a_\ell,a_{\ell-1}}
\delta_{i_\ell,i_{\ell-1}+1}
\partial
+
e_{(a_\ell,i_\ell),(a_{\ell-1},i_{\ell-1}+1)}
\big)
E_{(a_0,i_0),(a_\ell,i_\ell+1)}
\,.
\end{split}
\end{equation}
Note that the above sum is finite since the conditions on $i_0,\dots,i_\ell$
imply
$$
i_\ell-i_0\geq\frac{p_{a_\ell}-p_{a_0}+\ell}2
\,,
$$
which becomes an empty condition for $\ell$ large enough.
This proves, in particular, the existence of the quasideterminant \eqref{eq:EndT},
which is the first claim of the proposition.
Combining equations \eqref{eq:T1-1}, \eqref{eq:T1-2} and \eqref{eq:T1proof1}, we get
\begin{equation}\label{eq:T1-3}
\begin{split}
& t_{ab}(\partial)
=
\delta_{a,b}\delta_{p_a,1}
\partial
+
e_{(b,p_b),(a,1)} \\
& -
\sum_{\ell=0}^\infty (-1)^\ell 
\!\!\!\!\!\!
\sum_{a_0,\dots,a_\ell=1}^r
\sum_{\mc I_\ell}
\big(
\delta_{a,a_0}\delta_{i_0,1}\partial
+
e_{(a_0,i_0),(a,1)}
\big)
\big(
\delta_{a_1,a_0}
\delta_{i_1,i_0+1}
\partial
+
e_{(a_1,i_1),(a_0,i_0+1)}
\big)
\dots \\
& \qquad \dots
\big(
\delta_{a_\ell,a_{\ell-1}}
\delta_{i_\ell,i_{\ell-1}+1}
\partial
+
e_{(a_\ell,i_\ell),(a_{\ell-1},i_{\ell-1}+1)}
\big)
\big(
\delta_{a_\ell,b}\delta_{i_\ell+1,p_b}\partial
+
e_{(b,p_b),(a_\ell,i_\ell+1)}
\big)
\,,
\end{split}
\end{equation}
where
$\mc I_\ell$ is the set of $\ell$-tuples of positive integers $(i_0,\dots,i_\ell)$ such that 
$i_j<p_{a_j}$ for all $j=0,\dots,\ell$, and
\begin{equation}\label{eq:T1-4}
i_0\geq\frac{p_{a_0}-p_a+1}2
\,,\,\,
i_j-i_{j-1}\geq\frac{p_{a_j}-p_{a_{j-1}}+1}2 \,\forall\, j=1,\dots,\ell
\,,\,\,
p_b-i_\ell\geq\frac{p_b-p_{a_\ell}+1}2
\,.
\end{equation}
The contribution to the coefficient of $\partial^n$ in \eqref{eq:T1-3}
comes from the summands with $\ell+2\geq n$ and with at least $n$ of the indices $(a_j,i_j)$ such that
$a_j=a_{j-1}$ and $i_j=i_{j-1}+1$, $j=0,\dots,\ell+1$,
where we let $(a_{-1},i_{-1})=(a,0)$ and $(a_{\ell+1},i_{\ell+1})=(b,p_b)$.
In this case, summing the remaining $\ell+2-n$ inequalities in \eqref{eq:T1-4}
we get
$$
p_b-n\geq\frac{p_b-p_a+\ell+2-n}2
\,.
$$
This implies, in particular, 
that $n\leq \frac{p_a+p_b}2$.
Moreover, the contribution to the coefficient of $\partial^n$ for $n=\frac{p_a+p_b}2$
can only come from the summand with $\ell+2=n$ and 
$a_j=a=b$, $i_j=j+1$, for all $j=0,\dots,\ell=p_a-2$,
which gives $-\delta_{a,b}(-\partial)^{p_a}$.
This proves \eqref{eq:ordT}.
\end{proof}
\begin{lemma}\label{lem:T1}
For every $\phi\in\mf g$, we have
\begin{equation}\label{20190619:eq3}
\{\phi_\lambda A(z)\}=A(z+\lambda)\Phi-\Phi A(z)
\,\
\,,
\end{equation}
where, according to the convention introduced in Section \ref{sec:8a.0},
$\Phi$ is the element $\phi\in\mf g$, viewed as an element of $\End V$.
\end{lemma}
\begin{proof}
By the definition \eqref{lambda} of the $\lambda$-bracket on the classical affine PVA $\mc V(\mf g)$
and the definition \eqref{eq:Az} of the matrix differential operator $A(\partial)$, we have
\begin{align*}
\{\phi_\lambda A(z)\}
=
\sum_{i\in J}\{\phi_\lambda u_i\}U^i
=
\sum_{i\in J}
\big(
[\phi,u_i]+(\phi|u_i)\lambda
\big)
U^i
=
\sum_{i\in J}
u_i[U^i,\Phi]+\Phi\lambda
=
A(z+\lambda)\Phi-\Phi A(z)
\,.
\end{align*}
\end{proof}
\begin{lemma}\label{lem:T2}
We have
\begin{equation}\label{eq:XY1}
X(\partial)
:=
\id_{F^TV}(\rho A)^{-1}(\partial)\id_{V_+}T(\partial)
\,\in\mc V(\mf g_{\leq\frac12})[\partial]\otimes(\Hom(V_-,F^TV))\big[\geq\frac12\big]
\,,
\end{equation}
and
\begin{equation}\label{eq:XY2}
Y(\partial)
:=
T(\partial) \id_{V_-}(\rho A)^{-1}(\partial)\id_{FV}
\,\in\mc V(\mf g_{\leq\frac12})[\partial]\otimes(\Hom(FV,V_+))\big[\geq\frac12\big]
\,,
\end{equation}
i.e., they are both differential operators of strictly positive $\ad X$-eigenvalues.
\end{lemma}
\begin{proof}
We start from the obvious identity $\id_V=\rho A(\partial)(\rho A)^{-1}(\partial)$.
Recalling the splittings \eqref{eq:splitting}, we get
\begin{align*}
& 0
=
\id_{FV}\rho A(\partial)(\rho A)^{-1}(\partial)\id_{V_+}
=
\id_{FV}\rho A(\partial)(\id_{V_-}+\id_{F^TV})(\rho A)^{-1}(\partial)\id_{V_+} \\
& =
\id_{FV}\rho A(\partial)\id_{V_-}(\rho A)^{-1}(\partial)\id_{V_+}
+
\id_{FV}\rho A(\partial)\id_{F^TV}(\rho A)^{-1}(\partial)\id_{V_+} \\
& =
\id_{FV}\rho A(\partial)\id_{V_-}T(\partial)^{-1}
+
\id_{FV}\rho A(\partial)\id_{F^TV}(\rho A)^{-1}(\partial)\id_{V_+} 
\,.
\end{align*}
Hence,
$$
X(\partial)
=
\id_{F^TV}(\rho A)^{-1}(\partial)\id_{V_+}T(\partial)
=
-(\id_{FV}\rho A(\partial)\id_{F^TV})^{-1}
\id_{FV}\rho A(\partial)\id_{V_-}
\,.
$$
By \eqref{eq:T1proof1}, we have 
$(\id_{FV}\rho A(\partial)\id_{F^TV})^{-1}\in\mc V(\mf g_{\leq\frac12})[\partial]\otimes\Hom(FV,F^TV)[\geq1]$.
On the other hand, 
we obviously have 
$\id_{FV}\rho A(\partial)\id_{V_-}
\rho A(\partial)\in\mc V(\mf g_{\leq\frac12})[\partial]\otimes\Hom(V_-,FV)[\geq-\frac12]$.
Claim \eqref{eq:XY1} follows.
Similarly, by the obvious identity $\id_V=(\rho A)^{-1}(\partial)\rho A(\partial)$ we have
$$
0
=
\id_{V_-}(\rho A)^{-1}(\partial)\rho A(\partial)\id_{F^TV}
=
T(\partial)^{-1}\id_{V_+}\rho A(\partial)\id_{F^TV}
+
\id_{V_-}(\rho A)^{-1}(\partial)\id_{FV}\rho A(\partial)\id_{F^TV} 
\,,
$$
from which we get
$$
Y(\partial)
=
T(\partial) \id_{V_-}(\rho A)^{-1}(\partial)\id_{FV}
=
-\id_{V_+}\rho A(\partial)\id_{F^TV}
(\id_{FV}\rho A(\partial)\id_{F^TV})^{-1}
\,.
$$
Claim \eqref{eq:XY2} follows again by \eqref{eq:T1proof1}.
\end{proof}
\begin{proposition}\label{prop:T2}
For every $\phi\in\mf g_{\geq\frac12}$,
the following identity holds:
\begin{equation}\label{eq:T2}
\rho\{\phi_\lambda T(z)\}
=
T(z+\lambda+\partial)\id_{V_-}\Phi
\big(
\id_{V_-}+X(z)
\big)
-
\big(
\id_{V_+}+Y(z+\lambda+\partial)
\big)
\Phi\id_{V_+} T(z)
\,,
\end{equation}
where $X(z)$ and $Y(z)$ are the symbols of the differential operators \eqref{eq:XY1} and \eqref{eq:XY2}.
\end{proposition}
\begin{proof}
Applying formula \eqref{eq:AB-1} twice, we get, by the definition \eqref{eq:EndT} of $T(z)$,
\begin{align*}
& \{\phi_\lambda T(z)\}
=
-T(z+\lambda+\partial)\{\phi_\lambda T^{-1}(z+x)\}(|_{x=\partial}T(z)) \\
& =
-T(z+\lambda+\partial)\id_{V_-}\{\phi_\lambda (\rho A)^{-1}(z+x)\}\id_{V_+}(|_{x=\partial}T(z))
\\
&=
T(z+\lambda+\partial)\id_{V_-}(\rho A)^{-1}(z+\lambda+\partial)
\{\phi_\lambda \rho A(z+x)\}(|_{x=\partial}\rho A^{-1}(z+\partial)\id_{V_+}T(z))
\,.
\end{align*}
We then apply equation \eqref{eq:Walg1} and Lemma \ref{lem:T1} to get
\begin{align*}
& \rho\{\phi_\lambda T(z)\}
=
T(z+\lambda+\partial)\id_{V_-}(\rho A)^{-1}(z+\lambda+\partial)
\rho\{\phi_\lambda A(z+x)\}(|_{x=\partial}\rho A^{-1}(z+\partial)\id_{V_+}T(z)) \\
& =
T(z+\lambda+\partial)\id_{V_-}(\rho A)^{-1}(z+\lambda+\partial)
\Big(
\rho A(z+\lambda+\partial)\Phi
-
\Phi \rho A(z+\partial)
\Big)
\rho A^{-1}(z+\partial)\id_{V_+}T(z) \\
& =
T(z+\lambda+\partial)\id_{V_-} \Phi \rho A^{-1}(z+\partial)\id_{V_+}T(z) 
-
T(z+\lambda+\partial)\id_{V_-}(\rho A)^{-1}(z+\lambda+\partial) \Phi \id_{V_+}T(z)
\,.
\end{align*}
Equation \eqref{eq:T2} follows by the definitions \eqref{eq:XY1} and \eqref{eq:XY2}
of $X(\partial)$ and $Y(\partial)$ and the definition \eqref{eq:EndT} of $T(\partial)$.
\end{proof}

\begin{proposition}\label{prop:T3}
We have:
$\pi_{\mf g^f}T(\partial)=-(-\partial)^{\ul p}+Z(\partial)$,
where $(-\partial)^{\ul p}$ is as in \eqref{eq:Dp2}, and $Z(\partial)$ 
is the differential operator \eqref{eq:EndZ}.
\end{proposition}
\begin{proof}
Since $\pi_{\mf g^f}:\,\mc V(\mf g)\to\mc V(\mf g^f)$ is a differential algebra homomorphism,
we have, by the definition \eqref{eq:EndT} of $T(\partial)$ and equation \eqref{eq:rhoA}:
\begin{equation}\label{eq:T3p1}
\pi_{\mf g^f}T(\partial)
=
|\rho_f A(\partial)|_{V_+,V_-}
\,\in\mc V(\mf g_{\leq\frac12})[\partial]\otimes\Hom(V_-,V_+)
\,,
\end{equation}
where
\begin{equation}\label{eq:T3p2}
\rho_f A(\partial)
:=
\pi_f(\rho A(\partial))
=
\id_{V}\partial
+F
+\sum_{a,b=1}^r\sum_{i=0}^{\min\{p_a,p_b\}-1}
f_{ba;i}\,E_{(a,1),(b,p_b-i)}
\,.
\end{equation}
Here we used the dual bases \eqref{eq:basis-gf} of $\mf g^f$ and \eqref{eq:basis-U} of $U$.
By \eqref{eq:T3p2} we immediately get (cf. \eqref{eq:T1-2})
\begin{equation}\label{eq:T3p3}
\begin{split}
& \id_{V_+}\rho_f A(\partial)\id_{V_-}
=
\sum_{a\,|\,p_a=1}E_{(a,1),(a,1)}\partial
+
\sum_{a,b=1}^r
f_{ba;0}\,E_{(a,1),(b,p_b)}
\,, \\
& \id_{V_+}\rho_f A(\partial)\id_{F^TV}
=
\sum_{a\,|\,p_a>1}E_{(a,1),(a,1)}\partial
+
\sum_{a,b=1}^r\sum_{i=1}^{\min\{p_a,p_b\}-1}
f_{ba;i}\,E_{(a,1),(b,p_b-i)}
\,, \\
& \id_{FV}\rho_f A(\partial)\id_{V_-}
=
\sum_{a\,|\,p_a>1}E_{(a,p_a),(a,p_a)}\partial
\,, \\
& \id_{FV}\rho_f A(\partial)\id_{F^TV}
=
\sum_{a=1}^r\sum_{i=1}^{p_a-1}E_{(a,i+1),(a,i)}
+
\sum_{a=1}^r\sum_{i=2}^{p_a-1}E_{(a,i),(a,i)}\partial
\,.
\end{split}
\end{equation}
Recalling \eqref{eq:F-1},
we can easily invert the last operator in \eqref{eq:T3p3} by geometric series expansion:
\begin{equation}\label{eq:T3p4}
\big(\id_{FV}\rho_f A(\partial)\id_{F^TV}\big)^{-1}
=
\sum_{a=1}^r
\sum_{1\leq i<j\leq p_a}
E_{(a,i),(a,j)}(-\partial)^{j-i-1}
\,.
\end{equation}
We then use equations \eqref{eq:T3p3} and \eqref{eq:T3p4},
and the formula \eqref{eq:quasidet} for the quasideterminant
(with $U=V_+$, $W=V_-$, and the complementary subspaces $U^\prime=FV$ and $W^\prime=F^TV$),
to get
\begin{align*}
& \pi_{\mf g^f}T(\partial)
=
\id_{V_+} \rho_f A(\partial) \id_{V_-}
-
\id_{V_+} \rho_f A(\partial) \id_{F^TV}
(\id_{FV} \rho_f A(\partial) \id_{F^TV})^{-1}
\id_{FV} \rho_f A(\partial) \id_{V_-} \\
& =
- 
\sum_{a=1}^r
E_{(a,1),(a,p_a)}
(-\partial)^{p_a}
+ 
\sum_{a,b=1}^r\sum_{i=1}^{\min\{p_a,p_b\}-1}
f_{ba;i}\,E_{(a,1),(b,p_b)}
(-\partial)^{i}
=
- (-\partial)^{\ul p}
+Z(\partial)
\,.
\end{align*}
\end{proof}

\subsection{
Motivational interlude
}
\label{sec:8b.2a}

The present section gives just a motivation for the recursive construction described in Section \ref{sec:8b.2b};
the Bourbakist reader can decide to skip it without any harm.

Our main goal it to find an explicit construction for the matrix differential operator 
$W(\partial)\in\mc W(\mf{gl}_N,\ul p)[\partial]\otimes\Hom(V_-,V_+)$
defined in \eqref{eq:matrW} (or, equivalently, \eqref{eq:EndW}),
encoding all the $\mc W$-algebra generators.
Note that equation \eqref{eq:WL} can be rewritten, in terms of a quasideterminant \eqref{eq:quasidet},
using the new form \eqref{eq:EndW} of $W(\partial)$, as
\begin{equation}\label{eq:motiv1}
\mc L(\partial)
=
|-(-\partial)^{\ul p}+W(\partial)|_{V_{+,1},V_{-,1}}
\,.
\end{equation}
Recall also, from Section \ref{sec:8a.2.6} that
\begin{equation}\label{eq:motiv2}
\pi_{\mf g^f} W(\partial)
=
Z(\partial)
\,.
\end{equation}
In fact, this equation defines $W(\partial)$ uniquely, due to the Structure Theorem \ref{thm:structure-W}.

On the other hand, in Section \ref{sec:8b.1} we introduced the matrix differential operator 
$T(\partial)\in\mc V(\mf{g}_{\leq\frac12})[\partial]\otimes\Hom(V_-,V_+)$.
By its definition \eqref{eq:EndT} and the hereditary property \eqref{eq:hered} of quasideterminants,
we have
\begin{equation}\label{eq:motiv3}
\mc L(\partial)
=
|T(\partial)|_{V_{+,1},V_{-,1}}
\,,
\end{equation}
while by Proposition \ref{prop:T3} we have
\begin{equation}\label{eq:motiv4}
\pi_{\mf g^f} T(\partial)
=
-(-\partial)^{\ul p}+Z(\partial)
\,.
\end{equation}

Comparing equations \eqref{eq:motiv1} and \eqref{eq:motiv3}, 
and equations \eqref{eq:motiv2} and \eqref{eq:motiv4},
the naive reader could guess that $T(\partial)$ coincides with $-(-\partial)^{\ul p}+W(\partial)$.
Of course this is not true, and there are two obstructions to it.
The first, theoretical, obstruction is that $T(\partial)$ does NOT have coefficients in the $\mc W$-algebra
$\mc W(\mf{gl}_N,\ul p)$, but just in the differential algebra $\mc V(\mf{g}_{\leq\frac12})$.
The second, more practical, obstruction is that the entries of the matrix $T(\partial)$, as differential operators,
do not have the same orders as the corresponding entries of the matrix $-(-\partial)^{\ul p}+W(\partial)$.
Indeed, 
if we expand both $T(\partial)$ and $-(-\partial)^{\ul p}+W(\partial)$ 
it in the standard basis $\{E_{(a,1),(b,p_b)}\}_{a,b=1}^r$ of $\Hom(V_-,V_+)$,
the coefficient of $E_{(a,1),(b,p_b)}$ in $T(\partial)$ is as in \eqref{eq:ordT},
while, recalling \eqref{eq:matrW}, the same coefficient in $-(-\partial)^{\ul p}+W(\partial)$ is of the form
\begin{equation}\label{eq:ordW}
-\delta_{a,b}(-\partial)^{p_a}
+\,\big(\text{order }\leq\min\{p_a,p_b\}-1\big)
\,.
\end{equation}
Of course the second obstruction can be easily solved by Gauss elimination,
via a recursive construction described in Section \ref{sec:8b.2b}.
The good news is that, in solving the second obstruction,
the first obstruction is resolved too,
and, as a result, we end up with the matrix $-(-\partial)^{\ul p}+W(\partial)$.
This will be proved in Section \ref{sec:8b.2c}.

To see how to remove the second obstruction, let us consider a ``toy example''.
Consider a $2\times 2$-matrix differential operator
$$
M(\partial)
=
\left(\begin{array}{ll}
M_{11}(\partial) & M_{12}(\partial) \\ 
M_{21}(\partial) & M_{22}(\partial)
\end{array}\right)
\in\Mat_{2\times2}\mc V[\partial]
\,,
$$
with monic diagonal entries $M_{11}(\partial)$ and $M_{22}(\partial)$.
We want to perform a Gauss elimination to end up with a new matrix $\widetilde M(\partial)$
with the off-diagonal entries 
of order strictly less than $M_{22}(\partial)$.
We perform divisions with reminders in the ring $\mc V[\partial]$:
$$
M_{12}(\partial)
=
Q_{12}(\partial)M_{22}(\partial)+\widetilde M_{12}(\partial)
\,\,,\,\,\,\,
M_{21}(\partial)
=
M_{22}(\partial)Q_{21}(\partial)+\widetilde M_{21}(\partial)
\,,
$$
with $\widetilde M_{12}(\partial)$ and $\widetilde M_{21}(\partial)$
of order strictly less than $M_{22}(\partial)$.
Since, by assumption, $M_{22}(\partial)$ is monic, it is invertible in $\mc V((\partial^{-1}))$,
and the above equations give
$$
M_{12}(\partial)M_{22}(\partial)^{-1}
=
Q_{12}(\partial)+\widetilde M_{12}(\partial)M_{22}(\partial)^{-1}
\,\,,\,\,\,\,
M_{22}(\partial)^{-1}M_{21}(\partial)
=
Q_{21}(\partial)+M_{22}(\partial)^{-1}\widetilde M_{21}(\partial)
\,.
$$
Note that $\widetilde M_{12}(\partial)M_{22}(\partial)^{-1}$ and $M_{22}(\partial)^{-1}\widetilde M_{21}(\partial)$
lie in $\mc V[[\partial^{-1}]]\partial^{-1}$.
Hence, we get
$$
Q_{12}(\partial)
=
\big(M_{12}(\partial)M_{22}(\partial)^{-1}\big)_+
\,\,,\,\,\,\,
Q_{21}(\partial)
=
\big(M_{22}(\partial)^{-1}M_{21}(\partial)\big)_+
\,,
$$
and therefore
$$
\widetilde M_{12}(\partial)
=
M_{12}(\partial)
-
\big(M_{12}(\partial)M_{22}(\partial)^{-1}\big)_+
M_{22}(\partial)
\,,\,\,
\widetilde M_{21}(\partial)
=
M_{21}(\partial)
-
M_{22}(\partial)
\big(M_{22}(\partial)^{-1}M_{21}(\partial)\big)_+
\,.
$$
In conclusion, we can get the desired matrix $\widetilde M(\partial)$
by the following elementary row and column operations:
\begin{equation}\label{eq:toy}
\widetilde M(\partial)
=
\left(\begin{array}{cc}
1 & -\big(M_{12}(\partial)M_{22}(\partial)^{-1}\big)_+ \\ 
0 & 1
\end{array}\right)
M(\partial)
\left(\begin{array}{cc}
1 & 0 \\ 
-\big(M_{22}(\partial)^{-1}M_{21}(\partial)\big)_+ & 1
\end{array}\right)
\,.
\end{equation}

\subsection{
Inductive construction of $T^{(k)}(\partial)$ and $W^{(k)}(\partial)$
}
\label{sec:8b.2b}

Starting with the operator $T(\partial)$ in \eqref{eq:EndT},
we define recursively, by downward induction,
two sequences of operators 
$T^{(k)}(\partial)$ and $W^{(k)}(\partial)$, $k=1,\dots,s$, 
where $s$ is defined in Section \ref{sec:8a.2}, as follows.
We let
$T^{(s)}(\partial)=T(\partial)$,
$W^{(s)}(\partial)=\id_{V_{+,s}}T(\partial)\id_{V_{-,s}}$,
and, for $1\leq k\leq s-1$, we let, inspired by \eqref{eq:toy},
\begin{equation}\label{eq:TWk}
\begin{split}
& T^{(k)}(\partial)
=
E_-^{(k)}(\partial)T^{(k+1)}(\partial)F_-^{(k)}(\partial)
\,\in\mc V(\mf g_{\leq\frac12})[\partial]\otimes\Hom(V_-,V_+)
\,, \\
& W^{(k)}(\partial)
=
\id_{V_{+,\geq k}}T^{(k)}(\partial)\id_{V_{-,\geq k}}
\,\in\mc V(\mf g_{\leq\frac12})[\partial]\otimes\Hom(V_{-,\geq k},V_{+,\geq k})
\,,
\end{split}
\end{equation}
where
\begin{equation}\label{eq:EFk}
\begin{split}
& E_{\pm}^{(k)}(\partial)
=
\id_{V_+}
\pm
\id_{V_{+,k}}
\big(
T^{(k+1)}(\partial)
W^{(k+1)}(\partial)^{-1}
\big)_+
\,\in\mc V(\mf g_{\leq\frac12})[\partial]\otimes\End(V_+)
\,, \\
& F_{\pm}^{(k)}(\partial)
=
\id_{V_-}
\pm
\big(
W^{(k+1)}(\partial)^{-1}
T^{(k+1)}(\partial)
\big)_+
\id_{V_{-,k}}
\,\in\mc V(\mf g_{\leq\frac12})[\partial]\otimes\End(V_-)
\,.
\end{split}
\end{equation}
In order to prove that the above operators are well defined,
we need to show that $W^{(k)}(\partial)$ is invertible as a pseudodifferential operator.
This is stated in the following proposition.
\begin{proposition}\label{thm:recursion1}
\begin{enumerate}[(a)]
\item
If we expand $T^{(k)}(\partial)$ and $W^{(k)}(\partial)$ in the standard basis 
$\{E_{(a,1),(b,p_b)}\}_{a,b=1}^r$ of $\Hom(V_-,V_+)$ 
as 
$$
T^{(k)}(\partial)=\sum_{a,b=1}^r t^{(k)}_{ab}(\partial) E_{(a,1),(b,p_b)}
\,\,\text{ and }\,\, 
W^{(k)}(\partial)=\sum_{a,b=R_{k-1}+1}^r t^{(k)}_{ab}(\partial) E_{(a,1),(b,p_b)}
\,,
$$
then (cf. \eqref{eq:ordT})
\begin{equation}\label{eq:ordTk}
t^{(k)}_{ab}(\partial)
=
-\delta_{a,b}(-\partial)^{p_a}
+\,\big(\text{order }<\frac{p_a+p_b}{2}\big)
\,,
\end{equation}
and
\begin{equation}\label{eq:ordWk}
t^{(k)}_{ab}(\partial)
=
-\delta_{a,b}(-\partial)^{p_a}
+\,\big(\text{order }\leq\min\{p_a,p_b\}-1\big)
\,\text{ if }\, a,b\geq R_{k-1}+1
\,.
\end{equation}
\item
$W^{(k)}(\partial)$ is invertible,
with inverse in $\mc V(\mf{g}_{\leq\frac12})((\partial^{-1}))\otimes\Hom(V_{+,\geq k},V_{-,\geq k})$.
If we expand its inverse in the standard basis of $\Hom(V_{+,\geq k},V_{-,\geq k})$ 
as
$$
W^{(k)}(\partial)^{-1}=\sum_{a,b=R_{k-1}+1}^r \omega^{(k)}_{ab}(\partial) E_{(a,p_a),(b,1)}
\,,
$$ 
then
\begin{equation}\label{eq:ordOk}
\omega^{(k)}_{ab}(\partial)
=
-\delta_{a,b}(-\partial)^{-p_a}
+\,\big(\text{order }\leq-\max\{p_a,p_b\}-1\big)
\,.
\end{equation}
\end{enumerate}
\end{proposition}
\begin{proof}
We prove the proposition by downward induction on $k$.
For $k=s$, we have $T^{(s)}(\partial)=T(\partial)$,
hence condition \eqref{eq:ordTk} is the same as \eqref{eq:ordT},
and condition \eqref{eq:ordWk} is the same as \eqref{eq:ordTk}
since, for $a,b\geq R_{s-1}+1$, we have $p_a=p_b\,(=p_r)$.
As a consequence, $W^{(s)}(\partial)$, in matrix form (cf. \eqref{eq:identif-pm}), 
has leading term $-\id_{r_r}(-\partial)^{p_s}$,
hence its inverse has leading term $-\id_{r_s}(-\partial)^{-p_s}$,
proving condition \eqref{eq:ordOk} for $k=s$.
Next, assume that conditions (a) and (b) hold for $T^{(k+1)}(\partial)$ and $W^{(k+1)}(\partial)$,
and we will prove them for $T^{(k)}(\partial)$ and $W^{(k)}(\partial)$.
In some sense, if we recall the motivation behind the recursive formulas \eqref{eq:TWk},
explained in Section \ref{sec:8b.2a}, these conditions hold by construction.
We give here a formal proof.
If $a,b\not\in\{R_{k-1}+1,\dots,R_k\}$, we have 
$t^{(k)}_{ab}(\partial)=t^{(k+1)}_{ab}(\partial)$,
hence \eqref{eq:ordTk} and \eqref{eq:ordWk} hold by inductive assumption.
If $a\in\{R_{k-1}+1,\dots,R_k\}$ and $b\not\in\{R_{k-1}+1,\dots,R_k\}$, we have,
by \eqref{eq:TWk},
\begin{equation}\label{eq:recur-matr1a}
t^{(k)}_{ab}(\partial)
=
t^{(k+1)}_{ab}(\partial)
-
\sum_{c,d=R_k+1}^r
\big(
t^{(k+1)}_{ac}(\partial)
\omega^{(k+1)}_{cd}(\partial)
\big)_+
t^{(k+1)}_{db}(\partial)
\,.
\end{equation}
By inductive assumption,
$t^{(k+1)}_{ab}(\partial)$ has order strictly bounded from above by $\frac{p_a+p_b}2$,
and each other summand in the RHS of \eqref{eq:recur-matr1a}
has order strictly bounded from above by
$\frac{p_a+p_c}2-\max\{p_c,p_d\}+\frac{p_d+p_b}2\leq\frac{p_a+p_b}2$,
proving condition \eqref{eq:ordTk}.
Moreover, if $b\geq R_{k}+1$, 
we have
$$		
\sum_{d=R_k+1}^r\omega^{(k+1)}_{cd}(\partial)t^{(k+1)}_{db}(\partial)=\delta_{cb}		
\,,		
$$
hence equation \eqref{eq:recur-matr1a} can be rewritten as
\begin{equation}\label{eq:recur-matr1b}
t^{(k)}_{ab}(\partial)
=
\sum_{c,d=R_k+1}^r
\big(
t^{(k+1)}_{ac}(\partial)
\omega^{(k+1)}_{cd}(\partial)
\big)_-
t^{(k+1)}_{db}(\partial)
\,,
\end{equation}
which, by the inductive assumption, has order strictly bounded from above by 
$\min\{p_d,p_b\}\leq p_b=\min\{p_a,p_b\}$,
proving condition \eqref{eq:ordWk}.
If $a\not\in\{R_{k-1}+1,\dots,R_k\}$ and $b\in\{R_{k-1}+1,\dots,R_k\}$, we have,
by \eqref{eq:TWk},
\begin{equation}\label{eq:recur-matr2}
t^{(k)}_{ab}(\partial)
=
t^{(k+1)}_{ab}(\partial)
-
\sum_{c,d=R_k+1}^r
t^{(k+1)}_{ac}(\partial)
\big(
\omega^{(k+1)}_{cd}(\partial)
t^{(k+1)}_{db}(\partial)
\big)_+
\,,
\end{equation}
and conditions \eqref{eq:ordTk} and \eqref{eq:ordWk} are proved in the same way.
Finally, if $a,b\in\{R_{k-1}+1,\dots,R_k\}$ equation \eqref{eq:TWk} gives
\begin{equation}\label{eq:recur-matr}
\begin{split}
& t^{(k)}_{ab}(\partial)
=
t^{(k+1)}_{ab}(\partial)
+
\sum_{c,d,c',d'=R_k+1}^r
\big(
t^{(k+1)}_{ac}(\partial)
\omega^{(k+1)}_{cd}(\partial)
\big)_+
t^{(k+1)}_{dc'}(\partial)
\big(
\omega^{(k+1)}_{c'd'}(\partial)
t^{(k+1)}_{d'b}(\partial)
\big)_+ \\
& -
\sum_{c,d=R_k+1}^r
\big(
t^{(k+1)}_{ac}(\partial)
\omega^{(k+1)}_{cd}(\partial)
\big)_+
t^{(k+1)}_{db}(\partial) 
-
\sum_{c,d=R_k+1}^r
t^{(k+1)}_{ac}(\partial)
\big(
\omega^{(k+1)}_{cd}(\partial)
t^{(k+1)}_{db}(\partial)
\big)_+
\,.
\end{split}
\end{equation}
In this case, $p_a=p_b$, hence, by inductive assumption,
$t^{(k+1)}_{ab}(\partial)=-\delta_{a,b}(-\partial)^{p_a}+$ order $\leq p_a-1$,
and all other summands in the RHS of \eqref{eq:recur-matr}
have order strictly bounded by $p_a$, proving \eqref{eq:ordWk}.
Next, let us prove claim (b).
The matrix 
$$
-\sum_{a=R_{k-1}+1}^r(-\partial)^{p_a} E_{(a,1),(a,p_a)}
\in\mb C[\partial]\otimes\Hom(V_{-,\geq k},V_{+,\geq k})
$$
is clearly invertible, with inverse
$$
-\sum_{a=R_{k-1}+1}^r(-\partial)^{-p_a} E_{(a,p_a),(a,1)}
\in\mb C((\partial^{-1}))\otimes\Hom(V_{+,\geq k},V_{-,\geq k})
\,.
$$
The inverse of $W^{(k)}(\partial)$
in $\mc V(\mf g_{\leq\frac12})((\partial^{-1}))\otimes\Hom(V_{+,\geq k},V_{-,\geq k})$
can thus be computed by the geometric series expansion,
and the conditions \eqref{eq:ordOk} on the order of the matrix elements 
are immediate consequence of this geometric expansion and of condition \eqref{eq:ordWk}.
\end{proof}
\begin{remark}\label{rem:BK}
In the context of finite $\mc W$-algebras, the matrix $T(\partial)$ defined in \eqref{eq:EndT} appeared
in \cite{BK06} (see also \cite{DSFV19} for further details). The inductive construction described in Section
\ref{sec:8b.2b} is analogue to the construction of generators of finite $\mc W$-algebras in type $A$ using
Gauss factorization of the matrix $T(\partial)$ performed in \cite{BK06}.
\end{remark}

\subsection{
Properties of $T^{(k)}(\partial)$ and $W^{(k)}(\partial)$
}
\label{sec:8b.2c}

\begin{proposition}\label{thm:recursion3}
\begin{enumerate}[(a)]
\item
If $k\leq \ell$, we have
$\id_{V_{+,\geq \ell}}T^{(k)}(\partial)\id_{V_{-,\geq \ell}}
=W^{(\ell)}(\partial)$.
\item
For every $k=1,\dots,s$, we have $|T^{(k)}(\partial)|_{V_{+,1},V_{-,1}}=\mc L(\partial)$.
\item
For every $k=1,\dots,s$, we have $\pi_{\mf g^f}T^{(k)}(\partial)=-(-\partial)^{\ul p}+Z(\partial)$.
\end{enumerate}
\end{proposition}
\begin{proof}
Claim (a) for $k=\ell$ holds by construction.
For $k<\ell$, note that 
$\id_{V_{\pm,\geq \ell}}=\id_{V_{\pm,\geq \ell}}\id_{V_{\pm,\geq k+1}}$,
hence it suffices to prove the claim for $\ell=k+1$.
For this, we have, by \eqref{eq:TWk}
\begin{align*}
& \id_{V_{+,\geq k+1}}T^{(k)}(\partial)\id_{V_{-,\geq k+1}} 
=
\id_{V_{+,\geq k+1}}E_-^{(k)}(\partial)T^{(k+1)}(\partial)F_-^{(k)}(\partial)\id_{V_{-,\geq k+1}} \\
& =
\id_{V_{+,\geq k+1}}T^{(k+1)}(\partial)\id_{V_{-,\geq k+1}}
=
W^{(k+1)}(\partial)
\,,
\end{align*}
since, obviously, $\id_{V_{+,\geq k+1}}E_-^{(k)}(\partial)=\id_{V_{+,\geq k+1}}$
and $F_-^{(k)}(\partial)\id_{V_{-,\geq k+1}}=\id_{V_{-,\geq k+1}}$.

Next, we prove claim (b) by downward induction on $k$.
For $k=s$ it holds by \eqref{eq:motiv3}.
For $k=1,\dots,s-1$, we have, 
by the definition \eqref{eq:quasidet} of quasideterminant
and by \eqref{eq:TWk}, 
\begin{align*}
& |T^{(k)}(\partial)|_{V_{+,1},V_{-,1}}
=
\big(\id_{V_{-,1}}T^{(k)}(\partial)^{-1}\id_{V_{+,1}}\big)^{-1}
=
\big(
\id_{V_{-,1}}
F_+^{(k)}(\partial)
T^{(k+1)}(\partial)^{-1}
E_+^{(k)}(\partial)
\id_{V_{+,1}}
\big)^{-1} \\
& =
\big(
\id_{V_{-,1}}
T^{(k+1)}(\partial)^{-1}
\id_{V_{+,1}}
\big)^{-1}
=
|T^{(k+1)}(\partial)|_{V_{+,1},V_{-,1}}
=
\mc L(\partial)
\,,
\end{align*}
by the inductive assumption.
For the third equality we used 
the obvious identities 
$$
\id_{V_{-,1}}W^{(k+1)}(\partial)^{-1}=0=W^{(k+1)}(\partial)^{-1}\id_{V_{+,1}}		
\,.		
$$

Finally, we prove claim (c).
For $k=s$ it holds by Proposition \eqref{prop:T3}.
Let $k=1,\dots,s-1$.
By the inductive assumption, $\pi_{\mf g^f}(T^{(k+1)}(\partial))=-(-\partial)^{\ul p}+Z(\partial)$.
Recalling the matrix form \eqref{eq:Dp2}-\eqref{eq:EndZ} of $-(-\partial)^{\ul p}+Z(\partial)$
and the matrix form \eqref{eq:ordOk} of $W^{(k+1)}(\partial)^{-1}$,
we immediately have that 
$(-(-\partial)^{\ul p}+Z(\partial))W^{(k+1)}(\partial)^{-1}$
and 
$W^{(k+1)}(\partial)^{-1}(-(-\partial)^{\ul p}+Z(\partial))$
have negative order in $\partial$.
Hence
$$
\pi_{\mf g^f}
\big(
T^{(k+1)}(\partial)
W^{(k+1)}(\partial)^{-1}
\big)_+
=0
\,\,\text{ and }\,\,
\pi_{\mf g^f}
\big(
W^{(k+1)}(\partial)^{-1}
T^{(k+1)}(\partial)
\big)_+
=0
\,,
$$
so that, by \eqref{eq:EFk}, $\pi_{\mf g^f}E_{\pm}^{(k)}(\partial)=\id_{V_+}$
and $\pi_{\mf g^f}F_{\pm}^{(k)}(\partial)=\id_{V_-}$.
As a consequence, by \eqref{eq:TWk}
$$
\pi_{\mf g^f}T^{(k)}(\partial)
=
\pi_{\mf g^f}T^{(k+1)}(\partial)
=
-(-\partial)^{\ul p}+Z(\partial)
\,,
$$
proving claim (c).
\end{proof}
For every $\phi\in\mf g_{\geq\frac12}$, 
we introduce two auxiliary sequences of operators.
Recalling \eqref{eq:XY1}, \eqref{eq:XY2} and \eqref{eq:T2}, we let
\begin{equation}\label{eq:XYs}
\begin{split}
& X_\phi^{(s)}(\lambda,\partial)
=
\id_{V_-}\Phi(\id_{V_-}+X(\partial))
\,\in\mc V(\mf g_{\leq\frac12})[\lambda,\partial]\otimes\End(V_-)
\,, \\
& Y_\phi^{(s)}(\lambda,\partial)
=
-(\id_{V_+}+Y(\lambda+\partial))\Phi\id_{V_+}
\,\in\mc V(\mf g_{\leq\frac12})[\lambda,\partial]\otimes\End(V_+)
\,,
\end{split}
\end{equation}
and, for $1\leq k\leq s-1$, we let, by downward induction on $k$,
\begin{equation}\label{eq:XYk}
\begin{split}
& X_\phi^{(k)}(\lambda,\partial)
=
F_+^{(k)}(\lambda+\partial)X_\phi^{(k+1)}(\lambda,\partial)\id_{V_-,<k}
\,\in\mc V(\mf g_{\leq\frac12})[\lambda,\partial]\otimes\End(V_-)
\,, \\
& Y_\phi^{(k)}(\lambda,\partial)
=
\id_{V_{+,<k}}
Y_\phi^{(k+1)}(\lambda,\partial)
E_+^{(k)}(\partial)
\,\in\mc V(\mf g_{\leq\frac12})[\lambda,\partial]\otimes\End(V_+)
\,.
\end{split}
\end{equation}
\begin{proposition}\label{thm:recursion2}
\begin{enumerate}[(a)]
\item
For every $\phi\in\mf g_{\geq\frac12}$ and $k=1,\dots,s$, 
the operators $X^{(k)}(\partial)$ and $Y^{(k)}(\partial)$
have positive $\ad X$-eigenvalues.
\item
The following identity holds
\begin{equation}\label{eq:phik}
\rho\{\phi_\lambda T^{(k)}(z)\}
=
T^{(k)}(\lambda+z+\partial)X_\phi^{(k)}(\lambda,z)
+Y_\phi^{(k)}(\lambda,z+\partial)T^{(k)}(z)
\,.
\end{equation}
\item
We have
$W^{(k)}(\partial)
\in\mc W(\mf{gl}_N,\ul p)[\partial]\otimes\Hom(V_{-,\geq k},V_{+,\geq k})$.
\end{enumerate}
\end{proposition}
\begin{proof}
First, we prove all three claims for $k=s$.
By assumption, $\Phi$ has positive $\ad X$-eigenvalue,
and, by Lemma \ref{lem:T2}, $X(\partial)$ and $Y(\partial)$ have positive $\ad X$-eigenvalues.
As a consequence, $X^{(s)}_\phi(\partial)$ and $Y^{(s)}_\phi(\partial)$ 
have positive $\ad X$-eigenvalues as well, proving (a).
By the definition \eqref{eq:XYs} of $X^{(s)}_\phi(\lambda,\partial)$ and $Y^{(s)}_\phi(\lambda,\partial)$,
equation \eqref{eq:phik} for $k=s$ is the same as \eqref{eq:T2}, i.e. claim (b) holds.
For (c), note that 
$X^{(s)}_\phi(\lambda,\partial)\id_{V_{-,s}}=0$ and $\id_{V_{+,s}}Y^{(s)}_\phi(\lambda,\partial)=0$,
since $X^{(s)}_\phi$ and $Y^{(s)}_\phi$ have positive $\ad X$-eigenvalues.
Hence, by \eqref{eq:phik} with $k=s$, we get
$\rho\{\phi_\lambda W^{(s)}(z)\}
=
\id_{V_{+,s}}\rho\{\phi_\lambda T^{(s)}(z)\}\id_{V_{-,s}}=0$,
proving claim (c) for $k=s$.

Next, we fix $k=1,\dots,s-1$. We prove claims (a), (b) and (c) by downward induction.
By the inductive assumption,
$X^{(k+1)}_\phi(\partial)$ and $Y^{(k+1)}_\phi(\partial)$
have positive $\ad X$-eigenvalues.
On the other hand, 
recalling \eqref{eq:EFk},
$E^{(k)}_{\pm}(\partial)$ and $F^{(k)}_{\pm}(\partial)$
have non-negative $\ad X$-eigenvalues,
since, obviously, $\Hom(V_{+,\geq k+1},V_{+,k}),\,\Hom(V_{-,k},V_{-,\geq k+1})\subset(\End V)[>0]$.
As a result, 
$X^{(k)}_\phi(\partial)$ and $Y^{(k)}_\phi(\partial)$
have strictly positive $\ad X$-eigenvalues as well,
proving claim (a).
Next, we prove claim (b).
By \eqref{eq:TWk} and the PVA axioms, we have
\begin{equation}\label{eq:phi-p1}
\begin{split}
& \rho\{\phi_\lambda T^{(k)}(z)\}
=
\rho\big\{\phi_\lambda 
E_-^{(k)}(z+x)T^{(k+1)}(z+\partial)F_-^{(k)}(z)
\big\} \\
& =
\rho\{\phi_\lambda 
E_-^{(k)}(z+x)
\}
\big(\big|_{x=\partial}T^{(k+1)}(z+\partial)F_-^{(k)}(z)\big) \\
& +
E_-^{(k)}(\lambda+z+\partial)
\rho\{\phi_\lambda 
T^{(k+1)}(z+x)
\}
\big(\big|_{x=\partial}F_-^{(k)}(z)\big) \\
& +
E_-^{(k)}(\lambda+z+\partial)T^{(k+1)}(\lambda+z+\partial)
\rho\{\phi_\lambda 
F_-^{(k)}(z)
\}
\,.
\end{split}
\end{equation}
Note that the operators $E_{\pm}^{(k)}(\partial)$ are inverse to each other,
and the operators $F_{\pm}^{(k)}(\partial)$ are inverse to each other.
We then use
the inductive assumption \eqref{eq:phik} on 
$\rho\{\phi_\lambda T^{(k+1)}(z)\}$
and equations \eqref{eq:TWk}
to rewrite the RHS of \eqref{eq:phi-p1} as
\begin{equation}\label{eq:phi-p2}
\begin{split}
& \rho\{\phi_\lambda T^{(k)}(z)\}
=
\rho\{\phi_\lambda 
E_-^{(k)}(z+x)
\}
\big(\big|_{x=\partial}
E_+^{(k)}(z+\partial)T^{(k)}(z)
\big) \\
& +
T^{(k)}(\lambda+z+\partial)
F_+^{(k)}(\lambda+z+\partial)
X_\phi^{(k+1)}(\lambda,z+\partial)
F_-^{(k)}(z) \\
& +
E_-^{(k)}(\lambda+z+\partial)
Y_\phi^{(k+1)}(\lambda,z+\partial)
E_+^{(k)}(z+\partial)T^{(k)}(z) \\
& +
T^{(k)}(\lambda+z+\partial)F_+^{(k)}(\lambda+z+\partial)
\rho\{\phi_\lambda 
F_-^{(k)}(z)
\}
\,.
\end{split}
\end{equation}
Note that equation \eqref{eq:phi-p2} has the form \eqref{eq:phik}
with
\begin{equation}\label{eq:phi-p3}
\begin{split}
& X_\phi^{(k)}(\lambda,\partial)
=
F_+^{(k)}(\lambda+\partial)
X_\phi^{(k+1)}(\lambda,\partial)
F_-^{(k)}(\partial)
+
F_+^{(k)}(\lambda+\partial)
\rho\{\phi_\lambda F_-^{(k)}(\partial)\} 
\,,\\
& Y_\phi^{(k)}(\lambda,\partial)
=
\rho\{\phi_\lambda E_-^{(k)}(\partial) \}
E_+^{(k)}(\partial)
+
E_-^{(k)}(\lambda+\partial)
Y_\phi^{(k+1)}(\lambda,\partial)
E_+^{(k)}(\partial)
\,.
\end{split}
\end{equation}
In order to complete the proof, we are left to show that \eqref{eq:XYk} 
and \eqref{eq:phi-p3} coincide.
Equivalently, we need to prove the following two identities
\begin{equation}\label{eq:phi-p4}
\begin{split}
& 
X_\phi^{(k+1)}(\lambda,\partial)
F_-^{(k)}(\partial)
+
\rho\{\phi_\lambda F_-^{(k)}(\partial)\} 
=
X_\phi^{(k+1)}(\lambda,\partial)\id_{V_-,<k}
\,,\\
& 
\rho\{\phi_\lambda E_-^{(k)}(\partial) \}
+
E_-^{(k)}(\lambda+\partial)
Y_\phi^{(k+1)}(\lambda,\partial)
=
\id_{V_{+,<k}}
Y_\phi^{(k+1)}(\lambda,\partial)
\,.
\end{split}
\end{equation}
By the inductive assumption \eqref{eq:XYk}, 
we have 
$X^{(k+1)}_\phi(\lambda,\partial)\id_{V_{-,\geq k+1}}=0$
and
$\id_{V_{-,\geq k+1}}Y^{(k+1)}_\phi(\lambda,\partial)=0$.
On the other hand, 
we obviously have $W^{(k+1)}(\partial)^{-1}=\id_{V_{-,\geq k+1}}W^{(k+1)}(\partial)^{-1}\id_{V_{+,\geq k+1}}$.
Hence
\begin{equation}\label{eq:phi-p8}
X_\phi^{(k+1)}(\lambda,\partial)W^{(k+1)}(\partial)^{-1}
=
0
\,\,\text{ and }\,\,
W^{(k+1)}(\partial)^{-1}
Y_\phi^{(k+1)}(\lambda,\partial)
=
0
\,.
\end{equation}
Hence, by the definition \eqref{eq:EFk} of $E^{(k)}_-(\partial)$ and $F^{(k)}_-(\partial)$, we have
\begin{equation}\label{eq:phi-p5}
X_\phi^{(k+1)}(\lambda,\partial)F_-^{(k)}(\partial)
=
X_\phi^{(k+1)}(\lambda,\partial)
\,\,\text{ and }\,\,
E_-^{(k)}(\lambda+\partial)
Y_\phi^{(k+1)}(\lambda,\partial)
=
Y_\phi^{(k+1)}(\lambda,\partial)
\,.
\end{equation}
Furthermore, by the definition \eqref{eq:EFk} of $E^{(k)}_-(\partial)$,
and the left Leibniz rule, we have
\begin{equation}\label{eq:phi-p6}
\begin{split}
& \rho\{\phi_\lambda E_-^{(k)}(\partial) \}
=
-
\id_{V_{+,k}}
\big(
\rho\{\phi_\lambda 
T^{(k+1)}(\partial)
\}
W^{(k+1)}(\partial)^{-1}
\big)_+ \\
& =
-
\big(
\id_{V_{+,k}}
Y_\phi^{(k+1)}(\lambda,\partial)T^{(k+1)}(\partial)
W^{(k+1)}(\partial)^{-1}
\big)_+ 
=
-
\id_{V_{+,k}}
Y_\phi^{(k+1)}(\lambda,\partial)
\,.
\end{split}
\end{equation}
For the first equality in \eqref{eq:phi-p6} 
we used the fact that, by the inductive assumption (c), $W^{(k+1)}(\partial)$
has coefficients in the $\mc W$-algebra $\mc W(\mf{gl}_N,\ul p)$.
For the second equality in \eqref{eq:phi-p6} we used the inductive assumption \eqref{eq:phik}
and the first equation \eqref{eq:phi-p8}.
For the last equality in \eqref{eq:phi-p6} we used the facts that, 
by (a), $Y_\phi^{(k+1)}(\lambda,\partial)$ has positive $\ad X$-eigenvalues,
so that
$$
\id_{V_{+,k}}Y_\phi^{(k+1)}(\lambda,\partial)
=
\id_{V_{+,k}}Y_\phi^{(k+1)}(\lambda,\partial)\id_{V_{+,\leq k}}
\,,
$$
and that, by the definition \eqref{eq:TWk} of $W^{(k+1)}(\partial)$,
$$
\id_{V_{+,\leq k}}T^{(k+1)}(\partial)W^{(k+1)}(\partial)^{-1}
=
W^{(k+1)}(\partial)W^{(k+1)}(\partial)^{-1}
=
\id_{V_{+,\leq k}}
\,.
$$
Similarly,
\begin{equation}\label{eq:phi-p7}
\begin{split}
& \rho\{\phi_\lambda F_-^{(k)}(\partial) \}
=
-
\big(
W^{(k+1)}(\lambda+\partial)^{-1}
\rho\{\phi_\lambda T^{(k+1)}(\partial) \}
\id_{V_{-,k}}
\big)_+ \\
& =
-
\big(
W^{(k+1)}(\lambda+\partial)^{-1}
T^{(k+1)}(\lambda+\partial)
X_\phi^{(k+1)}(\lambda,\partial)
\id_{V_{-,k}}
\big)_+ 
=
-
X_\phi^{(k+1)}(\lambda,\partial)
\id_{V_{-,k}}
\,.
\end{split}
\end{equation}
Combining \eqref{eq:phi-p5} and \eqref{eq:phi-p7}, we get the first equation in \eqref{eq:phi-p4},
while combining \eqref{eq:phi-p5} and \eqref{eq:phi-p6}, we get the second equation in \eqref{eq:phi-p4}.
This proves claim (b).
Finally, we prove claim (c).
By \eqref{eq:XYk}, we have 
$X^{(k)}_\phi(\lambda,\partial)\id_{V_{-,\geq k}}=0$
and
$\id_{V_{-,\geq k}}Y^{(k)}_\phi(\lambda,\partial)=0$.
Hence, by \eqref{eq:phik},
\begin{align*}
& \rho\{\phi_\lambda W^{(k)}(z)\}
=
\id_{V_{+,\geq k}}\rho\{\phi_\lambda T^{(k)}(z)\}\id_{V_{-,\geq k}} \\
& =
\id_{V_{+,\geq k}}
T^{(k)}(\lambda+z+\partial)X_\phi^{(k)}(\lambda,z)
\id_{V_{-,\geq k}}
+
\id_{V_{+,\geq k}}
Y_\phi^{(k)}(\lambda,z+\partial)T^{(k)}(z)
\id_{V_{-,\geq k}}
=0
\,.
\end{align*}
\end{proof}
\begin{corollary}\label{cor:W}
$W^{(1)}(\partial)=-(-\partial)^{\ul p}+W(\partial)$.
\end{corollary}
\begin{proof}
By construction $W^{(1)}(\partial)=T^{(1)}(\partial)$.
By Proposition \ref{thm:recursion2}(c), 
$W^{(1)}(\partial)$ has coefficients in the $\mc W$-algebra
$\mc W(\mf{gl}_N,\ul p)$.
By Proposition \ref{thm:recursion3}(c), 
we have $\pi_{\mf g^f}T^{(1)}(\partial)=-(-\partial)^{\ul p}+Z(\partial)$.
Hence, by the Structure Theorem \ref{thm:structure-W},
$$
W^{(1)}(\partial)
=w(\pi(W^{(1)}(\partial)))
=w(\pi_{\mf g^f}(T^{(1)}(\partial)))
=-(-\partial)^{\ul p}+w(Z(\partial))
=-(-\partial)^{\ul p}+W(\partial)
\,.
$$
\end{proof}

\section{
The matrix \texorpdfstring{$W_{\bm2\bm2}(\partial)$}{W_{22}(d)} does not evolve
}
\label{sec:8c}

Recalling the basis $\{e_\alpha\}_{\alpha\in I}$ of $V$ defined in Section \ref{sec:8a.2},
we have the direct sum decompositions
\begin{equation}\label{eq:decomp2}
V=V_{\pm,1}\oplus V_{\pm,1}^\prime
\,,
\end{equation}
where $V_{+,1}^\prime=\Span\big\{e_{(a,h)}\,\big|\,(a,h)\in I\,\text{ s.t. }\,h\neq 1 \text{ if } a\leq r_1\big\}$,
and $V_{-,1}^\prime=\Span\big\{e_{(a,h)}\,\big|\,(a,h)\in I\,\text{ s.t. }\,h\neq p_a \text{ if } a\leq r_1\big\}$.
Consider the subspace $\Hom(V_{+,1}^\prime,V_{-,1}^\prime)\subset\End V$,
and let $\mf g^\prime\subset\mf g=\mf{gl}(V)$ be the same subspace,
viewed as a subspace of the Lie algebra $\mf g$,
and hence of the differential algebra $\mc V(\mf g)$:
\begin{equation}\label{eq:gprime}
\mf g^\prime
:=
\Hom(V_{+,1}^\prime,V_{-,1}^\prime)
\subset
\mc V(\mf g)
\,.
\end{equation}
We also denote 
\begin{equation}\label{eq:gprime12}
\mf g^\prime_{\leq\frac12}=\mf g^\prime\cap\mf g_{\leq\frac12}\subset\mc V(\mf g_{\leq\frac12})
\,.
\end{equation}
Recalling the definition \eqref{eq:Az} of the differential operator $A(\partial)$,
it is immediate to see that a basis of $\mf g^\prime$ is provided by the matrix entries
of the constant term of the operator
\begin{equation}\label{eq:Aprime}
\id_{V_{+,1}^\prime}A(\partial)\id_{V_{-,1}^\prime}
\,\in
\mc V(\mf g^\prime)[\partial]\otimes\Hom(V_{-,1}^\prime,V_{+,1}^\prime)
\,.
\end{equation}
\begin{lemma}\label{lem:dan1}
For every $v\in\mc V(\mf g^\prime)$, we have
$\{v_\lambda\id_{V_{-,1}}A^{-1}(w)\id_{V_{+,1}}\}=0$.
\end{lemma}
\begin{proof}
By equation \eqref{20190410:eq2}, we have
\begin{align*}
& \big\{(\id_{V_{+,1}^\prime}A(z)\id_{V_{-,1}}^\prime)_{\lambda}(\id_{V_{-,1}}A^{-1}(w)\id_{V_{+,1}})\big\} 
=
(\id_{V_{+,1}^\prime}\otimes\id_{V_{-,1}})
\big\{A(z)_{\lambda}A^{-1}(w)\big\}
(\id_{V_{-,1}}^\prime\otimes\id_{V_{+,1}}) \\
& =
\Omega_V
\Big(\id_{V_{-,1}}A^{-1}(w+\lambda+\partial)A(z)\id_{V_{-,1}}^\prime
\otimes
\id_{V_{+,1}^\prime}\id_{V_{+,1}}
\Big)
(z-w-\lambda)^{-1} \\
& -
(z-w-\lambda-\partial)^{-1}
\Big(\id_{V_{+,1}^\prime}A^{*,1}(\lambda-z)A^{-1}(w)\id_{V_{+,1}}
\otimes
\id_{V_{-,1}}\id_{V_{-,1}}^\prime
\Big)
\Omega_V
=0
\,,
\end{align*}
since, obviously, $\id_{V_{\pm,1}^\prime}\id_{V_{\pm,1}}=0$.
The claim follows since, as observed in \eqref{eq:Aprime},
the coefficients of the entries of $\id_{V_{+,1}^\prime}A(z)\id_{V_{-,1}}^\prime$
span $\mf g^\prime$.
\end{proof}
\begin{proposition}\label{prop:dan1}
An element $v\in\mc V(\mf g^\prime_{\leq\frac12})\cap\mc W(\mf{gl}_N,\ul p)$
does not evolve w.r.t. the time evolution given by the Hamiltonian flows \eqref{eq:hameq}:
$\frac{\partial v}{\partial t_j}=0$, for all $j\in\mb Z_{\geq1}$.
\end{proposition}
\begin{proof}
By the definition \eqref{20120511:eq3} of the $\mc W$-algebra $\lambda$-bracket
and equation \eqref{eq:L}, we have
$$
\{v_\lambda\mc L^{-1}(z)\}^{\mc W}
=
\{v_\lambda
\id_{V_{-,1}}
(\rho A)^{-1}(z)
\id_{V_{+,1}}
\}^{\mc W}
=
\rho
\{v_\lambda \id_{V_{-,1}}A^{-1}(z)\id_{V_{+,1}} \}
=
0
\,.
$$
For the first equality we used the definition \eqref{eq:quasidet} of quasideterminant,
for the second equality we used equation \eqref{eq:Walg2},
and the last equality is due to Lemma \ref{lem:dan1}.
It follows, by the PVA axioms, that $\{v_\lambda\mc L(z)\}^{\mc W}=0$,
and therefore
$$
\frac{\partial v}{\partial t_j}
=
\big\{\tint h_j,v\}^{\mc W}
=
\frac{p_1}{j}
\Res_z\tr
\big\{
\mc L^{\frac{j}{p_1}}(z)_\lambda v\}^{\mc W}|_{\lambda=0}
=0
\,.
$$
\end{proof}
\begin{lemma}\label{lem:dan2}
We have
$\id_{V_{+,\geq2}}T(\partial)\id_{V_{-,\geq2}}
\in
\mc V(\mf g^\prime_{\leq\frac12})[\partial]\otimes\Hom(V_{-,\geq2},V_{+,\geq2})$.
\end{lemma}
\begin{proof}
Note that $V_{+,\geq2},FV\subset V_{+,1}^\prime$ and $V_{-,\geq2},F^TV\subset V_{-,1}^\prime$.
As a consequence, all the operators
$$
\id_{V_{+,\geq2}}\rho A(\partial)\id_{V_{-,\geq2}}
\,\,,\,\,\,\,
\id_{V_{+,\geq2}}\rho A(\partial)\id_{F^TV}
\,\,,\,\,\,\,
\id_{FV}\rho A(\partial)\id_{V_{-,\geq2}}
\,\,,\,\,\,\,
\id_{FV}\rho A(\partial)\id_{F^TV}
\,,
$$
have coefficients of the entries in $\mc V(\mf g^\prime_{\leq\frac12})$.
The claim follows since, by the definition \eqref{eq:quasidet} of quasideterminant,
and the definition \eqref{eq:EndT} of $T(\partial)$,
$$
\id_{V_{+,\geq2}}T(\partial)\id_{V_{-,\geq2}}
\!=
\id_{V_{+,\geq2}}\rho A(\partial)\id_{V_{-,\geq2}}
-
\id_{V_{+,\geq2}}\rho A(\partial)\id_{F^TV}
\big(
\id_{FV}\rho A(\partial)\id_{F^TV}
\big)^{-1}
\id_{FV}\rho A(\partial)\id_{V_{-,\geq2}}
\,.
$$
\end{proof}
\begin{proposition}\label{prop:dan2}
We have
$\id_{V_{+,\geq2}}T^{(k)}(\partial)\id_{V_{-,\geq2}}
\in
\mc V(\mf g^\prime_{\leq\frac12})[\partial]\otimes\Hom(V_{-,\geq2},V_{+,\geq2})$,
for every $k=1,\dots,s$.
\end{proposition}
\begin{proof}
We prove the proposition by downward induction on $k$.
For $k=s$ the claim holds by Lemma \ref{lem:dan2}.
For $k=2,\dots,k-1$ we have, by \eqref{eq:TWk} and \eqref{eq:EFk},
\begin{align*}
& \id_{V_{+,\geq2}} T^{(k)}(\partial) \id_{V_{-,\geq2}}
=
\big(
\id_{V_{+,\geq2}}
-
\id_{V_{+,k}}
\big(
T^{(k+1)}(\partial)
W^{(k+1)}(\partial)^{-1}
\big)_+
\big)
T^{(k+1)}(\partial) \\
& \qquad\times
\big(
\id_{V_{-,\geq2}}
-
\big(
W^{(k+1)}(\partial)^{-1}
T^{(k+1)}(\partial)
\big)_+
\id_{V_{-,k}}
\big)
\,.
\end{align*}
Since, obviously,
$W^{(k+1)}(\partial)^{-1}=\id_{V_{-,\geq 2}}W^{(k+1)}(\partial)^{-1}\id_{V_{+,\geq 2}}$,
and $\id_{V_{\pm,k}}=\id_{V_{\pm,k}}\id_{V_{\pm,\geq2}}$,
we can use the inductive assumption
to conclude that RHS has coefficients in $\mc V(\mf g^\prime_{\leq\frac12})$,
as claimed.
Finally, the claim for $k=1$ holds since, 
by Proposition \ref{thm:recursion3}(a), we have
$\id_{V_{+,\geq2}}T^{(1)}(\partial)\id_{V_{-,\geq2}}=\id_{V_{+,\geq2}}T^{(2)}(\partial)\id_{V_{-,\geq2}}$.
\end{proof}
\begin{corollary}\label{prop:dW4dtn}
The coefficients of the entries of the operator $\id_{V_{+,\geq2}}W(\partial)\id_{V_{-,\geq2}}$
do not evolve w.r.t. the time evolution given by the Hamiltonian flows \eqref{eq:hameq}:
$$
\frac{\partial}{\partial t_j}
\id_{V_{+,\geq2}}W(z)\id_{V_{-,\geq2}}
=0
\,\,\text{ for all }\,\, j\in\mb Z_{\geq1}
\,.
$$
\end{corollary}
\begin{proof}
By Corollary \ref{cor:W} and equation \eqref{eq:TWk}, 
$W(\partial)=(-\partial)^{\ul p}+T^{(1)}(\partial)$.
Hence, the claim is an immediate consequence of 
Propositions \ref{prop:dan1} and \ref{prop:dan2}.
\end{proof}
\begin{remark}\label{rem:dan}
Note that, under the identification \eqref{eq:identif-pm},
the submatrix $W_{\bm2\bm2}(\partial)$ of $W(\partial)$
defined in \eqref{eq:blockW} coincides with $\id_{V_{+,\geq2}}W(\partial)\id_{V_{-,\geq2}}$.
Hence, Corollary \ref{prop:dW4dtn}
can be restated by saying that the coefficients of the entries of $W_{\bm2\bm2}(\partial)$
do not evolve.
\end{remark}

\section{
Lax equations vs Hamiltonian equations associated to the $\mc W$-algebra $\mc W(\mf{gl}_N,\ul p)$
}
\label{sec:9}

The present section is devoted to the proof of the following second main result
of the paper.
\begin{theorem}\label{prop:sylvain}
Consider the classical $\mc W$-algebra $\mc W(\mf{gl}_N,\ul p)$
associated to the partition $\ul p$ of $N$.
Let $W(\partial)\in\Mat_{r\times r}\mc W(\mf{gl}_N,\ul p)[\partial]$
be the matrix differential operator \eqref{eq:matrW} encoding all the $\mc W$-algebra generators,
which we write in block form as in \eqref{eq:blockW}.
Let $\mc L(\partial)\in\Mat_{r_1\times r_1}\mc W(\mf{gl}_N,\ul p)((\partial^{-1}))$
be the Lax operator defined in \eqref{eq:WL}.
Then the Hamiltonian evolution equations \eqref{eq:hameq} on the $\mc W$-algebra
are equivalent to 
the Lax equations \eqref{eq:527} for the operator $\mc L(\partial)$
together with the condition that the generators in the submatrix $W_{\bm2\bm2}(\partial)$ do not evolve:
\begin{equation}\label{eq:evol2}
\frac{\partial}{\partial t_j} W_{\bm2\bm2}(\partial)=0
\,\,\text{ for all }\,\,
j\in\mb Z_{\ge 1}
\,.
\end{equation}
In fact, we can write explicitly the evolution of all other generators as follows  ($j\in\mathbb{Z}_{\ge 1}$):
\begin{equation}\label{eq:evol1}
\begin{split}
& \frac{\partial}{\partial t_j} W_{\bm1\bm2}(\partial)
=
R_{\bm1\bm2}^{(j)}(\partial)
\,\,,\,\,\,\,
\frac{\partial}{\partial t_j} W_{\bm2\bm1}(\partial)
=
-R_{\bm2\bm1}^{(j)}(\partial)
\,, \\
& \frac{\partial}{\partial t_j} W_{\bm1\bm1}(\partial)
=
\big[(\mc L(\partial)^{\frac{j}{p_1}})_+,W_{\bm1\bm1}(\partial)\big]
+
Q_{\bm1\bm2}^{(j)}(\partial)W_{\bm2\bm1}(\partial)
-
W_{\bm1\bm2}(\partial)Q_{\bm2\bm1}^{(j)}(\partial)
\,,
\end{split}
\end{equation}

where the matrix differential operators
$R_{\bm1\bm2}^{(j)}(\partial)$,
$R_{\bm2\bm1}^{(j)}(\partial)$,
$Q_{\bm1\bm2}^{(j)}(\partial)$,
$Q_{\bm2\bm1}^{(j)}(\partial)$,
are uniquely determined by
\begin{equation}\label{eq:QR1}
\begin{split}
& (\mc L(\partial)^{\frac{j}{p_1}})_+
W_{\bm1\bm2}(\partial)
=
Q_{\bm1\bm2}^{(j)}(\partial)
(-(-\partial)^{\ul q}+W_{\bm2\bm2}(\partial))
+
R_{\bm1\bm2}^{(j)}(\partial)
\,, \\
& W_{\bm2\bm1}(\partial)
(\mc L(\partial)^{\frac{j}{p_1}})_+
=
(-(-\partial)^{\ul q}+W_{\bm2\bm2}(\partial))
Q_{\bm2\bm1}^{(j)}(\partial)
+
R_{\bm2\bm1}^{(j)}(\partial)
\,,
\end{split}
\end{equation}
and the conditions that the matrix entries 
of $R_{\bm1\bm2}^{(j)}(\partial)=\big(R_{ab}^{(j)}(\partial)\big)_{1\leq a\leq r_1<b\leq r}$
and of $R_{\bm2\bm1}^{(j)}(\partial)=\big(R_{ab}^{(j)}(\partial)\big)_{1\leq b\leq r_1<a\leq r}$
have the following bounds on their differential orders:
\begin{equation}\label{eq:QR2}
\ord\big(R_{ab}^{(j)}(\partial)\big)\leq\min\{p_a,p_b\}-1
\,.
\end{equation}
\end{theorem}

\subsection{A preliminary result on pseudodifferential operators}\label{sec:sylv1a}

\begin{lemma}\label{lem:sylv1.1}
Let $\mc W$ be a differential algebra with no zero divisors,
and assume that its subalgebra of constants coincides with the base field $\mb C$.
Let $\mc V$ be a subspace of $\mc W$ such that $\mc V \cap \partial \mc W=0$.
Consider the following vector subspaces of $\mc W((\partial^{-1}))$
$$
V_1
=
\Span\big\{
a\partial^{-1}b
\,\big|\,
a,b\in\mc W\big\}
\,\,,\,\,\,\,
V_2
=\Span\big\{
a\partial^{-1} b\partial^{-1} c
\,\big|\,
a, c\in\mc W, \, \, b \in \mc V
\big\}
\,.
$$
\begin{enumerate}[(a)]
\item
We have a vector space isomorphism $\mc W\otimes\mc W\stackrel{\sim}{\longrightarrow} V_1$, given by
$a\otimes b\mapsto a\partial^{-1}b$.
\item
We have a vector space isomorphism 
$\mc W\otimes \mc V \otimes\mc W \stackrel{\sim}{\longrightarrow} V_2$, 
given by $a\otimes  b\otimes c\mapsto a\partial^{-1}b\partial^{-1} c$. 
\item
$V_1 \cap V_2 =0$.
\end{enumerate}
\end{lemma}
\begin{proof}
Claim (a) is the same as \cite[Lem.4.4]{Car17}.
The proof of \cite[Lem.4.8]{Car17} implies (b) and (c),
which are stronger versions of that lemma.
Note also that claim (a) is an alternative version of Lemma \ref{lem:fund3} of the present paper.
\end{proof}

\subsection{Notation for differential order and polynomial degree}\label{sec:sylv1b}

We introduce some notation that we shall use throughout the remainder of Section \ref{sec:9}.
Consider a matrix differential operator $A(\partial)\in\Mat\mc W(\mf{gl}_N,\ul p)[\partial]$,
which we can expand as 
$A(\partial)=\sum_{i\in\mb Z_{\geq0}}A_i\partial^i$,
with $A_i\in\Mat\mc W(\mf{gl}_N,\ul p)$.
We say that $\ord A(\partial)=n$ if $A_n\neq0$ and $A_i=0$ for all $i>n$;
we also denote
\begin{equation}\label{eq:ord}
\ord_i A(\partial)=A_i
\,\,,\,\,\,\,
i\in\mb Z_{\geq0}
\,.
\end{equation}
Next, recall, by Theorem \ref{thm:structure-W},
that $\mc W(\mf{gl}_N,\ul p)$ is an algebra of differential polynomials,
and let $\{w_\alpha\}_{\alpha\in I}$ be a set of differential generators
(with $\#(I)=\dim(\mf g^f)$).
Denote by $w_\alpha^{(n)}=\partial^n w_\alpha$, for all $\alpha\in I$ and $n\in\mb Z_{\geq0}$, and let $\deg(w_\alpha^{(n)})=1$,		which we call the polynomial degree on the $\mc W$-algebra.		
We can expand each coefficient $A_i$ 
in homogeneous components with respect to the polynomial degree:
$A_{i}=\sum_{j\in\mb Z_{\geq0}}A_{i}^j$,
where $A_{i}^j$ is a matrix whose entries are homogeneous polynomials of degree $j$.
Then, we denote
\begin{equation}\label{eq:deg}
\deg^j A(\partial)=\sum_{i=0}^nA_i^j\partial^i
\,\,,\,\,\,\,
j\in\mb Z_{\geq0}
\,,
\end{equation}
the homogeneous component of $A(\partial)$ of degree $j$
w.r.t. the polynomial degree of $\mc W(\mf{gl}_N,\ul p)$.
For example, $\deg^0 A(\partial)\in\mb C[\partial]$ is the constant term of $A(\partial)$,
while $\deg^1(A(\partial))$ has the form
$$
\deg^1 A(\partial)
=
\sum_{i\in\mb Z_{\geq0}}\sum_{\alpha\in I}\sum_{n\in\mb Z_{\geq0}} \gamma_{i,\alpha,n} w_\alpha^{(n)} \partial^i
\,\,,\,\,\,\,
\gamma_{i,\alpha,n}\in\Mat\mb C
\,.
$$ 
Finally, using the above notation, we set
$$
\overline{\deg}^1 A(\partial)
:=
\sum_{\alpha\in I}\sum_{i\in\mb Z_{\geq0}} \gamma_{i,\alpha,0} w_\alpha \partial^i
\,\in\,
\Mat \big(\bigoplus_{\alpha\in I}\mb Cw_\alpha\big)[\partial]
\,.
$$
In other words, 
$\overline{\deg}^1(A(\partial))$ is 
the projection of $A(\partial)$ on the vector space spanned by the generators $\{w_\alpha\}_{\alpha\in I}$
of $\mc W(\mf{gl}_N,\ul p)$.
As an example, by the definition \eqref{eq:matrW} of the matrix 
$W(\partial)\in\Mat_{r\times r}\mc W(\mf{gl}_N,\ul p)[\partial]$,
we have that
\begin{equation}\label{eq:degW}
\deg^j W(\partial)=0\,\text{ for }\, j\neq1
\,,\,\text{ and }\,
\deg^1 W(\partial)
=
\overline{\deg}^1 W(\partial)
=
W(\partial)
\,.
\end{equation}

\subsection{Existence and uniqueness of the Euclidean division \eqref{eq:QR1}-\eqref{eq:QR2}}\label{sec:sylv2}

\begin{lemma}\label{lem:sylv2}
\begin{enumerate}[(a)]
\item
For every matrix differential operator 
$B_{\bm1\bm2}(\partial)\in\Mat_{r_1\times(r-r_1)}\mc W(\mf{gl}_N,\ul p)[\partial]$
there exist unique
$Q_{\bm1\bm2}(\partial),\,R_{\bm1\bm2}(\partial)\,\in\,\Mat_{r_1\times(r-r_1)}\mc W(\mf{gl}_N,\ul p)[\partial]$
such that
\begin{equation}\label{eq:QR1a}
B_{\bm1\bm2}(\partial)
=
Q_{\bm1\bm2}(\partial)
(-(-\partial)^{\ul q}+W_{\bm2\bm2}(\partial))
+
R_{\bm1\bm2}(\partial)
\,,
\end{equation}
and the matrix entries 
of $R_{\bm1\bm2}(\partial)=\big(R_{ab}(\partial)\big)_{1\leq a\leq r_1<b\leq r}$
are such that
$\ord\big(R_{ab}(\partial)\big)\leq p_b-1$.
\item
For every matrix differential operator 
$B_{\bm2\bm1}(\partial)\in\Mat_{(r-r_1)\times r_1}\mc W(\mf{gl}_N,\ul p)[\partial]$
there exist unique
$Q_{\bm2\bm1}(\partial),\,R_{\bm2\bm1}(\partial)\,\in\,\Mat_{(r-r_1)\times r_1}\mc W(\mf{gl}_N,\ul p)[\partial]$
such that
\begin{equation}\label{eq:QR1b}
B_{\bm2\bm1}(\partial)
=
(-(-\partial)^{\ul q}+W_{\bm2\bm2}(\partial))
Q_{\bm2\bm1}(\partial)
+
R_{\bm2\bm1}(\partial)
\,,
\end{equation}
and the matrix entries 
of $R_{\bm2\bm1}(\partial)=\big(R_{ab}(\partial)\big)_{1\leq b\leq r_1<a\leq r}$
are such that
$\ord\big(R_{ab}(\partial)\big)\leq p_a-1$.
\end{enumerate}
\end{lemma}
\begin{proof}
We prove claim (a); the proof of (b) is similar.
First, we prove uniqueness.
Suppose that
\begin{equation}\label{eq:sylv1}
Q_{\bm1\bm2}(\partial)
(-(-\partial)^{\ul q}+W_{\bm2\bm2}(\partial))
+
R_{\bm1\bm2}(\partial)
=
\widetilde{Q}_{\bm1\bm2}(\partial)
(-(-\partial)^{\ul q}+W_{\bm2\bm2}(\partial))
+
\widetilde{R}_{\bm1\bm2}(\partial)
\,,
\end{equation}
with both $R_{\bm1\bm2}(\partial)$ and $\widetilde{R}_{\bm1\bm2}(\partial)$
satisfying the stated bounds on the orders of their matrix entries:
\begin{equation}\label{eq:sylv2}
\ord(R_{ab}(\partial))
,\,
\ord(\widetilde{R}_{ab}(\partial))
\,\leq\,
p_b-1
\,.
\end{equation}
Suppose, by contradiction, that 
$(Q_{\bm1\bm2}(\partial),R_{\bm1\bm2}(\partial))
\neq
(\widetilde{Q}_{\bm1\bm2}(\partial),\widetilde{R}_{\bm1\bm2}(\partial))$,
and let $n$ be the smallest degree
at which they do not match:
\begin{equation}\label{eq:sylv3}
(\deg^n Q_{\bm1\bm2}(\partial),\deg^n R_{\bm1\bm2}(\partial))
\neq
(\deg^n \widetilde{Q}_{\bm1\bm2}(\partial),\deg^n \widetilde{R}_{\bm1\bm2}(\partial))
\,,
\end{equation}
and
\begin{equation}\label{eq:sylv4}
(\deg^j Q_{\bm1\bm2}(\partial),\deg^j R_{\bm1\bm2}(\partial))
=
(\deg^j \widetilde{Q}_{\bm1\bm2}(\partial),\deg^j \widetilde{R}_{\bm1\bm2}(\partial))
\,\text{ if }\,
j<n
\,.
\end{equation}
Taking the $n$-degree components of both sides of equation \eqref{eq:sylv1}
we get, recalling \eqref{eq:degW},
\begin{align*}
& \deg^n(Q_{\bm1\bm2}(\partial)) (-(-\partial)^{\ul q})
+
\deg^{n-1}(Q_{\bm1\bm2}(\partial))
W_{\bm2\bm2}(\partial)
+
\deg^n R_{\bm1\bm2}(\partial) \\
& =
\deg^n(\widetilde{Q}_{\bm1\bm2}(\partial))
(-(-\partial)^{\ul q})
+
\deg^{n-1}(\widetilde{Q}_{\bm1\bm2}(\partial))
W_{\bm2\bm2}(\partial)
+
\deg^n \widetilde{R}_{\bm1\bm2}(\partial)
\,.
\end{align*}
Hence, using \eqref{eq:sylv4}, we get
$$
\deg^n(Q_{\bm1\bm2}(\partial)) (-(-\partial)^{\ul q})
+
\deg^n R_{\bm1\bm2}(\partial) 
=
\deg^n(\widetilde{Q}_{\bm1\bm2}(\partial))
(-(-\partial)^{\ul q})
+
\deg^n \widetilde{R}_{\bm1\bm2}(\partial)
\,.
$$
Taking the $(a,b)$-entry of both sides of the above equation,
we get
$$
(-1)^{p_b+1}
\deg^n(Q_{ab}(\partial)) \partial^{p_b}
+
\deg^n R_{ab}(\partial) 
=
(-1)^{p_b+1}
\deg^n(\widetilde{Q}_{ab}(\partial)) \partial^{p_b}
+
\deg^n(\widetilde{R}_{ab}(\partial))
\,,
$$
which clearly implies 
$$
\deg^n Q_{ab}(\partial)
=
\deg^n \widetilde{Q}_{ab}(\partial)
\,\,\text{ and }\,\,
\deg^n R_{ab}(\partial)
=
\deg^n \widetilde{R}_{ab}(\partial)
\,,
$$
by the assumption \eqref{eq:sylv2}.
This contradicts \eqref{eq:sylv3}.

Next, we prove the existence of $Q_{\bm1\bm2}(\partial)$ and $R_{\bm1\bm2}(\partial)$
by induction on $m=\ord B_{\bm1\bm2}(\partial)$.
First note that, if $\ord B_{ab}(\partial)\leq p_b-1$ for all $1\leq a\leq r_1<b\leq r$, 
we can set $Q_{\bm1\bm2}(\partial)=0$ and $R_{\bm1\bm2}(\partial)=B_{\bm1\bm2}(\partial)$.
Otherwise, 
for each $a,b$, we can uniquely decompose
\begin{equation}\label{eq:sylv5}
B_{ab}(\partial)
=
(-1)^{p_b+1}Q_{ab}^0(\partial)\partial^{p_b}+R_{ab}^0(\partial)
\,,
\end{equation}
where
\begin{equation}\label{eq:sylv6}
\ord R_{ab}^0(\partial)\leq p_b-1
\,\,\text{ and }\,\,
\ord Q_{ab}^0(\partial)
\leq
\ord B_{ab}(\partial)-p_b
\leq m-p_b
\,.
\end{equation}
Let 
$Q_{\bm1\bm2}^0(\partial)$ and $R_{\bm1\bm2}^0(\partial)$
be the $r_1\times(r-r_1)$-matrices with entries 
$Q_{ab}^0(\partial)$ and $R_{ab}^0(\partial)$ respectively,
so that equation \eqref{eq:sylv5} can be written in matrix form as
\begin{equation}\label{eq:sylv7}
B_{\bm1\bm2}(\partial)
=
Q_{\bm1\bm2}^0(\partial)(-(-\partial)^{\ul q})+R_{\bm1\bm2}^0(\partial)
\,,
\end{equation}
Next, consider the matrix differential operator
\begin{equation}\label{eq:sylv8}
C_{\bm1\bm2}(\partial)
:=
Q_{\bm1\bm2}^0(\partial)W_{\bm2\bm2}(\partial)
\,.
\end{equation}
By \eqref{eq:matrW} and the second inequality in \eqref{eq:sylv6},
its $(a,b)$-entry has differential order 
\begin{align*}
& \ord C_{ab}(\partial)
\leq
\max\big\{
\ord Q_{ac}^0(\partial)+\ord W_{cb}(\partial)
\big\}_{c=r_1+1}^r \\
& \leq
\max\big\{
m-p_c+\min\{p_c,p_b\}-1
\big\}_{c=r_1+1}^r
\leq
m-1
\,.
\end{align*}
Hence, $\ord C_{\bm1\bm2}(\partial)\leq m-1$, and we can apply the inductive assumption
to get matrices 
$Q_{\bm1\bm2}^1(\partial)$ and $R_{\bm1\bm2}^1(\partial)$
such that
\begin{equation}\label{eq:sylv9}
C_{\bm1\bm2}(\partial)
=
Q_{\bm1\bm2}^1(\partial)(-(-\partial)^{\ul q}+W_{\bm2\bm2}(\partial))+R_{\bm1\bm2}^1(\partial)
\,,
\end{equation}
with $\ord R_{ab}^1(\partial)\leq p_b-1$.
Combining equations \eqref{eq:sylv7}, \eqref{eq:sylv8} and \eqref{eq:sylv9},
we get that equation \eqref{eq:QR1a} holds with
$$
Q_{\bm1\bm2}(\partial)
=
Q_{\bm1\bm2}^0(\partial)
-
Q_{\bm1\bm2}^1(\partial)
\,\,\text{ and }\,\,
R_{\bm1\bm2}(\partial)
=
R_{\bm1\bm2}^0(\partial)
-
R_{\bm1\bm2}^1(\partial)
\,.
$$
\end{proof}
\begin{remark}
As a special case of Lemma \ref{lem:sylv2}, 
applied to 
$B_{\bm1\bm2}(\partial)=(\mc L(\partial)^{\frac{j}{p_1}})_+W_{\bm1\bm2}(\partial)$
and
$B_{\bm2\bm1}(\partial)=W_{\bm2\bm1}(\partial)(\mc L(\partial)^{\frac{j}{p_1}})_+$,
we get that the matrix differential operators 
$R_{\bm1\bm2}^{(j)}(\partial)$,
$R_{\bm2\bm1}^{(j)}(\partial)$,
$Q_{\bm1\bm2}^{(j)}(\partial)$,
and $Q_{\bm2\bm1}^{(j)}(\partial)$,
in Theorem \ref{prop:sylvain} exist and are unique.
\end{remark}

\subsection{Unique decomposition of certain operators}\label{sec:sylv3}

\begin{lemma}\label{lem:sylvain0}
Let $b \in \mc W(\mf{gl}_N,\ul p)\backslash\big(\mb C\oplus\partial\mc W(\mf{gl}_N,\ul p)\big)$,
and let $a(\partial),\widetilde a(\partial),c(\partial),\widetilde c(\partial)\in\mc W(\mf{gl}_N,\ul p)[\partial]$
be such that
\begin{equation}\label{eq:sylv0.1}
a(\partial)\partial^{-m}b\partial^{-n}c(\partial)
=
\widetilde a(\partial)\partial^{-m}b\partial^{-n}\widetilde c(\partial)
\,\text{ in }\,
\mc W(\mf{gl}_N,\ul p)((\partial^{-1}))
\,,
\end{equation}
for some integers
\begin{equation}\label{eq:sylv0.2}
m>\ord a(\partial),\ord\widetilde a(\partial)
\,\,\text{ and }\,\,
n>\ord c(\partial),\ord\widetilde c(\partial)
\,.
\end{equation}
Then
\begin{equation}\label{eq:sylv0.3}
a(\partial)\otimes c(\partial)=\widetilde a(\partial)\otimes \widetilde c(\partial)
\,\text{ in }\,
\mc W(\mf{gl}_N,\ul p)[\partial]\otimes\mc W(\mf{gl}_N,\ul p)[\partial]
\,.
\end{equation}
\end{lemma}
\begin{proof}
Expand the differential operators $a(\partial),c(\partial),\widetilde a(\partial),\widetilde c(\partial)$ as
$$
a(\partial)=\sum_{i=0}^{m-1}a_i\partial^i
\,,\,\,
\widetilde a(\partial)=\sum_{i=0}^{m-1}\widetilde a_i\partial^i
\,,\,\,
c(\partial)=\sum_{j=0}^{n-1}\partial^jc_j
\,,\,\,
\widetilde c(\partial)=\sum_{j=0}^{n-1}\partial^j\widetilde c_j
\,.
$$
Then, equation \eqref{eq:sylv0.1} reads
\begin{equation}\label{20190822:eq1}
\sum_{i=0}^{m-1}\sum_{j=0}^{n-1}
\big(
a_i\partial^{-(m-i)}b\partial^{-(n-j)}c_j
-
\widetilde a_i\partial^{-(m-i)}b\partial^{-(n-j)}\widetilde c_j
\big)
=0
\,\text{ in }\,
\mc W(\mf{gl}_N,\ul p)((\partial^{-1}))
\,.
\end{equation}
For any integer $n \geq 1$ we have the following identity
of pseudodifferential
operators, which can be easily proved by induction on $n$:
\begin{equation} \label{partial-n}
\partial^{-n}
=
\sum_{k=0}^{n-1}
\frac{(-1)^k}{k!(n-1-k)!}
x^{n-1-k}\partial^{-1} \circ x^k
\,,\, \text{where}\, \,\partial=\frac{\partial}{\partial x}.
\end{equation}

Using \eqref{partial-n}, equation \eqref{20190822:eq1} becomes
\begin{align*}
& \sum_{\substack{ i,h\in\mb Z_{\geq0} \\ (i+h\leq m-1) }}
\sum_{\substack{ j,k\in\mb Z_{\geq0} \\ (j+k\leq n-1) }}
\frac{(-1)^{n+h+k+j+1}}{h!k!(m-h-i-1)!(n-k-j-1)!} \\
& \times
\big(
a_i x^{m-h-i-1}
\partial^{-1} 
b x^{h+k}
\partial^{-1} 
c_j x^{n-k-j-1}
-
\widetilde a_i x^{m-h-i-1}
\partial^{-1} 
b x^{h+k}
\partial^{-1} 
\widetilde c_j x^{n-k-j-1}
\big)
=0
\,.
\end{align*}
We then apply Lemma \ref{lem:sylv1.1} for the differential algebra $\mc W=\mc W(\mf{gl}_N,\ul p)[x]$
to deduce that
\begin{align*}
& \sum_{\substack{ i,h\in\mb Z_{\geq0} \\ (i+h\leq m-1) }}
\sum_{\substack{ j,k\in\mb Z_{\geq0} \\ (j+k\leq n-1) }}
\frac{(-1)^{n+h+k+j+1}}{h!k!(m-h-i-1)!(n-k-j-1)!} \\
& \times
\big(
a_i x^{m-h-i-1}
\otimes
\tint b x^{h+k}
\otimes
c_j x^{n-k-j-1}
-
\widetilde a_i x^{m-h-i-1}
\otimes
\tint b x^{h+k}
\otimes
\widetilde c_j x^{n-k-j-1}
\big)
=0
\,,
\end{align*}
in the space $\mc W\otimes(\mc W/\partial\mc W)\otimes\mc W$.
Next, we observe that,
under the assumption that $b\not\in\mb C\oplus\partial\mc W(\mf{gl}_N,\ul p)$,
the elements 
$\{\tint bx^\ell\}_{\ell\in\mb Z_{\geq0}}\subset\mc W(\mf{gl}_N,\ul p)[x]/\partial(\mc W(\mf{gl}_N,\ul p)[x])$
are linearly independent over $\mb C$.
Indeed, it is not hard to check that a relation of linear dependence
$\alpha_0\tint b+\alpha_1\tint bx+\dots+\alpha_n\tint bx^n=0$,
with $\alpha_0,\dots,\alpha_n\in\mb C$ and $\alpha_n\neq0$,
is possible only if $b\in\mb C\oplus\partial^{n+1}(\mc W(\mf{gl}_N,\ul p))$.
Hence, the term with $h=k=0$ in the above equation must vanish:
$$
\Big(
\sum_{i=0}^{m-1}
a_i 
\frac{x^{m-i-1}}{(m-i-1)!}
\Big)
\otimes
\Big(
\sum_{j=0}^{n-1}
c_j 
\frac{(-x)^{n-j-1}}{(n-j-1)!}
\Big)
=
\Big(
\sum_{i=0}^{m-1}
\widetilde a_i 
\frac{x^{m-i-1}}{(m-i-1)!}
\Big)
\otimes
\Big(
\sum_{j=0}^{n-1}
\widetilde c_j 
\frac{(-x)^{n-j-1}}{(n-j-1)!}
\Big)
\,,
$$
in the space $\mc W(\mf{gl}_N,\ul p)[x]\otimes\mc W(\mf{gl}_N,\ul p)[x]$.
This is of course equivalent to saying that
$a(\partial)\otimes c(\partial)=\widetilde a(\partial)\otimes \widetilde c(\partial)$
in the space
$\mc W(\mf{gl}_N,\ul p)[\partial]\otimes\mc W(\mf{gl}_N,\ul p)[\partial]$.
\end{proof}

Given two (matrix) pseudodifferential operators $A(\partial)$ and $B(\partial)$, we shall write $A(\partial) \equiv B(\partial)$ if they differ by a (matrix) differential operator.
\begin{lemma}\label{prop:sylvain0}
Let 
$\widetilde W_{\bm1\bm2}(\partial)
=\big(\widetilde W_{ab}(\partial)\big)_{1\leq a\leq r_1<b\leq r}
\in\Mat_{r_1\times(r-r_1)}\mc W(\mf{gl}_N,\ul p)[\partial]$
and 
$\widetilde W_{\bm2\bm1}(\partial)
=\big(\widetilde W_{ab}(\partial)\big)_{1\leq b\leq r_1<a\leq r}
\in\Mat_{(r-r_1)\times r_1}\mc W(\mf{gl}_N,\ul p)[\partial]$
be such that
\begin{equation}\label{20190822:pm1}
\ord\widetilde W_{ab}(\partial)\leq\min\{p_a,p_b\}-1
\,\,\text{ for all }\,\, a,b
\,.
\end{equation}
Then 
\begin{equation}\label{eq3}
W_{\bm1\bm2}(\partial)
\big(-(-\partial)^{\ul q}+ W_{\bm2\bm2}(\partial)\big)^{-1}
\widetilde W_{\bm2\bm1}(\partial)
\equiv
\widetilde W_{\bm1\bm2}(\partial)
\big(-(-\partial)^{\ul q}+ W_{\bm2\bm2}(\partial)\big)^{-1}
W_{\bm2\bm1}(\partial)
\,
\end{equation}
if and only if there exists $\alpha\in\mb C$ such that
\begin{equation}\label{eq3b}
\widetilde W_{\bm1\bm2}(\partial)=
\alpha
W_{\bm1\bm2}(\partial)
\,\,,\,\,\,\,
\widetilde W_{\bm2\bm1}(\partial)
=
\alpha
W_{\bm2\bm1}(\partial)
\,.
\end{equation}
\end{lemma}

\begin{proof}
Clearly, \eqref{eq3b} implies \eqref{eq3}, so we only have to prove the ``only if'' part.
We then fix $a,b\in\{1,\dots,r_1\}$
and we equate the $(a,b)$-entry of both sides of \eqref{eq3}.
As a result, we get
\begin{equation}\label{eq3c}
W_{a\bm2}(\partial)
\big(-(-\partial)^{\ul q}+ W_{\bm2\bm2}(\partial)\big)^{-1}
\widetilde W_{\bm2b}(\partial)
\equiv
\widetilde W_{a\bm2}(\partial)
\big(-(-\partial)^{\ul q}+ W_{\bm2\bm2}(\partial)\big)^{-1}
W_{\bm2b}(\partial)
\,.
\end{equation}
Next, we take the homogeneous component of degree 1 
(w.r.t. the polynomial degree of $\mc W(\mf{gl}_N,\ul p)$)
in both sides of \eqref{eq3c}.
Recalling \eqref{eq:degW}, we get
$$
W_{a\bm2}(\partial)
(-\partial)^{-\ul q}
\deg^0(\widetilde W_{\bm2b}(\partial))
\equiv
\deg^0(\widetilde W_{a\bm2}(\partial))
(-\partial)^{-\ul q}
W_{\bm2b}(\partial)
\,,
$$
which can be expanded in terms of matrix coefficients as
$$
\sum_{c=r_1+1}^r
\sum_{\substack{i,j\in\mb Z_{\geq0} \\ i+j\leq p_c-1}}
\big(
w_{ca;i}
(-\partial)^{-p_c+i+j}
\deg^0(\widetilde w_{bc;j})
-
\deg^0(\widetilde w_{ca;i})
(-\partial)^{-p_c+i+j}
w_{bc;j}
\big)
=0
\,.
$$
Using formula \eqref{partial-n},
and applying Lemma \ref{lem:sylv1.1}(a) for the differential algebra $\mc W=\mc W(\mf{gl}_N,\ul p)[x]$,
we get
\begin{align*}
& \sum_{c=r_1+1}^r
\sum_{\substack{i,j,k\in\mb Z_{\geq0} \\ i+j+k\leq p_c-1}}
\frac{(-1)^{p_c-i-j-k}}{k!(p_c-i-j-k-1)!} \\
&\times  \big(
w_{ca;i}
x^{p_c-i-j-k-1}
\otimes
\deg^0(\widetilde w_{bc;j})
x^k
-
\deg^0(\widetilde w_{ca;i})
x^{p_c-i-j-k-1}
\otimes
w_{bc;j}
x^k
\big)
=0
\,,
\end{align*}
in the space $\mc W\otimes\mc W$.
Since, obviously, $w_{ca;i}x^\ell$, for $c=r_1+1,\dots,r$ and $\ell\in\mb Z_{\geq0}$,
and $x^\ell$, for $\ell\in\mb Z_{\geq0}$,
are all linearly independent in $\mc W(\mf{gl}_N,\ul p)[x]$,
the above equation immediately implies
$\deg^0(\widetilde w_{bc;j})=\deg^0(\widetilde w_{ca;i})=0$, for all $c,i,j$.
Hence,
$$
\deg^0(\widetilde W_{a\bm2}(\partial))=0
\,\,,\,\,\,\,
\deg^0(\widetilde W_{\bm2b}(\partial))=0
\,.
$$

Next, we take the homogeneous component of degree 3 
in both sides of \eqref{eq3c}:
\begin{equation} \label{eq3d}
\begin{split}
&W_{a\bm2}(\partial)
(-\partial)^{-\ul q}
\deg^2(\widetilde W_{\bm2b}(\partial))+W_{a\bm2}(\partial)
(-\partial)^{-\ul q} W_{\bm 2 \bm 2}(\partial) (-\partial)^{-\ul q}
\deg^1(\widetilde W_{\bm2b}(\partial)) \\
\equiv &
\deg^2(\widetilde W_{a\bm2}(\partial))
(-\partial)^{-\ul q}
W_{\bm2b}(\partial)+
\deg^1(\widetilde W_{a \bm2}(\partial))
(-\partial)^{-\ul q} W_{\bm 2 \bm 2}(\partial) (-\partial)^{-\ul q}
W_{\bm2b}(\partial)
\,.
\end{split}
\end{equation}
Consider as above the differential domain $\mc W=\mc W(\mf{gl}_N,\ul p)[x]$. Let $\mc V$ be the subspace of $\mc W(\mf{gl}_N,\ul p)$ spanned by the coefficients of the entries of $W_{\bm 2 \bm 2}(\partial)$. It is clear that $\mc V[x] \cap \partial \mc W=0$ since the elements that span $\mc V$ are some of the generators of the differential algebra $\mc W(\mf{gl}_N,\ul p)$. Recall the spaces $V_1$ and $V_2$ defined in Lemma \ref{lem:sylv1.1}. Using formula \eqref{partial-n}, we see that the first terms in both the LHS and the RHS of \eqref{eq3d} are in $V_1 \oplus \mc W(\mf{gl}_N,\ul p)[\partial]$. Therefore
\begin{equation} \label{eq3e}
\begin{split}
&W_{a\bm2}(\partial)
(-\partial)^{-\ul q} W_{\bm 2 \bm 2}(\partial) (-\partial)^{-\ul q}
\deg^1(\widetilde W_{\bm2b}(\partial)) \\
\equiv
&\deg^1(\widetilde W_{a \bm2}(\partial))
(-\partial)^{-\ul q} W_{\bm 2 \bm 2}(\partial) (-\partial)^{-\ul q}
W_{\bm2b}(\partial)
\, \,
\mod
V_1 \oplus \mc W(\mf{gl}_N,\ul p)[\partial].
\,
\end{split}
\end{equation}
Let us pick $r_1+1 \leq c,d \leq r$. Let $\widetilde{\mc V}$ be the subspace of $\mc V$ spanned by the elements $w_{ef;k}$ for $(e,f,k) \neq (d,c,0)$. By construction $\mc V=\widetilde{\mc V} \oplus \mb{C} w_{dc;0}$. Let $V_3$ and $V_4$ be the following subspaces of $V_2$:
\begin{align*}
 V_3&=\{ u \partial^{-1} v \partial^{-1} w \, | \, u,  w \in \mc W(\mf{gl}_N, \ul p)[x], \, v \in \widetilde{ \mc V}[x] \} ,\\
V_4&=\{ u \partial^{-1} v \partial^{-1} w \, | \, u,  w \in \mc W(\mf{gl}_N, \ul p)[x], \, v \in w_{dc;0} \mb{C} [x] \}.
\end{align*} 
It follows from part (b) of Lemma \ref{lem:sylv1.1} that $V_2=V_3 \oplus V_4$. The equation \eqref{eq3e} is an equation in $V_2 \oplus V_1 \oplus \mc W(\mf{gl}_N, \ul p)[\partial]$ modulo $V_1 \oplus \mc W(\mf{gl}_N, \ul p)[\partial]$. We can project it on $V_4 \oplus V_1 \oplus  \mc W(\mf{gl}_N, \ul p)[\partial]$ modulo $V_1 \oplus \mc W(\mf{gl}_N, \ul p)[\partial]$ using the decomposition $V_2=V_3 \oplus V_4$. It is clear by definition of $V_3$ and $V_4$ and formula \eqref{partial-n} that we thus obtain
\begin{equation} \label{eq3f}
\begin{split}
&W_{ac}(\partial)
(-\partial)^{-p_c} w_{dc;0} (-\partial)^{-p_d}
\deg^1(\widetilde W_{db}(\partial)) \\
\equiv
&\deg^1(\widetilde W_{ac}(\partial))
(-\partial)^{-p_c} w_{dc;0} (-\partial)^{-p_d}
W_{db}(\partial)
\, \,
\mod
V_1 \oplus \mc W(\mf{gl}_N,\ul p)[\partial]
\,.
\end{split}
\end{equation}
The differential order of $W_{ac}(\partial)$ and $ \deg^1(\widetilde W_{ac})(\partial)$ (resp. $W_{db}(\partial)$ and $ \deg^1(\widetilde W_{db})(\partial)$ is strictly less than $p_c$ (resp. $p_d$) hence both sides of \eqref{eq3f} are in $V_4$, which means we can remove $\mod
V_1 \oplus \mc W(\mf{gl}_N,\ul p)[\partial]$:

\begin{equation} \label{eq3g}
\begin{split}
&W_{ac}(\partial)
(-\partial)^{-p_c} w_{dc;0} (-\partial)^{-p_d}
\deg^1(\widetilde W_{db}(\partial)) \\
= 
&\deg^1(\widetilde W_{ac}(\partial))
(-\partial)^{-p_c} w_{dc;0} (-\partial)^{-p_d}
W_{db}(\partial).
\end{split}
\end{equation}
It follows from Lemma \ref{lem:sylvain0} that there exists a constant $\alpha_{abcd} \in \mb{C}$ such that 
\begin{equation*}
\deg^1(\widetilde W_{ac}(\partial))=\alpha_{abcd} W_{ac}(\partial), \, \, \, 
\deg^1(\widetilde W_{db}(\partial))=\alpha_{abcd} W_{db}(\partial).
\end{equation*}
We deduce from these identities, valid for all $1 \leq a,b \leq r_1$ and all $r_1  +1 \leq ,c,d \leq r$, that the constants $\alpha_{abcd}$ are equal to the same constant $\alpha$. Indeed, one can see from the first identity that $\alpha_{abcd}$ does not depend on tha pair $(b,d)$ and from the second one that it does not depend on the pair $(a,c)$. Therefore we have proved that 
\begin{equation*}
\deg^1(\widetilde W_{\bm 1 \bm 2}(\partial))=\alpha W_{\bm 1 \bm 2}(\partial), \, \, \, 
\deg^1(\widetilde W_{\bm 2 \bm 1}(\partial))=\alpha W_{\bm 2 \bm 1}(\partial).
\end{equation*}
Finally, to remove the $deg^1$ above, we let 
\begin{equation*}
\widehat W_{\bm 1 \bm 2}(\partial)=\widetilde W_{\bm 1 \bm 2}(\partial))-\alpha W_{\bm 1 \bm 2}(\partial), \, \, \,
\widehat W_{\bm 2 \bm 1}(\partial)=\widetilde W_{\bm 2 \bm 1}(\partial))-\alpha W_{\bm 2 \bm 1}(\partial).
\end{equation*}
By construction we have $\deg^0(\widehat W_{\bm 1 \bm 2}(\partial))=\deg^1(\widehat W_{\bm 1 \bm 2}(\partial))=\deg^0(\widehat W_{\bm 2 \bm 1}(\partial))=\deg^1(\widehat W_{\bm 2 \bm 1}(\partial))=0$.
Moreover, the pair $(\widehat W_{\bm 1 \bm 2}(\partial),\widehat W_{\bm 2 \bm 1}(\partial))$ also satisfies \eqref{eq3c}. Let $n\geq 2$ be the smallest integer such that the pair $(\deg^n(\widehat W_{\bm 1 \bm 2}(\partial)),\deg^n(\widehat W_{\bm 2 \bm 1}(\partial)))$ is non-zero. Taking the $(n+1)$-th homogeneous component of \eqref{eq3c} with $(\widehat W_{\bm 1 \bm 2}(\partial),\widehat W_{\bm 2 \bm 1}(\partial))$, we obtain a contradiction: $\deg^n(\widehat W_{\bm 1 \bm 2}(\partial))=\deg^n(\widehat W_{\bm 2 \bm 1}(\partial))=0$ by exactly the same argument used above to prove that
$\deg^0(\widetilde W_{\bm 1 \bm 2}(\partial))=\deg^0(\widetilde W_{\bm 2 \bm 1}(\partial))=0$.
\end{proof}

\subsection{Evolution of $W(\partial)$}\label{sec:sylv4}

\begin{proposition}\label{prop:sylvain1}
In the same setting and notation as of Theorem \ref{prop:sylvain},
suppose that we have time evolution in the $\mc W$-algebra,
with time denoted $t_j$,
for which $W_{\bm2\bm2}(\partial)$ does not evolve, i.e. \eqref{eq:evol2} holds.
Then,
$\mc L(\partial)$ evolves according to the Lax equation \eqref{eq:527}
if and only if there exists $\alpha\in\mb C$ such that
\begin{equation}\label{eq:evol1b}
\begin{split}
& \frac{\partial}{\partial t_j} W_{\bm1\bm2}(\partial)
=
R_{\bm1\bm2}^{(j)}(\partial)
+\alpha
W_{\bm1\bm2}(\partial)
\,\,,\,\,\,\,
\frac{\partial}{\partial t_j} W_{\bm2\bm1}(\partial)
=
-R_{\bm2\bm1}^{(j)}(\partial)
-\alpha
W_{\bm2\bm1}(\partial)
\,, \\
& \frac{\partial}{\partial t_j} W_{\bm1\bm1}(\partial)
=
\big[(\mc L(\partial)^{\frac{j}{p_1}})_+,W_{\bm1\bm1}(\partial)\big]
+
Q_{\bm1\bm2}^{(j)}(\partial)W_{\bm2\bm1}(\partial)
-
W_{\bm1\bm2}(\partial)Q_{\bm2\bm1}^{(j)}(\partial)
\,,
\end{split}
\end{equation}
where 
$R_{\bm1\bm2}^{(j)}(\partial)$,
$R_{\bm2\bm1}^{(j)}(\partial)$,
$Q_{\bm1\bm2}^{(j)}(\partial)$,
$Q_{\bm2\bm1}^{(j)}(\partial)$,
are defined by \eqref{eq:QR1}-\eqref{eq:QR2} (cf. Lemma \ref{lem:sylv2}).
\end{proposition}
\begin{proof}
By \eqref{eq:WL} and the assumption \eqref{eq:evol2}, we have
\begin{equation}\label{20190822:eqsera1}
\begin{split}
\frac{\partial}{\partial t_j}
\mc L(\partial)
& =
\frac{\partial W_{\bm1\bm1}}{\partial t_j}(\partial)
-
\frac{\partial W_{\bm1\bm2}}{\partial t_j}(\partial)
\circ
\big(-(-\partial)^{\ul q}+W_{\bm2\bm2}(\partial)\big)^{-1}
\circ
W_{\bm2\bm1}(\partial) \\
& -
W_{\bm1\bm2}(\partial)
\circ
\big(-(-\partial)^{\ul q}+W_{\bm2\bm2}(\partial)\big)^{-1}
\circ
\frac{\partial W_{\bm2\bm1}}{\partial t_j}(\partial)
\,.
\end{split}
\end{equation}
It is immediate to check that equations \eqref{eq:evol1b} and \eqref{20190822:eqsera1}
imply the Lax equation \eqref{eq:527}, proving the ``if'' part.
Conversely, by the Lax equation \eqref{eq:527}, we have
\begin{equation}\label{20190822:eqsera2}
\begin{split}
\frac{\partial}{\partial t_j}
\mc L(\partial)
& 
=
\big[
\big(\mc L(\partial)^{\frac{j}{p_1}}\big)_+,
\frac{\partial W_{\bm1\bm1}}{\partial t_j}(\partial)
]
-
\big(\mc L(\partial)^{\frac{j}{p_1}}\big)_+
W_{\bm1\bm2}(\partial)
\circ
\big(-(-\partial)^{\ul q}+W_{\bm2\bm2}(\partial)\big)^{-1}
\circ
W_{\bm2\bm1}(\partial) \\
& +
W_{\bm1\bm2}(\partial)
\circ
\big(-(-\partial)^{\ul q}+W_{\bm2\bm2}(\partial)\big)^{-1}
\circ
W_{\bm2\bm1}(\partial)
\big(\mc L(\partial)^{\frac{j}{p_1}}\big)_+
\,.
\end{split}
\end{equation}
Combining equations \eqref{20190822:eqsera1} and \eqref{20190822:eqsera2},
and using \eqref{eq:QR1},
we get
\begin{align*}
& W_{\bm1\bm2}(\partial)
\circ
\big(-(-\partial)^{\ul q}+W_{\bm2\bm2}(\partial)\big)^{-1}
\Big(
\frac{\partial W_{\bm2\bm1}}{\partial t_j}(\partial)
+
R^{(j)}_{\bm2\bm1}(\partial)
\Big) \\
& \equiv
-\Big(
\frac{\partial W_{\bm1\bm2}}{\partial t_j}(\partial)
-
R^{(j)}_{\bm1\bm2}(\partial)
\Big)
\big(-(-\partial)^{\ul q}+W_{\bm2\bm2}(\partial)\big)^{-1}
W_{\bm2\bm1}(\partial)
\,,
\end{align*}
modulo $\Mat_{r_1\times r_1}\mc W(\mf{gl}_N,\ul p)[\partial]$.
The claim follows by Lemma \ref{prop:sylvain0}.
\end{proof}

\subsection{Proof that $\alpha=0$ in \eqref{eq:evol1b}}\label{sec:sylv5}

\begin{lemma}\label{lem:last1}
In  any $\mc W$-algebra $\mc W(\mf g,f)$, we have:
\begin{enumerate}[(a)]
\item
if $u,v\in\mc W(\mf g,f)$ are homogeneous of polynomial degree 1,
then $\deg^0\{u_\lambda v\}^{\mc W}\big|_{\lambda=0}=0$;
\item
if $u\in\mc W(\mf g,f)$ is homogeneous of polynomial degree 1 and $v\in\mc W(\mf g,f)$
is homogeneous of polynomial degree 2,
then $\deg^1\{u_\lambda v\}^{\mc W}\big|_{\lambda=0}=0$
and 
$\overline{\deg}^1\{v_\lambda u\}^{\mc W}\big|_{\lambda=0}=0$;
\item
$\overline{\deg}^1\{u_\lambda v\}^{\mc W}\big|_{\lambda=0}
=
\overline{\deg}^1\{\overline{\deg}^1(u)_\lambda\overline{\deg}^1(v)\}^{\mc W}\big|_{\lambda=0}$, 
for every $u,v\in\mc W$;
\item
$\overline{\deg}^1\{w(p)_\lambda w(q)\}^{\mc W}\big|_{\lambda=0}=w([p,q])$, for every $p,q\in\mf g^f$,
where $w:\,\mf g^f\to\mc W(\mf g,f)$ is the isomorphism defined in Theorem \ref{thm:structure-W}.
\end{enumerate}
\end{lemma}
\begin{proof}
Recall the conformal weight $\Delta$ on $\mc W(\mf g,f)$, defined by  $\Delta(a^{(n)})=1-j+n$ if $a \in w({\mf g}_j^f)$ ($j \geq 0$), $\Delta(ab)=\Delta(a) \Delta(b)$; then $\Delta(a_{(k)}b)=\Delta(a)+\Delta(b)-k-1$ \cite{DSKV14}. We have $\Delta(u_{(0)}v)=\Delta(u)+\Delta(v)-1 \geq 1$, proving (a).

For claim (b), let, without loss of generality, 
$v=v_1v_2$, with $v_1,v_2\in\mc W(\mf g,f)$ of degree $1$.
Hence, by the Leibniz rule,
$$
\deg^1
\{u_\lambda v\}^{\mc W}
=
\deg^0(\{u_\lambda v_1\}^{\mc W}) v_2
+
\deg^0(\{u_\lambda v_2\}^{\mc W}) v_1
\,.
$$
Setting $\lambda=0$ the RHS vanishes, by claim (a).
On the other hand, by the right Leibniz rule,
$$
\deg^1
\{v_\lambda u\}^{\mc W}
=
\deg^0(\{{v_1}_{\lambda+\partial} u\}^{\mc W}_\to) v_2
+
\deg^0(\{{v_2}_{\lambda+\partial} u\}^{\mc W}_\to) v_1
\,.
$$
Setting $\lambda=0$ and applying $\overline{\deg}^1$,
which amounts to setting $\partial=0$,
we get 0, again by claim (a).

Next, we prove claim (c).
Every element $v\in\mc W$ can be expanded as
$v=\deg^0v+\deg^1v+\deg^2v+\dots$.
Since $\mb C$ is central w.r.t. the $\lambda$-bracket,
$\deg^0u$ and $\deg^0v$ do not contribute to $\{u_\lambda v\}^{\mc W}$.
Moreover, by the Leibniz rules, for $i,j\geq1$,
$\{\deg^i u_\lambda \deg^j v\}^{\mc W}$ contributes only to degrees greater than or equal to $i+j-2$.
Hence,
$$
\deg^1
\{u_\lambda v\}^{\mc W}
=
\deg^1\{\deg^1 u_\lambda \deg^1 v\}^{\mc W}
+
\deg^1\{\deg^2 u_\lambda \deg^1 v\}^{\mc W}
+
\deg^1\{\deg^1 u_\lambda \deg^2 v\}^{\mc W}
\,.
$$
As a consequence, using claim (b), we get
$$
\overline{\deg}^1
\{u_\lambda v\}^{\mc W}
\big|_{\lambda=0}
=
\overline{\deg}^1
\{\deg^1 u_\lambda \deg^1 v\}^{\mc W}
\big|_{\lambda=0} 
=
\overline{\deg}^1\{\overline{\deg}^1 u_\lambda \overline{\deg}^1 v\}^{\mc W}\big|_{\lambda=0}
\,,
$$
by sesquilinearity.

Finally, we prove claim (d).
We need to distinguish the polynomial degree \eqref{eq:deg}
in the $\mc W$-algebra $\mc W(\mf g,f)$,
which, just for this proof, we denote $\deg_{\mc W}$,
from the polynomial degree in the algebra $\mc V(\mf g)$,
which we denote $\deg_{\mc V}$.
Recall from \cite{DSKV14} that
\begin{equation}\label{eq:last-pm1}
\deg^0_{\mc V}w(p)=0
\,\text{ and }\,
\overline{\deg}^1_{\mc V}w(p)=p
\,\text{ for all }\,
p\in\mf g^f
\,.
\end{equation}
As a consequence,
$$
\overline{\deg}^1_{\mc V}v
=
\overline{\deg}^1_{\mc V}\overline{\deg}^1_{\mc W}v
\,\text{ for all }\,
v\in\mc W(\mf g,f)
\,.
$$
for every $v\in\mc W(\mf g,f)$.
Moreover, it follows by Theorem \ref{thm:structure-W} and equation \eqref{eq:last-pm1} that
\begin{equation}\label{eq:last-pm2}
\overline{\deg}^1_{\mc W}v
=
w\big(
\overline{\deg}^1_{\mc V}
\overline{\deg}^1_{\mc W}v
\big)
=
w\big(
\overline{\deg}^1_{\mc V}v
\big)
\,\text{ for all }\,
v\in\mc W(\mf g,f)
\,.
\end{equation}
Then,
\begin{align*}
& \overline{\deg}^1_{\mc W}\{w(p)_\lambda w(q)\}^{\mc W}\big|_{\lambda=0}
= 
w\big(
\overline{\deg}^1_{\mc V}
\{w(p)_\lambda w(q)\}^{\mc W}\big|_{\lambda=0}
\big) \\
& =
w\big(
\overline{\deg}^1_{\mc V}
\rho\{w(p)_\lambda w(q)\}\big|_{\lambda=0}
\big)
\,.
\end{align*}
For the first equality we used equation \eqref{eq:last-pm2},
and for the second equality we used the definition \eqref{20120511:eq3}
of the $\mc W$-algebra $\lambda$-bracket.
To conclude the proof of claim (d), we need to show that
\begin{equation}\label{eq:last-pm3}
\overline{\deg}^1_{\mc V}
\rho\{w(p)_\lambda w(q)\}\big|_{\lambda=0}
=
[p,q]
\,.
\end{equation}
This can be easily checked using \eqref{eq:last-pm1}
and the axioms of Poisson vertex algebras.
\end{proof}
\begin{lemma}\label{lem:last2}
If $\big\{f_{ab;i}\,|\,1\leq a,b\leq r,\,0\leq i\leq\min\{p_a,p_b\}-1\big\}$ is the basis of $\mf g^f$
defined in \eqref{eq:basis-gf}, then
\begin{equation}\label{eq:last1}
[f_{ab;i},f_{cd,0}]
=
\delta_{b,c}\delta_{i,p_b-1}f_{ad;0}
-
\delta_{a,d}\delta_{i,p_a-1}f_{cb;0}
\,.
\end{equation}
\end{lemma}
\begin{proof}
Straightforward.
\end{proof}
\begin{lemma}\label{lem:last3}
If $\mc L(\partial)$ is the Lax operator \eqref{eq:WL},
we have
\begin{equation}\label{eq:last2a}
{\deg}^0\big(\mc L(\partial)^{\frac{j}{p_1}}\big)
=
(-1)^{\frac{j}{p_1}}
\id_{r_1}(-\partial)^j
\,,
\end{equation}
and
\begin{equation}\label{eq:last2b}
\overline{\deg}^1\big(\mc L(\partial)^{\frac{j}{p_1}}\big)
=
(-1)^{\frac{j-p_1}{p_1}}
\frac{j}{p_1}
W_{\bm1\bm1}(\partial)(-\partial)^{j-p_1}
\,.
\end{equation}
\end{lemma}
\begin{proof}
Applying $\deg^0$ to both sides of \eqref{eq:WL} and using \eqref{eq:degW}, we get
$$
-\id_{r_1}(-\partial)^{p_1}
=
\deg^0\mc L(\partial)
=
\big(\deg^0\mc L(\partial)^{\frac1{p_1}}\big)^{p_1}
\,,
$$
so that
\begin{equation}\label{eq:last2a-1}
\deg^0\big(\mc L(\partial)^{\frac1{p_1}}\big)
=
(-1)^{\frac1{p_1}}\id_{r_1}(-\partial)
\,.
\end{equation}
Equation \eqref{eq:last2a} is an immediate consequence of \eqref{eq:last2a-1}.
Next, applying $\deg^1$ to both sides of \eqref{eq:WL}, we get, by \eqref{eq:last2a-1},
\begin{align*}
& W_{\bm1\bm1}(\partial)
=
\deg^1\mc L(\partial)
=
\deg^1\big(\big(\mc L(\partial)^{\frac1{p_1}}\big)^{p_1}\big)
=
\deg^1\big(
(-1)^{\frac1{p_1}}\id_{r_1}(-\partial)
+
\deg^1\big(\mc L(\partial)^{\frac1{p_1}}\big)
\big)^{p_1} \\
& =
\sum_{i=0}^{p_1-1}
(-1)^{\frac{p_1-1}{p_1}}(-\partial)^{p_1-1-i}
\deg^1\big(\mc L(\partial)^{\frac1{p_1}}\big)
(-\partial)^i
\,.
\end{align*}
Hence, further applying $\overline{\deg}^1$, we get
$$
W_{\bm1\bm1}(\partial)
=
(-1)^{\frac{p_1-1}{p_1}}
p_1
\overline{\deg}^1\big(\mc L(\partial)^{\frac1{p_1}}\big)
(-\partial)^{p_1-1}
\,,
$$
so that
\begin{equation}\label{eq:last2b-1}
\overline{\deg}^1\big(\mc L(\partial)^{\frac1{p_1}}\big)
=
\frac{(-1)^{\frac{1-p_1}{p_1}}}{p_1}
W_{\bm1\bm1}(\partial)
(-\partial)^{1-p_1}
\,,
\end{equation}
which is the same as \eqref{eq:last2b} for $j=1$.
Then, using \eqref{eq:last2a-1} and \eqref{eq:last2b-1}, we get
\begin{align*}
& \overline{\deg}^1\big(\mc L(\partial)^{\frac{j}{p_1}}\big)
=
\overline{\deg}^1\big(\mc L(\partial)^{\frac{1}{p_1}}\big)^j
=
\overline{\deg}^1\big(
(-1)^{\frac1{p_1}}\id_{r_1}(-\partial)
+
\deg^1
\mc L(\partial)^{\frac{1}{p_1}}
\big)^j \\
& =
\overline{\deg}^1\Big(
\sum_{i=0}^{j-1}
(-1)^{\frac{j-1}{p_1}}
(-\partial)^{j-1-i}
\deg^1
\mc L(\partial)^{\frac{1}{p_1}}
(-\partial)^{i} 
\Big) \\
& =
j (-1)^{\frac{j-1}{p_1}}
\overline{\deg}^1\big(\mc L(\partial)^{\frac{1}{p_1}}\big)
(-\partial)^{j-1}
=
(-1)^{\frac{j-p_1}{p_1}}
\frac{j}{p_1}
W_{\bm1\bm1}(\partial)
(-\partial)^{j-p_1}
\,.
\end{align*}
\end{proof}
\begin{proposition}\label{prop:sylvain2}
In equation \eqref{eq:evol1b} it must be $\alpha=0$.
\end{proposition}
\begin{proof}
We start from the first equation in \eqref{eq:evol1b}, which defines the constant $\alpha\in\mb C$:
\begin{equation}\label{eq:last3}
\frac{\partial}{\partial t_j} W_{\bm1\bm2}(\partial)
=
R_{\bm1\bm2}^{(j)}(\partial)
+\alpha
W_{\bm1\bm2}(\partial)
\,.
\end{equation}
We take the order $0$ and linear (in the polynomial degree of $\mc W(\mf{gl}_N,\ul p)$) 
constribution in both sides of \eqref{eq:last3}.
For this, denote
$$
W_{\bm1\bm2;0}=\ord^0W_{\bm1\bm2}(\partial)
=
\big(w_{ba;0}\big)_{1\leq a\leq r_1<b\leq r}
\,\text{ and }\,
R^{(j)}_{\bm1\bm2;0}=\ord^0R^{(j)}_{\bm1\bm2}(\partial)
\,.
$$
Then, we get, recalling \eqref{eq:degW},
\begin{equation}\label{eq:last4}
\overline{\deg}^1\big(\frac{\partial}{\partial t_j} W_{\bm1\bm2;0}\big)
=
\overline{\deg}^1R_{\bm1\bm2;0}^{(j)}
+\alpha
W_{\bm1\bm2;0}
\,.
\end{equation}
We compute separately the two terms 
$\overline{\deg}^1\big(\frac{\partial}{\partial t_j} W_{\bm1\bm2;0}\big)$
and
$\overline{\deg}^1R_{\bm1\bm2;0}^{(j)}$.
For the first one we have, for $1\leq a\leq r_1<b\leq r$,
\begin{align*}
& \overline{\deg}^1\big(\frac{\partial w_{ba;0}}{\partial t_j} \big)
=
\overline{\deg}^1\big(
\{\tint h_j,w_{ba;0}\}^{\mc W}
\big) 
=
\frac{p_1}{j}
\overline{\deg}^1
\big\{
{\Res_\partial\tr\big(\mc L(\partial)^{\frac{j}{p_1}}\big)}_\lambda 
w_{ba;0}
\big\}^{\mc W}
\big|_{\lambda=0} \\
& =
\frac{p_1}{j}
\overline{\deg}^1
\big\{
{\Res_\partial\tr \overline{\deg}^1 \big(\mc L(\partial)^{\frac{j}{p_1}}\big)}_\lambda 
w_{ba;0}
\big\}^{\mc W}
\big|_{\lambda=0} \\
& =
(-1)^{\frac{j-p_1}{p_1}}
\overline{\deg}^1
\big\{
\Res_\partial\tr 
\big(W_{\bm1\bm1}(\partial)(-\partial)^{j-p_1}\big)
_\lambda 
w_{ba;0}
\big\}^{\mc W}
\big|_{\lambda=0}
\,.
\end{align*}
For the first equality we used the definition of the Hamiltonian equations \eqref{eq:hameq},
for the second equality we used the definition \eqref{eq:densities} of the Hamiltonian densities $h_j$,
for the third equality we used Lemma \ref{lem:last1}(c),
and for the fourth equality we used equation \eqref{eq:last2b}.
The RHS above obviously vanishes for $j\geq p_1$,
while, for $1\leq j\leq p_1-1$, it is, by equation \eqref{eq:wabk} and Lemma \ref{lem:last1}(d),
$$
-
(-1)^{\frac{j-p_1}{p_1}}
\sum_{c=1}^{r_1}
\overline{\deg}^1
\big\{
{w_{cc;p_1-j-1}}
_\lambda 
w_{ba;0}
\big\}^{\mc W}
\big|_{\lambda=0} 
=
- 
(-1)^{\frac{j-p_1}{p_1}}
\sum_{c=1}^{r_1}
w[f_{cc;p_1-j-1},f_{ba;0}]
\,,
$$
which vanishes by Lemma \ref{lem:last2},
since $j\neq0$.
Hence,
\begin{equation}\label{eq:last5}
\overline{\deg}^1\big(\frac{\partial}{\partial t_j} W_{\bm1\bm2;0}\big)
=
0
\,.
\end{equation}
Next, we compute $\overline{\deg}^1R_{\bm1\bm2;0}^{(j)}$.
If we take the degree $0$ contribution of both sides of the first equation in \eqref{eq:QR1},
we get, recalling \eqref{eq:degW},
$$
0=
\deg^0 Q_{\bm1\bm2}^{(j)}(\partial)
(-(-\partial)^{\ul q})
+
\deg^0 R_{\bm1\bm2}^{(j)}(\partial)
\,,
$$
from which we immediately conclude, by \eqref{eq:QR2}, that
$$
\deg^0 Q_{\bm1\bm2}^{(j)}(\partial)=0
\,\text{ and }\,
\deg^0 R_{\bm1\bm2}^{(j)}(\partial)=0
\,.
$$
We then use this and equation \eqref{eq:last2a}
to compute the degree $1$ contribution of both sides 
of the first equation in \eqref{eq:QR1}:
$$
(-1)^{\frac{j}{p_1}}
(-\partial)^j
W_{\bm1\bm2}(\partial)
=
\deg^1Q_{\bm1\bm2}^{(j)}(\partial)
(-(-\partial)^{\ul q})
+
\deg^1R_{\bm1\bm2}^{(j)}(\partial)
\,.
$$
Taking the order $0$ contribution of both sides, we get
$$
\deg^1R_{\bm1\bm2;0}^{(j)}
=
(-1)^j(-1)^{\frac{j}{p_1}}
(W_{\bm1\bm2;0})^{\prime\dots\prime}
\,,
$$
where in the RHS we are taking $j$ derivatives.
As an immediate consequence, we get
\begin{equation}\label{eq:last6}
\overline{\deg}^1R_{\bm1\bm2;0}^{(j)}=0
\,.
\end{equation}
Combining \eqref{eq:last4}, \eqref{eq:last5} and \eqref{eq:last6}, we conclude that $\alpha=0$,
proving the claim.
\end{proof}

\subsection{Proof of Theorem \ref{prop:sylvain}}\label{sec:sylv6}

By the last assertion in Theorem \ref{thm:laxeq},
the Hamiltonian equation \eqref{eq:hameq} implies the Lax equation \eqref{eq:527}.
Moreover, by Corollary \ref{prop:dW4dtn} (see also Remark \ref{rem:dan}),
equation \eqref{eq:hameq} also implies \eqref{eq:evol2}.
Conversely, 
by Propositions \ref{prop:sylvain1} and \ref{prop:sylvain2},
the Lax equation \eqref{eq:527} and equation \eqref{eq:evol2} 
uniquely determine the evolution of all generators of the $\mc W$-algebra,
given by equations \eqref{eq:evol1}.
Hence, by uniqueness, this evolution must coincide with the Hamiltonian equation \eqref{eq:hameq}.
The claim follows.

\section{
Proof of Theorem \ref{thm:main}
}
\label{sec:9b}

By Proposition \ref{thm:lax-sol} the matrix pseudodifferential operator 
$\mc L(\ul m,x,\bm t,\partial)\in\Mat_{r_1\times r_1}\mc F((\partial^{-1}))$ 
given by \eqref{eq:WLsol} solves the Lax equations \eqref{eq:527}.
Note that equation \eqref{eq:WLsol} is the same as equation \eqref{eq:WL},
since, by \eqref{eq:W22-sol}, we set $W_{\bm2\bm2}(\ul m,x,\bm t,\partial)=0$;
in particular, equation \eqref{eq:evol2} obviously holds.
We can thus apply Theorem \ref{prop:sylvain}
to conclude that 
$W(\ul m,x,\bm t,\partial)\in\Mat_{r\times r}\mc F[\partial]$
evolves according to the Hamiltonian equations \eqref{eq:hameq},
as claimed.

\section{
Examples
}
\label{sec:ex}

As a direct application of Theorem \ref{thm:laxeq} and Theorem \ref{prop:sylvain} we give the integrable hierarchies associated to the partitions $(p,1,...,1)$ where $p>1$ and $(p,2)$ where $p>2$. In the second case, we only consider the first equation of the hierarchy, so that the explicit evolution of the generators of the $\mc W$-algebra can be given.

\begin{example}
Consider the partition $\ul p=(p,1,...,1)$ with $r>1$ parts, where $p>1$. The generators of the $\mc{W}$-algebra $\mc W(\mf{gl}_N,\ul p)$ are given by the coefficients of the (square) matrix differential operator $W(\partial)$ of size $r$, composed of four blocks $W_{\bm1\bm1}(\partial)$, $W_{\bm1\bm2}(\partial)$, $W_{\bm2\bm1}(\partial)$ and $W_{\bm2\bm2}(\partial)$ of size $ 1 \times 1$, $1 \times( r-1)$, $(r-1) \times 1$ and $(r-1) \times (r-1)$, such that all $W_{ij}$'s are order $0$ differential operators except for $W_{\bm1\bm1}(\partial)$, which is of order $p-1$. Since $W_{\bm 1 \bm 2}(\partial)$, $W_{\bm 2 \bm 1}(\partial)$ and $W_{\bm 2 \bm 2}(\partial)$ do not depend on $\partial$ in this example, we will simply denote them by $W_{\bm 1 \bm 2}$, $W_{\bm 2 \bm 1}$ and $W_{\bm 2 \bm 2}$.
 Let $n \geq 1$. Consider the evolutionary derivations $d/d{t_n}$ of $\mathcal{W}(\mf{gl}_N,\ul p)$ (i.e. commuting with $\partial$) given by the Hamiltonians 
 $$ \int h_n= \int res_{\partial} \, \frac{p}{n}{{\mc L(\partial)}^{n/p}}, \, \, \text{where}\, \, \mc L(\partial)=-(-\partial)^p +W_{\bm 1 \bm 1}(\partial)-W_{ \bm 1 \bm 2} ( \mathbbm{1}_{r-1}\partial+W_{\bm 2 \bm 2})^{-1} W_{\bm 2 \bm 1}.$$ We know by Theorem \ref{thm:laxeq} that $\mc L(\partial)$ evolves according to the Lax equation
\begin{equation} \label{eqp00}
\frac{d{\mc L}}{dt_n}(\partial)=[B_n(\partial), \mc L(\partial)],
\end{equation} 
where $B_n(\partial)=({\mc L(\partial)}^{n/p})_+$. Let $C_n(\partial)$, $D_n(\partial)$, $R_n$ and $S_n$ be the unique matrix differential operators such that 
\begin{equation}
\begin{split}
B_n(\partial)W_{\bm 1 \bm 2}=C_n(\partial)(\partial+W_{\bm 2 \bm 2})+R_n, \\
W_{\bm 2 \bm 1}B_n(\partial)=(\partial+W_{\bm 2 \bm 2})D_n(\partial)+S_n,
\end{split}
\end{equation}
where $R_n$ and $S_n$ are zero order row and column differential operators (no dependence of $\partial$). Then, by Theorem \ref{prop:sylvain}, we can describe explicitly the evolution of the generators of $\mathcal{W}(\mf{gl}_N,\ul p)$ as follows:
\begin{equation} \label{eqp11}
\begin{split}
\frac{dW_{\bm 1 \bm 1}}{dt_n} (\partial)&=[B_n(\partial), I_{r-1}(-\partial)^p+W_{\bm 1 \bm 1}(\partial)]-C_n(\partial)W_{\bm 2 \bm 1}+W_{ \bm 1 \bm 2}D_n(\partial), \\
\frac{dW_{\bm 1 \bm 2}}{dt_n}  &=R_n, \\
\frac{dW_{\bm 2 \bm 1}}{dt_n}  &=-S_n, \\
\frac{dW_{\bm 2 \bm 2}}{dt_n} &=0.
\end{split}
\end{equation}
 The zero order differential operators $C_1$, $D_1$, $R_1$ and $S_1$ are given by
\begin{equation*}
\begin{split}
C_1&=W_{\bm 1 \bm 2} ,\, \,  R_1=W_{\bm 1 \bm 2}'+bW_{\bm 1 \bm 2}-W_{\bm 1 \bm 2}W_{\bm 2 \bm 2} ,\\
D_1&=W_{\bm 2 \bm 1}, \, \,  S_1=-W_{\bm 2 \bm 1}'+W_{\bm 2 \bm 1}b-W_{\bm 2 \bm 2}W_{\bm 2 \bm 1}.
\end{split}
\end{equation*}
Note that the set of equations \eqref{eqp11} can be reduced by letting $W_{\bm 2 \bm 2}=0$, in which case we have
\begin{equation} \label{eqp13}
\frac{dW_{\bm 1 \bm 1}}{dt_n} (\partial)=[B_n(\partial),(-\partial)^p +W_{\bm 1 \bm 1}(\partial)]-C_n(\partial)W_{\bm 2 \bm 1}+W_{ \bm 1 \bm 2}D_n(\partial), 
\end{equation}
\begin{equation} \label{eqp14}
\frac{dW_{\bm 1 \bm 2}}{dt_n}  =B_n(W_{\bm 1 \bm 2}), 
\end{equation}
\begin{equation} \label{eqp15}
\frac{dW_{\bm 2 \bm 1}}{dt_n}  =-B_n^*(W_{\bm 2 \bm 1}).
\end{equation} 
Equations \eqref{eqp00}, \eqref{eqp14} and \eqref{eqp15} are precisely the equations of the well-known $p$-constrained $(r-1)$-vector KP hierarchy (see e.g. \cite{Zhang99}).
\end{example}

\begin{example}\label{ex2}
Consider the partition $\ul p=(p,2)$ where $p>2$. The corresponding $\mc W$-algebra is generated, as a differential algebra, by the coefficients of the entries of a $2 \times 2$ matrix $W(\partial)$ such that the orders of the differential operators $W_{\bm 1 \bm 1}(\partial)$, $W_{\bm 1 \bm 2}(\partial)$, $W_{\bm 2 \bm 1}(\partial)$ and $W_{\bm 2 \bm 2}(\partial)$ are respectively $p-1$, $1$, $1$ and $1$. Explicitly 
\begin{align*}
W_{\bm 1 \bm 2}(\partial)&=-w_{21;1}\partial+w_{21;0}, \\
W_{\bm 2 \bm 1}(\partial)&=-w_{12;1}\partial +w_{12;0}=-\partial  \circ w_{12;1}+w_{12;0}+w_{12;1}', \\
W_{\bm 2 \bm 2}(\partial)&=-w_{22;1} \partial +w_{22;0}, \\ 
W_{\bm 1 \bm 1}(\partial)&=w_{11;p-1}(-\partial)^{p-1}+...+w_{11;0}.
\end{align*}
 Let $d/dt$ be the evolutionary derivation of $\mc W(\mf{gl}_N,\ul p)$ associated to the Hamiltonian $\int h=p\int res_{\partial} \, \, {{\mc L(\partial)}^{1/p}}$, where 
 $$\mc L(\partial)=-(-\partial)^p+W_{\bm1\bm1}(\partial)-W_{\bm1\bm2}(\partial)(-\partial^2+W_{\bm2\bm2}(\partial))^{-1}W_{\bm2\bm1}(\partial).$$
 We know by Theorem \ref{thm:laxeq} that the pseudodifferential operator $\mc L(\partial)$ evolves according to the Lax equation 
\begin{equation} \label{BL2}
\frac{d{\mc L}}{dt}(\partial)=[B(\partial),\mc L(\partial)],
\end{equation} 
where $B(\partial)=-\alpha( \partial + \frac{1}{p} w_{11;p-1})$ is the differential part of ${\mc L(\partial)}^{1/p}$. Note that $\alpha$ is a $p$-th root of $-1$. Let $c, d$ be the unique elements of $\mc W(\mf{gl}_N, \ul p)$ and $R(\partial), S(\partial)$ be the unique differential operators of order $1$ such that 
\begin{equation}
\begin{split}
B(\partial)W_{\bm1\bm2}(\partial)=c(-\partial^2+W_{\bm2\bm2}(\partial))+R(\partial), \\
W_{\bm2\bm1}(\partial)B(\partial)=(-\partial^2+W_{\bm2\bm2}(\partial))d+S(\partial).
\end{split}
\end{equation}
By Theorem \ref{prop:sylvain}, we can describe the evolution of the generators of $\mc W(\mf{gl}_N,\ul p)$ as follows:
\begin{equation}
\begin{split}
\frac{dW_{\bm1\bm1}}{dt}(\partial) &=[B(\partial),-(-\partial)^p+W_{\bm1\bm1}(\partial)]-cW_{\bm2\bm1}(\partial)+W_{\bm1\bm2}(\partial)d, \\
\frac{dW_{\bm1\bm2}}{dt}(\partial) &=R(\partial), \, \, \, \,
\frac{dW_{\bm2\bm1}}{dt}(\partial)  =-S(\partial), \, \,  \, \,
\frac{dW_{\bm2\bm2}}{dt}(\partial) =0 .
\end{split}
\end{equation}
Explicitly,  $R(\partial)=d\partial+e$ and $S(\partial)=\partial \circ f +g$ are given by
\begin{equation} \label{eqp2}
\begin{split}
d&= \alpha(\frac{1}{p} w_{11;p-1}w_{21;1}+{w_{21;1}}'-w_{21;0}-w_{21;1}w_{22;1}),\\
e&=-\alpha(\frac{1}{p}w_{11;p-1}w_{21;0}+{w_{21;0}}'+w_{21;1}w_{22.0} ),\\
f&= \alpha(\frac{1}{p} w_{11;p-1}w_{12;1}- 2{w_{12;1}}'+w_{12;0}-w_{12;1}w_{22;1}),  \\
g&=\alpha(-\frac{1}{p}w_{11;p-1}(w_{12;0}+{w_{12;1}}')+{w_{12;1}}''+{w_{12;0}}'+w_{12;1}(w_{22;0}+{w_{22;1}}')).
\end{split}
\end{equation}
 Moreover, $c=-\alpha w_{21;1}$ and $d=-\alpha w_{12;1}$, hence $cW_{\bm2\bm1}(\partial)-W_{\bm1\bm2}(\partial)d=\alpha(w_{12;1}w_{21;0}-w_{21;1}w_{12;0})$.
Note that the equations \eqref{eqp2} can be reduced by letting $W_{\bm 2 \bm 2}(\partial)=0$, in which case they can be rewritten as 
\begin{equation} \label{eqp2red}
\begin{split}
d&= \alpha(\frac{1}{p} w_{11;p-1}w_{21;1}+{w_{21;1}}'-w_{21;0}),\\
e&=-\alpha(\frac{1}{p}w_{11;p-1}w_{21;0}+{w_{21;0}}' ),\\
f&= \alpha(\frac{1}{p} w_{11;p-1}w_{12;1}- 2{w_{12;1}}'+w_{12;0}),  \\
g&=\alpha(-\frac{1}{p}w_{11;p-1}(w_{12;0}+{w_{12;1}}')+{w_{12;1}}''+{w_{12;0}}') ,
\end{split}
\end{equation}
which imply the four evolution equations
\begin{equation} \label{eqp3red}
\begin{split}
\frac{d w_{21;1}}{dt_{red}} &= -\alpha(\frac{1}{p} w_{11;p-1}w_{21;1}+{w_{21;1}}'-w_{21;0}),\\
\frac{d w_{21;0}}{dt_{red}}&=-\alpha(\frac{1}{p}w_{11;p-1}w_{21;0}+{w_{21;0}}' ),\\
\frac{d w_{12;1}}{dt_{red}}&= \alpha(\frac{1}{p} w_{11;p-1}w_{12;1}- 2{w_{12;1}}'+w_{12;0}),  \\
\frac{d (w_{12;0}+{w_{12;1}}')}{dt_{red}}&=-\alpha(-\frac{1}{p}w_{11;p-1}(w_{12;0}+{w_{12;1}}')+{w_{12;1}}''+{w_{12;0}}').
\end{split}
\end{equation}
In particular,
\begin{equation}
\begin{split}
\frac{d}{dt_{red}} (-w_{21;1}+w_{21;0}x)&=B(-w_{21;1}+w_{21;0}x), \\
\frac{d }{dt_{red}}(w_{12;1}+(w_{12;1}'+w_{12;0})x) &= -B^*(w_{12;1}+(w_{12;1}'+w_{12;0})x).
\end{split}
\end{equation}
\end{example}

\section{Solutions}
\subsection{Polynomial tau-functions}
\begin{example} 
To construct a tau-function for Example $14.1$, we let, cf.  \eqref{eq:hi},
\[
h=h_1=S_{(r+1)p+1}(\bm t^{(1)})+\sum_{i=2}^r S_{r+1}(\bm t^{(i)} +\bm c^{(i)}).
\]
To construct the corresponding Lax equation we let $\ul m=(2,1,1,\ldots,1)$, and calculate the following $2r-1$ tau-functions
\[
\tau^{(2,1,\ldots,1)}(\underline{\bm t})=
\left|
\begin{matrix}
S_{rp-1}(\bm t^{(1)})&S_{(r-1)p-1}(\bm t^{(1)})&\cdots&S_{p-1}(\bm t^{(1)})&0\\
S_{rp}(\bm t^{(1)})&S_{(r-1)p}(\bm t^{(1)})&\cdots&S_{p}(\bm t^{(1)})&1\\[2mm]
S_{r}(\bm t^{(2)}+\bm c^{(2)})&S_{r-1}(\bm t^{(2)}+\bm c^{(2)})&\cdots &S_{1}(\bm t^{(2)}+\bm c^{(2)})&1\\
S_{r}(\bm t^{(3)}+\bm c^{(3)})&S_{r-1}(\bm t^{(3)}+\bm c^{(3)})&\cdots &S_{1}(\bm t^{(3)}+\bm c^{(3)})&1\\
\vdots&\vdots&  &\vdots&\vdots\\
S_{r}(\bm t^{(r)}+\bm c^{(r)})&S_{r-1}(\bm t^{(r)}+\bm c^{(r)})&\cdots &S_{1}(\bm t^{(r)}+\bm c^{(r)})&1\\
\end{matrix}
\right|,
\]
\[
\tau^{(3,1,\ldots,1,0,1,\ldots,1)}(\underline{\bm t})=
\left|
\begin{matrix}
S_{rp-2}(\bm t^{(1)})&S_{(r-1)p-2}(\bm t^{(1)})&\cdots&S_{p-2}(\bm t^{(1)})&0\\
S_{rp-1}(\bm t^{(1)})&S_{(r-1)p-1}(\bm t^{(1)})&\cdots&S_{p-1}(\bm t^{(1)})&0\\
S_{rp}(\bm t^{(1)})&S_{(r-1)p}(\bm t^{(1)})&\cdots&S_{p}(\bm t^{(1)})&1\\[2mm]
S_{r}(\bm t^{(2)}+\bm c^{(2)})&S_{r-1}(\bm t^{(2)}+\bm c^{(2)})&\cdots &S_{1}(\bm t^{(2)}+\bm c^{(2)})&1\\
\vdots&\vdots& &\vdots&\vdots\\
S_{r}(\bm t^{(i-1)}+\bm c^{(i-1)})&S_{r-1}(\bm t^{(i-1)}+\bm c^{(i-1)})&\cdots &S_{1}(\bm t^{(i-1)}+\bm c^{(i-1)})&1\\[2mm]
S_{r}(\bm t^{(i+1)}+\bm c^{(i+1)})&S_{r-1}(\bm t^{(i+1)}+\bm c^{(i+1)})&\cdots &S_{1}(\bm t^{(i+1)}+\bm c^{(i+1)})&1\\
\vdots&\vdots& &\vdots&\vdots\\
S_{r}(\bm t^{(r)}+\bm c^{(r)})&S_{r-1}(\bm t^{(r)}+\bm c^{(r)})&\cdots &S_{1}(\bm t^{(r)}+\bm c^{(r)})&1\\
\end{matrix}
\right|,
\]
\[
\tau^{(1,1,\ldots,1,2,1,\ldots,1)}(\underline{\bm t})=
\left|
\begin{matrix}
S_{rp}(\bm t^{(1)})&S_{(r-1)p}(\bm t^{(1)})&\cdots&S_{p}(\bm t^{(1)})&1\\[2mm]
S_{r}(\bm t^{(2)}+\bm c^{(2)})&S_{r-1}(\bm t^{(2)}+\bm c^{(2)})&\cdots &S_{1}(\bm t^{(2)}+\bm c^{(2)})&1\\
\vdots&\vdots& &\vdots&\vdots\\
S_{r}(\bm t^{(i-1)}+\bm c^{(i-1)})&S_{r-1}(\bm t^{(i-1)}+\bm c^{(i-1)})&\cdots &S_{1}(\bm t^{(i-1)}+\bm c^{(i-1)})&1\\[2mm]
S_{r-1}(\bm t^{(i)}+\bm c^{(i)})&S_{r-2}(\bm t^{(i)}+\bm c^{(i)})&\cdots &1&0\\
S_{r}(\bm t^{(i)}+\bm c^{(i)})&S_{r-1}(\bm t^{(i)}+\bm c^{(i)})&\cdots &S_{1}(\bm t^{(i)}+\bm c^{(i)})&1\\[2mm]
S_{r}(\bm t^{(i+1)}+\bm c^{(i+1)})&S_{r-1}(\bm t^{(i+1)}+\bm c^{(i+1)})&\cdots &S_{1}(\bm t^{(i+1)}+\bm c^{(i+1)})&1\\
\vdots&\vdots& &\vdots&\vdots\\
S_{r}(\bm t^{(r)}+\bm c^{(r)})&S_{r-1}(\bm t^{(r)}+\bm c^{(r)})&\cdots &S_{1}(\bm t^{(r)}+\bm c^{(r)})&1\\
\end{matrix}
\right|.
\]
In the latter two cases we have $0$, respectively $2$, in the $i$-th place in 
the upper index of the tau-function.
Following the procedure of Section \ref{S4}, let $S_k(\bm c^{(i)})=\alpha_{i,k}\in\mathbb{C}$,
$T(\bm t)=\tau^{(2,1,\ldots,1)}(\bm t,0,\ldots,0)$, $q^{(i)}(\bm t)=\frac{T^{(i)}(\bm t)}{T(\bm t)}$, 
$r^{(i)}(\bm t)=-\frac{T^{(-i)}(\bm t)}{T(\bm t)}$, where
$T^{(i)}(\bm t)=\tau^{(3,1,\ldots,1,0,1,\ldots,1)}(\bm t,0,\ldots,0)$ and 
$T^{(-i)}(\bm t)=\tau^{(1,1,\ldots,1,2,1,\ldots,1)}(\bm t,0,\ldots,0)$.
Then 
\[
T(\bm t)=
\left|
\begin{matrix}
S_{rp-1}(\bm t)&S_{(r-1)p-1}(\bm t)&\cdots&S_{p-1}(\bm t)&0\\
S_{rp}(\bm t)&S_{(r-1)p}(\bm t)&\cdots&S_{p}(\bm t)&1\\[2mm]
\alpha_{2,r}&\alpha_{2,r-1}&\dots&\alpha_{2,1}&1\\
\alpha_{3,r}&\alpha_{3,r-1}&\dots&\alpha_{3,1}&1\\
\vdots&\vdots&  &\vdots&\vdots\\
\alpha_{r,r}&\alpha_{r,r-1}&\dots&\alpha_{r,1}&1\\
\end{matrix}
\right|,
\]
\[
T^{(i)}(\bm t)=
\left|
\begin{matrix}
S_{rp-2}(\bm t)&S_{(r-1)p-2}(\bm t)&\cdots&S_{p-2}(\bm t)&0\\
S_{rp-1}(\bm t)&S_{(r-1)p-1}(\bm t)&\cdots&S_{p-1}(\bm t)&0\\
S_{rp}(\bm t)&S_{(r-1)p}(\bm t)&\cdots&S_{p}(\bm t)&1\\[2mm]
\alpha_{2,r}&\alpha_{2,r-1}&\cdots&\alpha_{2,1}&1\\
\vdots&\vdots& &\vdots&\vdots\\
\alpha_{i-1,r}&\alpha_{i-1,r-1}&\cdots&\alpha_{i-1,1}&1\\[2mm]
\alpha_{i+1,r}&\alpha_{i+1,r-1}&\cdots&\alpha_{i+1,1}\\
\alpha_{r,r}&\alpha_{r,r-1}&\cdots&\alpha_{r,1}&1\\
\end{matrix}
\right|,
\]
\[
T^{(-i)}(\bm t)=
\left|
\begin{matrix}
S_{rp}(\bm t)&S_{(r-1)p}(\bm t)&\cdots&S_{p}(\bm t)&1\\[2mm]
\alpha_{2,r}&\alpha_{2,r-1}&\dots&\alpha_{2,1}&1\\
\vdots&\vdots& &\vdots&\vdots\\
\alpha_{i-1,r}&\alpha_{i-1,r-1}&\cdots&\alpha_{i-1,1}&1\\[2mm]
\alpha_{i,r-1}&\alpha_{i,r-2}&\cdots&1&0\\
\alpha_{i,r}&\alpha_{i,r-1}&\cdots&\alpha_{i,1}&1\\[2mm]
\alpha_{i+1,r}&\alpha_{i+1,r-1}&\cdots&\alpha_{i+1,1}\\
\alpha_{r,r}&\alpha_{r,r-1}&\cdots&\alpha_{r,1}&1\\
\end{matrix}
\right|.
\]
Next (cf. \eqref{eq:Q}), $Q^+_{11}(\ul m,\bm t,z)=\frac1{T(\bm t)}\times$
\[
\left|
\begin{matrix}
\sum_{k=0}^{rp-1}S_{rp-k-1}(\bm t)z^{-k}&\sum_{k=0}^{(r-1)p-1}S_{(r-1)p-k-1}(\bm t)z^{-k}&\cdots&\sum_{k=0}^{p-1}S_{p-k-1}(\bm t)z^{-k}&0\\
\sum_{k=0}^{rp}S_{rp-k}(\bm t)z^{-k}&\sum_{k=0}^{(r-1)p}S_{(r-1)p-k}(\bm t)z^{-k}&\cdots&\sum_{k=0}^{p}S_{p-k}(\bm t)z^{-k}&0\\
S_{rp}(\bm t^{(1)})&S_{(r-1)p}(\bm t^{(1)})&\cdots&S_{p}(\bm t^{(1)})&1\\[2mm]
\alpha_{2,r}&\alpha_{2,r-1}&\dots&\alpha_{2,1}&1\\
\alpha_{3,r}&\alpha_{3,r-1}&\dots&\alpha_{3,1}&1\\
\vdots&\vdots&  &\vdots&\vdots\\
\alpha_{r,r}&\alpha_{r,r-1}&\dots&\alpha_{r,1}&1\\
\end{matrix}
\right|
\]
and 
$Q^+_{11}(\ul m,\bm t,-z)=\frac1{T(\bm t)}\times$
\[
\left|
\begin{matrix}
S_{rp-1}(\bm t)-S_{rp-2}(\bm t)z^{-1}&S_{(r-1)p-1}(\bm t)-S_{(r-1)p-2}(\bm t)z^{-1}&\cdots&S_{p-1}(\bm t)-S_{p-2}(\bm t)z^{-1}&0\\
S_{rp}(\bm t)-S_{rp-1}(\bm t)z^{-1}&S_{(r-1)p}(\bm t)-S_{(r-1)p-1}(\bm t)z^{-1}&\cdots&S_{p}(\bm t)-S_{p-1}(\bm t)z^{-1}&1\\[2mm]
\alpha_{2,r}&\alpha_{2,r-1}&\dots&\alpha_{2,1}&1\\
\alpha_{3,r}&\alpha_{3,r-1}&\dots&\alpha_{3,1}&1\\
\vdots&\vdots&  &\vdots&\vdots\\
\alpha_{r,r}&\alpha_{r,r-1}&\dots&\alpha_{r,1}&1\\
\end{matrix}
\right|,
\]
Let $Q_{11}^\pm(\ul m,x,\bm t,z)$,  $q^{(i)}(x,\bm t)$ and  $r^{(i)}(x,\bm t)$ be defined as in \eqref{eq:Q2} by replacing $t_1$ by $t_1+x$. Then
\[
\tilde{\mc L}(\ul m, x,\bm t,\partial)^{\frac{j}{p}}=Q^+_{11}(\ul m,x,\bm t,\partial)\circ \partial^j \circ Q^-_{11}(\ul m,x,\bm t,\partial)^*
\]
and
\[
\begin{split}
\tilde{\mc L}(\ul m, x,\bm t,\partial)_-&=\sum_{i=2}^r q^{(i)}(x,\bm t)\partial^{-1}\circ r^{(i)}(x,\bm t)\\
&=(q^{(2)}(x,\bm t),\ldots,q^{(r)}(x,\bm t))\partial^{-1}\circ (r^{(2)}(x,\bm t),\ldots,r^{(r)}(x,\bm t))^T\, .
\end{split}
\]
\end{example}
\begin{example}

To construct a tau-function for Example $14.2$, we let, cf.  \eqref{eq:hi}, where $a\in\mathbb{C}$ is non-zero, 
\[
\begin{split}
h_1=&S_{p+3}(\bm t^{(1)})+S_{4}(\bm t^{(2)}+\bm c_1),\\
h_2=&S_{p+2}(\bm t^{(1)})+aS_{4}(\bm t^{(2)}+\bm c_2).
\end{split}
\]
We will give three tau-functions which are non-zero, viz. 
\[
\tau^{(2,2)}=
\left|
\begin{matrix}
S_{p+1}(\bm t^{(1)})&S_1(\bm t^{(1)})&S_{p}(\bm t^{(1)})&1\\
S_{p+2}(\bm t^{(1)})&S_2(\bm t^{(1)})&S_{p+1}(\bm t^{(1)})&S_1(\bm t^{(1)})\\[2mm]
S_{2}(\bm t^{(2)}+\bm c_1)&1& aS_{2}(t^{(2)}+\bm c_2)&a\\
S_{3}(\bm t^{(2)}+\bm c_1)&S_1(\bm t^{(2)}\bm c_1)&a S_{3}(\bm t^{(2)}+\bm c_2)&a S_1(\bm t^{(2)}+\bm c_2)\\
\end{matrix}
\right|,
\]
\[
\tau^{(3,1)}=
\left|
\begin{matrix}
S_{p}(\bm t^{(1)})&1&S_{p-1}(\bm t^{(1)})&0\\
S_{p+1}(\bm t^{(1)})&S_1(\bm t^{(1)})&S_{p}(\bm t^{(1)})&1\\
S_{p+2}(\bm t^{(1)})&S_2(\bm t^{(1)})&S_{p+1}(\bm t^{(1)})&S_1(\bm t^{(1)})\\[2mm]
S_{3}(\bm t^{(2)}+\bm c_1)&S_1(\bm t^{(2)}\bm c_1)&a S_{3}(\bm t^{(2)}+\bm c_2)&a S_1(\bm t^{(2)}+\bm c_2)\\
\end{matrix}
\right|,
\]
\[
\tau^{(1,3)}=
\left|
\begin{matrix}
S_{p+2}(\bm t^{(1)})&S_2(\bm t^{(1)})&S_{p+1}(\bm t^{(1)})&S_1(\bm t^{(1)})\\[2mm]
S_{1}(\bm t^{(2)}+\bm c_1)&0& aS_{1}(t^{(2)}+\bm c_2)&0\\
S_{2}(\bm t^{(2)}+\bm c_1)&1& aS_{2}(t^{(2)}+\bm c_2)&a\\
S_{3}(\bm t^{(2)}+\bm c_1)&S_1(\bm t^{(2)}\bm c_1)&a S_{3}(\bm t^{(2)}+\bm c_2)&a S_1(\bm t^{(2)}+\bm c_2)\\
\end{matrix}
\right|.
\]
Now let $\ul m=(2,2)$, and define $\alpha_i ,\beta_i \in\mathbb{C}$ ($1\le i\le 3$), by
\[
\alpha_i=S_i(\bm c_1), \qquad\beta_i=S_i(\bm c_2),
\]
and
let $T(\bm t)=  \tau^{(2,2)}(\bm t^{(1)}=\bm t, \bm t^{(2)}=\bm 0)$, then
\[
T(\bm t)=
\left|
\begin{matrix}
S_{p+1}(\bm t)&S_1(\bm t)&S_{p}(\bm t)&1\\
S_{p+2}(\bm t)&S_2(\bm t)&S_{p+1}(\bm t)&S_1(\bm t)\\[2mm]
\alpha_2&1& a\beta_2&a\\
\alpha_3&\alpha_1&a \beta_3&a \beta_1\\
\end{matrix}
\right|
\] 
and
(cf. \eqref{eq:Q}): 
\[
	Q^+_{11}((2,2),\bm t,z)=\frac1{T(\bm t)}
\left|
\begin{matrix}
\sum_{k=0}^{p+1} S_{p+1-k}(\bm t)z^{-k}&S_1(\bm t)+\frac1{z}&\sum_{k=0}^{p} S_{p-k}(\bm t)z^{-k}&1\\[2mm]
\sum_{k=0}^{p+2} S_{p+2-k}(\bm t)z^{-k}&S_2(\bm t)+S_1(\bm t)z^{-1}+\frac1{z^2}&\sum_{k=0}^{p+1} S_{p+1-k}(\bm t)z^{-k}&S_1(\bm t)+\frac1{z}\\[2mm]
\alpha_2&1& a\beta_2&a\\
\alpha_3&\alpha_1&a \beta_3&a \beta_1\\
\end{matrix}
\right|,
\]
\[
	Q^-_{11}((2,2),\bm t,-z)=\frac1{T(\bm t)}
\left|
\begin{matrix}
S_{p+1}(\bm t)-S_{p}(\bm t)z^{-1}&S_1(\bm t)-\frac1{z}&S_{p}(\bm t)-S_{p-1}(\bm t)z^{-1}&1\\
S_{p+2}(\bm t)-S_{p+1}(\bm t)z^{-1}&S_2(\bm t)-S_1(\bm t)z^{-1}&S_{p+1}(\bm t)-S_{p}(\bm t)z^{-1}&S_1(\bm t)-\frac1{z}\\[2mm]
\alpha_2&1& a\beta_2&a\\
\alpha_3&\alpha_1&a \beta_3&a \beta_1\\
\end{matrix}
\right|,
\]
\[
	Q^+_{12}((2,2),\bm t,z)=-\frac{z^{-1}}{T(\bm t)}
\left|
\begin{matrix}
S_{p}(\bm t)&1&S_{p-1}(\bm t)&0\\
S_{p+1}(\bm t)&S_1(\bm t)&S_{p}(\bm t)&1\\
S_{p+2}(\bm t)&S_2(\bm t)&S_{p+1}(\bm t)&S_1(\bm t^)\\[2mm]
\alpha_3+\frac{\alpha_2}{z}+\frac{\alpha_1}{z^{2}}+\frac{1}{z^{3}}&
\alpha_1+\frac1{z}&
a\left(\beta_3+\frac{\beta_2}{z}+\frac{\beta_1}{z^{2}}+\frac{1}{z^{3}}\right)&
a\left(\beta_1+\frac1{z}\right)
\\
\end{matrix}
\right|,
\]
\[
Q^-_{12}((2,2),\bm t,-z)=-\frac{z^{-1}}{T(\bm t)}
\left|
\begin{matrix}
S_{p+2}(\bm t)&S_2(\bm t)&S_{p+1}(\bm t)&S_1(\bm t)\\[2mm]
\alpha_1-\frac{1}{z}&0&a\left(\beta_1-\frac1{z}\right)&0\\
\alpha_2-\frac{\alpha_1}{z}&1&a\left(\beta_2-\frac{\beta_1}{z}\right)&a\\
\alpha_3-\frac{\alpha_2}{z}&\alpha_1-\frac{1}{z}&a\left(\beta_3-\frac{\beta_2}{z}\right)&a\left(\beta_1-\frac1{z}\right)\\
\end{matrix}
\right|.
\]
Let $S_p(x,\bm t)$ and $Q_{ab}^\pm((2,2),x,\bm t,z)$ be defined as in \eqref{eq:Q2} by replacing $t_1$ by $t_1+x$. Then 
\[
\tilde{\mc L}((2,2), x,\bm t,\partial)^{\frac{j}{p}}=Q^+_{11}((2,2),x,\bm t,\partial)\circ \partial^j \circ Q^-_{11}((2,2),x,\bm t,\partial)^*
\]
and
\[
\tilde{\mc L}((2,2), x,\bm t,\partial)_-=q_1(x,\bm t)\partial^{-1}\circ r_1(x,\bm t)+q_2(x,\bm t)\partial^{-1}\circ r_2(x,\bm t),
\]
where
\[
\begin{split}
q_1(x,\bm t)&=Q^+_{12;1}((2,2),x,\bm t),\quad q_2(x,\bm t)=Q^+_{12;2}((2,2),x,\bm t),\\
r_1(x,\bm t)&=-Q^-_{12;2}((2,2),x,\bm t),\quad r_2(x,\bm t)=Q^-_{12;1}((2,2),x,\bm t).
\end{split}
\]
Explicitly:
\[
\begin{split}
q_1(x,\bm t)=&
-\frac{1}{T(x,\bm t)}((\alpha_1S_1(x,\bm t) - a \beta_1 S_2(x,\bm t)) S_p(x,\bm t)^2 + \alpha_3 S_{p+1}(x,\bm t) + a \beta_3 S_1(x,\bm t) S_{p+1}(x,\bm t) \\
&- 
 a \beta_1 S_{p+1}(x,\bm t)^2 - a \beta_3 S_{p+2}(x,\bm t) + 
S_p(x,\bm t) (-\alpha_3 S_1(x,\bm t) - a \beta_3 S_1(x,\bm t)^2 + 
    a \beta_3 S_2(x,\bm t) \\
&+ (-\alpha_1 + a \beta_1 S_1(x,\bm t)) S_{p+1}(x,\bm t) + a \beta_1 S_{p+2}(x,\bm t)) + 
 S_{p-1}(x,\bm t) (\alpha_3 (S_1(x,\bm t)^2 -S_2(x,\bm t)) \\
&+ (-\alpha_1 S_1(x,\bm t) + a \beta_1 S_2(x,\bm t)) S_{p+1}(x,\bm t) + (\alpha_1 - a \beta_1 S_1(x,\bm t)) S_{p+2}(x,\bm t))),\\
q_2(x,\bm t)=&
-\frac{1}{T(x,\bm t)}(S_1(x,\bm t)  - a S_2(x,\bm t) ) S_p(x,\bm t) ^2 + \alpha_2 S_{p+1}(x,\bm t)  + a \beta_2 S_1(x,\bm t)  S_{p+1}(x,\bm t)  - 
 a S_{p+1}(x,\bm t) ^2 \\
&- a \beta_2 S_{p+2}(x,\bm t) + 
 S_p(x,\bm t) (-\alpha_2 S_1(x,\bm t)  - a \beta_2 S_1(x,\bm t) ^2 + a \beta_2 S_2(x,\bm t) \\
& + (-1 + a S_1(x,\bm t) ) S_{p+1}(x,\bm t)  + 
    a S_{p+2}(x,\bm t) ) + 
 S_{p-1}(x,\bm t)  (\alpha_2 (S_1(x,\bm t)^2 - S_2(x,\bm t) ) +\\
& (-S_1(x,\bm t)  + a S_2(x,\bm t) ) S_{p+1}(x,\bm t) + (1 - a S_1(x,\bm t) ) S_{p+2}(x,\bm t) ),\\
r_1(x,\bm t)=&
\frac{a}{T(x,\bm t)} ((
\alpha_1 \alpha_2  - \alpha_3  - \alpha_1 \beta_2  + \beta_3 )S_1(x,\bm t) +a (\alpha_3  - 
   \alpha_2 \beta_1  +  \beta_1 \beta_2  -  \beta_3 )S_2(x,\bm t) \\
&+ (\alpha_1 - \beta_1) S_{p+1}(x,\bm t) + 
   a ( \beta_1-\alpha_1 ) S_{p+2}(x,\bm t)),\\
r_2(x,\bm t)=&
\frac{a}{T(x,\bm t)}(( \alpha_3 \beta_1  + \alpha_1^2 \beta_2  - \alpha_1 \beta_3 -\alpha_1 \alpha_2 \beta_1  )S_1(x,\bm t) 
- a
   ( \alpha_3 \beta_1  +  \alpha_2 \beta_1^2  -  \alpha_1 \beta_1 \beta_2  + \alpha_1 \beta_3) S_2(x,\bm t)\\
& + 
   \alpha_1 ( \beta_1-\alpha_1 ) S_{p+1}(x,\bm t)+ a (\alpha_1 - \beta_1) \beta_1 S_{p+2}(x,\bm t)).
\end{split}
\]
\end{example}
\subsection{
Soliton type tau-functions
}
\label{sec:soliton}
In 
\cite{KvdL03}, Section VI, soliton type tau-functions were constructed for the multicomponent KP hierarchy. They are not tau-functions of the $\ul p$-reduced multicomponent KP hierarchy
since the modes of the following expression
\[
G_{ij}(z,w)=(z-w)^{-\delta_{ij}}Q_iQ_j^{-1}z^{\alpha_0^{(i)}}w^{-\alpha_0^{(j)}}e_+^{(i)} (\ul{\bm t},z)e_+^{(j)} (\ul{\bm t},w)e_-^{(i)} (\frac{\partial}{\partial \ul{\bm t}},z)e_-^{(j)} (\frac{\partial}{\partial \ul{\bm t}},w),
\]
where $Q_i^{\pm 1}(\tau^{\ul m} (\ul{\bm t}))=\tau^{\ul m\pm \ul e_j} (\ul{\bm t})$,
do not lie in  the affine Lie algebra of type $A^{(1)}_{p_1+p_2+\cdots +p_r -1}$. One can fix this by replacing $z$ and $w$ by $\omega_i^kz$, respectively $\omega_j^{\ell} z$, where $\omega_a=
e^{\frac{2\pi i}{p_a}}$, $1\le k\le p_i$, $1\le \ell\le p_j$, and  $k\not =\ell$ when $i=j$.
The modes of 
$
G_{ij}(\omega_i^kz,\omega_j^{\ell} z)
$
are elements of this affine Lie algebra (see \cite{KvdL03}, Section VII), hence exponentials of such elements lie in the corresponding completed loop group. This produces $N$-solitary type solutions of the $\ul p$-reduced multicomponent KP hierarchy. See \cite{KvdL03}, Section VI and VII, for some more details. Introduce  triples $(h,i,j)$, where $1\le h\le N$,
$1\le  i,j\le r$. Let $s=(h,i,j)$ be such a triple and  $a_s=a_{hij}\in\mathbb{C}$. 
Let $\sigma$ be the number of inversions of a tuple of positive integers $(i_1, i_2, \ldots ,i_{2m})$, i.e.  $\sigma (i_1, i_2, \ldots, i_{2m})$  is the number of pairs $(i_p,i_q)$ such that $p<q$ but $i_p>i_q$.
Introduce the following constants:
\[
\begin{split}
	c(s_1,\ldots ,s_m)=&(-1)^{\sigma(i_1,j_1,i_2,\ldots, i_m,j_m)}\prod_{p=1}^m a_{s_p}
((\omega_{i_p}^{k_p}-\omega_{j_p}^{\ell_p})z_{h_p})^{-\delta_{i_p,j_p}}
\prod_{q=p+1}^m (\omega_{i_p}^{k_p}z_{k_p}-\omega_{i_q}^{k_q}z_{h_q})^{\delta_{i_p i_q}} \times\\
&\quad
(\omega_{i_p}^{k_p}z_{k_p}-\omega_{j_q}^{\ell_q}z_{h_q})^{-\delta_{i_pj_q}} 
(\omega_{i_p}^{k_p}z_{h_p}-\omega_{j_q}^{\ell_q}z_{h_q})^{\delta_{j_p j_q}}
(\omega_{i_p}^{k_p}z_{h_p}-\omega_{i_q}^{k_q}z_{h_q})^{-\delta_{j_pi_q}}
\end{split}
\]
Then the $N$-soliton type tau-functions $\vec{\tau}(\ul{\bm t})$ of charge $0$ are:
\[
\begin{split}
\{ 1\}^{\ul 0}+&\sum_{m=1}^{N}\sum_{1\le h_1<\cdots<h_m\le N}\sum_{1\le i_p,j_p\le r}\cdots \sum_{1\le i_m,j_m\le r}
\Biggl\{
c(s_1,\ldots,s_m)\Biggr.\times
\\
&\Biggl.
\exp\left(
\sum_{p=1}^m
\sum_{q=1}^\infty 
\left(
(\omega_{i_p}^q t_q^{(i_p)}-\omega_{j_p}^q t_q^{(j_p)})z_{h_p}^q 
\right)
\right)
\Biggr\}^{\ul e_{i_1}-\ul e_{j_2}+\ul e_{i_2}+\cdots -\ul e_{j_m}}\, .
\end{split}
\]
Here the notation $\{ \ \}^{\underline m}$ means that this expression is the
$\tau^{\underline m}$ part of the tau-function.
(A correction to \cite{KvdL03}, formula (219): $r$ should run from $1$ to $N$, and the triples
$s_a (p_a, i_a, j_a)$ should satisfy  the condition $a<a'$ implies $p_a< p_{a'}$.)


\end{document}